\documentclass[10pt,a4paper,reqno]{amsart}

\usepackage[T1]{fontenc}
\usepackage[utf8]{inputenc}
\usepackage[english]{babel}
\usepackage{amssymb,graphicx,xcolor}
\usepackage{amssymb,bbm,enumerate}
\usepackage{esint}
\usepackage{colortbl}
\usepackage{xcolor}
\usepackage{dsfont}
\usepackage{mathtools}
\usepackage{hyperref}
\usepackage{cleveref}
\usepackage{comment}
\usepackage{graphicx}
\graphicspath{ {./images/} }

\usepackage{amsmath,amssymb,amstext,amsthm}

\numberwithin{equation}{section}

\newtheorem{theorem}{Theorem}[section]

\newtheorem{lemma}[theorem]{Lemma}

\newtheorem{proposition}[theorem]{Proposition}
\theoremstyle{remark}
\newtheorem{remark}[theorem]{Remark}
\newtheorem{examples}{Examples}

\theoremstyle{definition}
\newtheorem{definition}[theorem]{Definition}

\newcommand{\Z}{\mathbb{Z}}

\newcommand{\R}{\mathbb{R}}
\newcommand{\C}{\mathbb{C}}
\newcommand{\N}{\mathbb{N}}

\DeclareMathOperator{\sech}{sech}
\DeclareMathOperator\supp{supp}

\newcommand{\comments}[1]{\vskip.3cm \fbox{\parbox{0.93\linewidth}{\footnotesize #1}} \vskip.3cm}

\title[MFD Summer School Course. The binormal flow]{The Binormal Flow: a toy model for turbulence}

\author{Elena Danesi}
\address{Politecnico di Torino, Torino.
		  e-mail: {\tt elena.danesi@polito.it}
		}
		
\author{Pablo Merino}
\address{Basque Center for Applied Mathematics, Bilbao.
		  e-mail: {\tt pmerino@bcamath.org}
		}

\author{Luis Vega}
\address{Basque Center for Applied Mathematics, Bilbao.
		  e-mail: {\tt lvega@bcamath.org}
		}

\thanks{}
\date{\today}

\begin{document}
\maketitle 

\begin{abstract}
	These notes summarize the content of a one-week course given by Luis Vega on the binormal flow model applied to describe the dynamics of three-dimensional vortex filaments in incompressible fluids following the Euler equations. This topic touches connections between the binormal flow, the Schrödinger map and the 1d cubic nonlinear Schrödinger equation, providing a rich casuistry linked to the physical and numerical phenomenology. This course was part of the $2^{\text{nd}}$ Advanced Summer School held at the Institut d'Etudes Scientifiques de Cargèse (IESC), at Cargèse (Corsica), from $20^{\text{th}}$ to $30^{\text{th}}$ June 2023. These notes have been written and completed by two of its attendees, Elena Danesi and Pablo Merino. 
\end{abstract}

\tableofcontents

\section{Introduction}

The aim of this course is to present a toy model for describing the evolution of vortex filaments in a incompressible fluid in $\R^3$ moving according to the Euler equations. When using the term vortex filament we mean that the fluid is characterized by the following property: the vorticity is given by a vector measure that has the support contained on a curve (the filament). In this model the filament is described by a curve $\chi(s,t) \in \R^3$ parametrized by its arc lenght $s$ that evolves in time $t$ according to the binormal flow
\[
\chi_t = \chi_s \wedge \chi_{ss}.
\]
This model appears for the first time in 1906 in the work \cite{DaR06} of Da Rios, and was rederived in 1965 by Arms and Hama \cite{AH65} as an approximation of the dynamics of a vortex filament under the Euler equations. It is a purely ideal model. 
It aims to answer to the necessity of having a proper toy model for vortex filaments in fluid dynamics and, from there and in a potential way, for turbulence theory. The study of this particular model is highly motivated for its accordance with real experiments and numerical simulations, as the reader will find out reading these notes. 
One can ask if it is true that it is possible to prove the existence of solutions of the Euler equations whose vorticity remains concentrated along a curve (filament) under the evolution of the binormal flow for long time (see \cite{JS17}). This problem is known as \emph{vortex filament conjecture} and it is still open in the general case. We collect some results in this direction:
\begin{enumerate}
    \item Benedetto-Caglioti-Marchioro (1999) \cite{M99}: Vorticity is supported in an annulus, in non-viscous fluid with axial symmetry without swirl.
    \item Aftalion-Jerrard (2002) \cite{AJ02}: Vortex rings in Bose-Einstein condensates. If initially the current/Jacobian $Ju_{\varepsilon} = \frac{1}{2} \text{curl}(u_{\varepsilon} \wedge \nabla u_{\varepsilon})$ (see notation of the aforementioned paper) concentrates in the limit to an $n-2$ dimensional set, then this happens for latter times, and the limit set evolves by binormal flow.
    \item Dàvila-Del Pino-Musso-Wei (2022) \cite{DPMW22}:  Euler equations for particular cases with reductions, given by helices and leap-frogging in axysymmetry without swirl. 
    \item Bedrossian-Germain-Harrop-Griffiths (2018) \cite{BGHG18}: First step in Navier-Stokes. On small times one vortex filament evolves as Lamb-Ossen vortex around the initial curve. This result goes beyond the now classical result constructed by Giga and Miyakawa in \cite{GM89}.  
    \item Jerrard-Smets (2020) \cite{JS21}: In supefluids one almost-parallel to the $z$-axis quantized vortex filament in Bose condensate in the incompressible limit evolves by a linearized version of the binormal flow. 
    \item Fontelos-Vega (2023) \cite{FV23}: They consider Giga-Miyakawa solutions that at time zero are regular closed curves without self-intersections \cite{GM89}. More concretely, they describe the behaviour of the solution by means of an asymptotic expansion given in terms of negative powers of $\sqrt{\nu t}$, with $\nu$ the viscosity and $t$ the time. The leading order behaves as a Lamb-Ossen vortex adapted to a curve that moves according to an appropriately normalized binormal flow.
    
    \end{enumerate}
We underline that vortex filaments with the shapes of straight lines, circles, and helices do exist, both in experiments and as solutions of Euler equations.
A remarkable disadvantage is that this does not exclude dynamics that are not present in fluids (e.g. closed curves preserve length while in fluids one has filamentation phenomena).
\\
In these lectures we will mainly focus on singular filaments, that is curves $\chi$ that are smooth everywhere but in some points where they present a singularity in a shape of a corner. We will first discuss the case of one corner, considering self-similar solutions of the binormal flow.
\begin{figure}[h]
    \centering
\includegraphics[width=1cm, height=6cm, width=\textwidth]{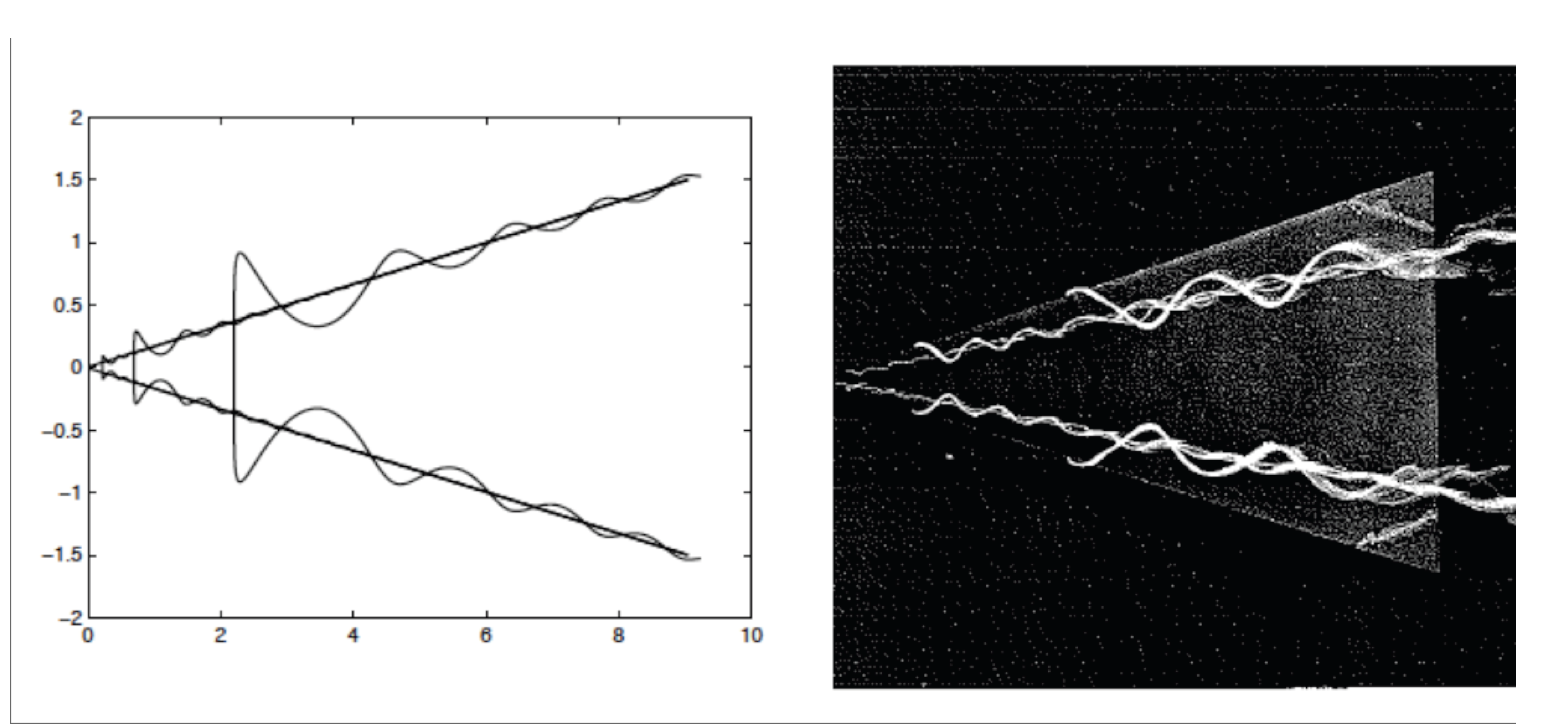}
\caption{``Delta wing'' vortex formed by fluid flowing over a triangular obstacle (see \cite{HGV09} for details).}
\label{fig:vortex_fil}
\end{figure}
\begin{figure}[h]
    \centering
\includegraphics[width=1cm, height=6cm, width=\textwidth]{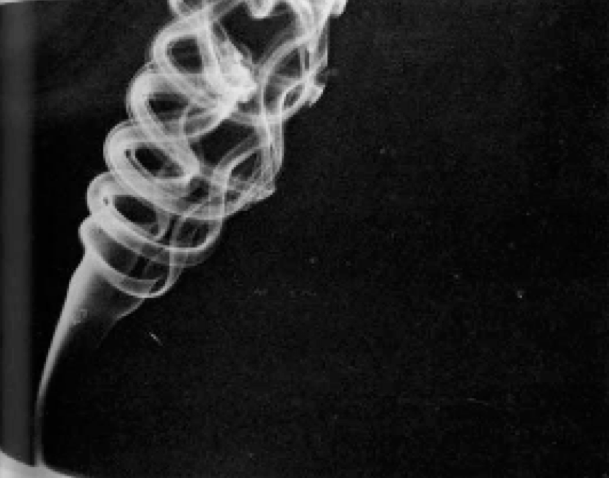}
\caption{Smoke ring. What we see in this picture are `vortex tubes' that propagate in a self-similar way. These vortex tubes have two distinguished parts: a first one which looks very much as a horseshoe, and a second one with a helicoidal shape.}
\label{fig:vortex_smoke}
\end{figure}
If, for example, we look at \Cref{fig:vortex_fil,fig:vortex_smoke} we see that in real experiments the vortex appears to be made by some self-similar curves, is the sense that the same picture seems to appear at different scales. Moreover, if we try to characterize them from a geometric point of view, we see that the shape seems to be some kind of circle close to the vertex of the corner, and helices with the axis from the corner itself. That is why we focus on the self-similar solutions, describing their geometric structure. They are, indeed, analytical solutions of the binormal flow for positive times, $V$-shaped at the initial time. Moreover, these curves are stable: if one considers a filament `close' to the self-similar solution, the corner is still formed. One of the key tools to prove these results is to exploit a connection between the binormal flow and other PDEs. Indeed, as will be explained in the following sections, through the Hasimoto transformation it is possible to derive a 1d focusing cubic nonlinear Schrödinger equation (NLS). In this way, the stability issue for the self-similar solutions will be translated into a scattering problem for solutions to the NLS. However, it is important to underline that the solutions to the NLS that will be considered live in a low-regularity regime. This will require the introduction of new tools which do not rely on the standard conservation laws. 
Furthermore, the model connects with another classical equation, in relation with the work in ferromagnetism theory of Landau-Lifshitz \cite{LL35}. 
Indeed, if $\chi(s,t)$ is a solution of BF, then its tangent vector $T(s,t)$ satisfies 
\[
T_t = T \wedge T_{ss},
\]
which is the 1d Schrödinger map on $\mathbb S^2$. The interplay of different fields of mathematics and physics is even more visible when studying the trajectories of polygonal vortices, to which the last part of the course was dedicated. From our point of view, this makes the BF very interesting also from a mathematical side. 
We mention, as an example, how ideas from harmonic analysis and analytical number theory comes into play in the study of the evolution of regular polygons. Let $\phi_{x_0}$ be a slightly different version to the Riemann's `non-differentiable' function, given by
\[
\phi_{x_0} \colon [0,1] \rightarrow \C, 
\quad \phi_{x_0}(t) = \sum_{n \ne 0} \frac{e^{2\pi i (n^2 t + nx_0)}}{n^2},
\]
where $x_0 \in \R$ fixed. 
\begin{figure}[h]
\centering
\includegraphics[scale=0.5, width=\textwidth]{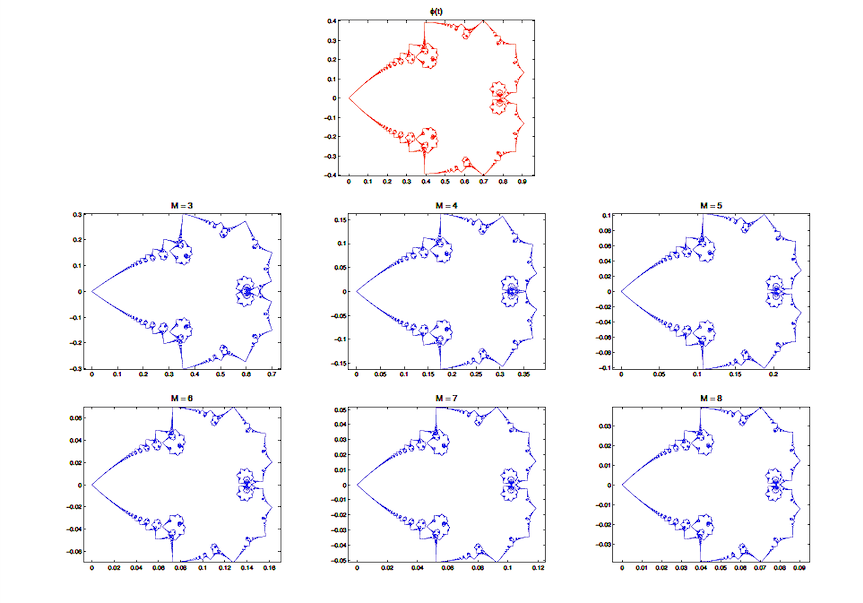}
\caption{Comparison between the evolution of one of the corners of $\chi(0,t)$ for several numbers of sides (the plots in blue, which in order correspond to $M=3$, $M=4$, $M=5$, $M=6$, $M=7$ and $M=8$), and a plot in red of Riemann's non-differentiable function $\phi(t)$.}
\label{fig:RiemannInkPoint}
\end{figure}
We present in \Cref{fig:RiemannInkPoint} the graph made by the evolutions of the previous version of Riemann's function (for $x_0 = 0$), and of one corner of regular polygons with variable $M$ sides. The similarity of the graphs has a remarkable interest. In an article by De la Hoz and Vega \cite{DLHV14}, inspired by Jerrard and Smet \cite{JS}, they discovered numerically that a mild variation of $\phi_{0}$,
\[
\phi(t) = \sum_{n \in \Z} \frac{e^{2\pi i n^2 t} - 1}{n^2} = 2\pi i t - \frac{\pi^2}3 + \phi_0(t)
\]
very closely represents the trajectory of the corners of such vortices. More precisely, let the filament be a curve $\chi$ that evolves according to the binormal flow equation. 
Suppose that the initial filament $\chi(s, 0)$ is a regular polygon with corners at the integers $s \in \Z$. De la Hoz-Vega  observed that $\chi(0,t)$ is a plane curve which, after identifying the plane with $\C$, behaves like $\phi(t)$. Later, Banica and Vega \cite{BV22} rigorously proved this under certain hypotheses. What is more, if $M \in \N$ is the number of sides of the initial polygon, and if $\chi_M$ is the corresponding filament, they proved that the rescaled trajectory $M\chi_M$ tends to the plane curve
\[
\lim_{M \rightarrow + \infty} M \chi_M (x_0,t) = \phi_{x_0}(t).
\]
In this way, the function $\phi_{x_0}$ can be seen as a representative of the trajectory of polygonal vortex filaments. In a very recent article \cite{BENV23}, this function has been studied. The authors proved that when $x_0$ is a rational it satisfies the Frish-Parisi multifractal formalism studied in the theory of turbulence. When $x_0$ is irrational, it has a multifractal behaviour. Fractal behaviour already appears at the linear level of the 1d NLS
\[
\psi_t = \frac{i}{4\pi}\psi_{ss}
\]
with an initial datum that is piecewise constant 
in, for instance, $s \in [-\pi,\pi]$, and then extended by periodicity. In \cite{Be96} and \cite{BK96} this framework is used to model the Talbot effect and show that at rational times the solution can be obtained as a finite overlapping of translates of the initial datum, while at irrational times the images have a fractal profile. Moreover, in \cite{ET13} Erdogan and Tzirakis provide results in this direction for the nonlinear setting in the cubic NLS, with initial data given by piecewise continuous functions. Their contribution uses the regularizing effect related to the Duhamel's integral, which implies that the fractal behaviour in their context is due to the linear part of the solution. Following these ideas, we also approach such issues for the case of the aforementioned filaments with corners, expressed through sums of Dirac distributions.
\\
Starting from the aforementioned one corner results, progress has also been done with general initial data with several corners and, in particular, polygons. In the recent papers of Jerrard and Smets \cite{JS13} and \cite{JS}, it is given a global existence theorem that allows to consider such types of initial data. They also simulate numerically the evolution of the unit square at different times. Indeed, at a certain time, the square seems to reappear, but with the axes turned by a $\pi/4$-angle with respect to the initial ones. We go beyond: this \textit{axis switching} phenomena is not restricted to the square's case, but given a regular polygon with $M$ sides as the initial datum, another polygon is obtained at times of the form $t_{p,q} = (2\pi/M^2)(p/q)$, for $p$ and $q$ coprimes. In particular, and except for $q=2$, skew polygons with $Mq$ sides (if $q$ is odd), or $Mq/2$ sides (if $q$ is even), are formed. We refer to \Cref{sec:Evolution of $M$-regular polygonal lines through the BF} for a precise description and for real experiments on this phenomenon.
\\
To conclude, one could wonder whether or not there is any connection between all the suggested solutions and real fluids. The results obtained in \cite{FV23} show that Da Rios' and Arms and Hama's approximations are correct for Giga-Miyakawa solutions that at time zero are regular closed curves without self-intersections \cite{GM89}. It is an open question if this result can be extended to  more singular solutions to allow for the existence of corners in the initial vortex filament. At this respect we recall the aforementioned \Cref{fig:vortex_fil,fig:vortex_smoke}. Also, in \cite{HGV09}, a similitude between self-similar solutions and a flow traversing a triangular wing is presented. Moreover, with respect to the axis switching phenomena, and apart from the domestic experiment presented in the last section, it appears that this is something quite well documented for noncircular jets, for nozzles which could be an ellipse, a rectangle, a square, or an equilateral triangle, among others. See the survey \cite{GG99}, where some papers related to these observations can be found. Also, the appearance of smooth polygonal structures with more angles has been appreciated for nozzles with the shape of an equilateral triangle (see \cite{GG95} and \cite{GG99}).

\subsection{Structure of the lecture notes}
In \Cref{sec:Construction of selfsimilar solutions of the binormal flow} we firstly describe the binormal flow equation and the connection with the 1d cubic NLS, through the Hasimoto transformation and the Schrödinger map, discussing some explicit examples. Then, in the second subsection, we focus on the self-similar solutions, describing it from a geometrical point of view (the main result will be \Cref{thm:asymp_onecorner}). To conclude, we present an introduction to 1d cubic NLS, for the unfamiliar reader. In particular, we describe some useful conservation laws and symmetry properties of the equation, and prove a global well-posedness result  in $L^2(\R)$.
In \Cref{sec:Stability of one corner} we discuss the stability of the corner, emphasizing the study for the existence of the wave operator, and for the asymptotic completeness. In \Cref{sec:Several Corners}, we deal with the case of several corners, polygonal lines, within an analytical-geometrical framework, see \Cref{thm:SolsNLSSeveralCorners}. This result will be obtained by translating the corresponding result for several Dirac masses for the 1d cubic NLS (see \Cref{thm:SeveralDiracMassesNLS}) to the binormal flow framework.
To conclude, we discuss the Talbot effect and the evolution of an $M$-side regular polygon, mentioned above.\\

As we mentioned before, these notes summarize the content of a course imparted by Luis Vega during the $2^{\text{nd}}$ Advanced Summer School in Mathematical Fluid Dynamics at the IESC (Cargèse, Corsica) in 2023. They have been written and completed by two of its attendees, Elena Danesi and Pablo Merino.

 \subsection{Notations} 
Let $u = u(s,t) \colon \R \times \R \rightarrow \R$. We denote by both $\partial_s u$ and $u_s$ the partial derivative of $u$ with respect to $s$. The same for $t$. \\
For $p,q \in [1 , +\infty]$ we denote $L^p_tL^q_s$ the space $L^p((\R),L^q(\R))$ with the norm
\[
\lVert u \rVert_{L^p_tL^q_s} = \bigg ( \int_{- \infty}^{+ \infty} \bigg ( \int_{-\infty}^{+\infty} \lvert u(s,t) \rvert^q ds \bigg )^{\frac pq} dt \bigg)^{\frac 1p}. 
\]
For $s \in \R$ we denote as $H^s(\R)$ the space $\{ \phi \in \mathcal S' \colon \langle D \rangle^s \phi \in L^2(\R)\}$ with the norm
\[
\lVert \phi \rVert_{H^s} = \big \lVert \langle D \rangle^s \phi \big \rVert_{L^2},
\]
where $\langle D \rangle^s \phi \coloneqq \mathcal F^{-1} \big (\langle \xi \rangle^s \mathcal F \phi \big ) $, $\mathcal F \phi$ denoting the Fourier transform of $\phi$.

\subsection{Acknowledgements}
The authors want to thank Valeria Banica, Emmanuel Dormy and Christophe Lacave for organizing the $2^{\text{nd}}$ Advanced Summer School on Mathematical Fluid Dynamics at the IESC, where interesting courses, talks and discussions took place.

 \section{Construction of selfsimilar solutions of the binormal flow}\label{sec:Construction of selfsimilar solutions of the binormal flow}
    \subsection{The Binormal Flow}\noindent\label{sct:sectionBF}

We are interested in the evolution of a vortex filament in 3d. Therefore we consider a fluid with time-dependent velocity field $u(x,t)$ where $x \in \R^3$ and $t \in \R$. The field is moving according to Euler equations such that its velocity filed is irrotational except in a curve $\Sigma$. This curve is parametrized by arc length $s \in \R$ as
\[
(s,t) \mapsto \chi(s,t) \in \R^3.
\]
The velocity field in a point $P$ outside the curve is then obtained using the Biot-Savart integral
\begin{equation}\label{eq:BSintegral}
u(P) = \frac{\Gamma}{4\pi} \int_{-\infty}^\infty \frac{P - \chi(s,t)}{\lvert P - \chi(s,t) \rvert^3} \wedge T(s,t) ds
\end{equation}
where $T$ is the tangent vector of the curve given by
\[
T(s,t) = \chi_s (s,t)
\]
and $\Gamma$ denotes the strength of the filament.

\begin{examples}\label{ex:vortex}\noindent
Typical examples of vortex filaments are:
\begin{enumerate}[i)]
\item the \emph{straight line} vortex: $\chi(s,t) = \chi(s) = (0,0,s)$. 
The velocity can be explicitly computed by \eqref{eq:BSintegral} getting
\[
u(t,(x,y,0)) = \frac{\Gamma}{2\pi} \frac{(-y, x, 0)}{x^2 + y^2};
\]
We briefly illustrate such computations. Denote $P=(x,y,0)$ the coordinates of the point outside the curve. Moreover, notice that 
\[
(\chi(s,t) - P) \wedge T(s,t) = (-y,x,0).
\]
Then, 
\[
\left( \int_{\mathbb{R}} \frac{ds}{(x^2 + y^2 + s^2)^{3/2}} \right)(-y,x,0).
\]
Thus, given $\alpha \in \mathbb{R}$, the difficulty is in solving the integral
\begin{align*}
    I = \int_{\mathbb{R}} \frac{ds}{ (\alpha^2 + s^2)^{3/2}}.
\end{align*}
However, to solve it, it suffices to substitute  in the integral $s = \sqrt{x_1^2+x_2^2} \tan u$, $ds= \sqrt{x_1^2 + x_2^2} \sec^2 u du$.
\item the \emph{circle}: $\chi(s,0) = (\cos s, \sin s, 0)$ (see \cite{Saf92});
\item the \emph{helical} vortex: $\chi(s,0) = \frac 1{\sqrt {h^2 +1}} (\cos s, \sin s, hs) $, $h \ne 0$. Observe that this helix degenerates into a the circle when $h \rightarrow 0$ and to a straight line when $\lvert h \rvert \rightarrow + \infty.$ The existence of helical vortices was experimentally proved in  \cite{HB81}, \cite{MHR85} and more recently in \cite{AKO07}.


\end{enumerate}
\end{examples}

In order to find the evolution of the filament for later times one has to compute the velocity field for points $P\in \Sigma$. As we see from \eqref{eq:BSintegral} this can not be easily done due to the singularity of the Biot-Savart integral. Different methods of desingularizing this integral can be found in the literature. In these lectures we shall follow the so called Localized Induction Approximation (LIA), which only takes into consideration the local effects of the corresponding Biot-Savart integral. It was first proposed by Da Rios in 1906 \cite{DaR06} (see also \cite{Ric96}). We refer the reader to the classical reference of Batchelor \cite{B00} and the more recent ones of Saffman \cite{Saf92} and Alekseenko-Kuibin-Okulov \cite{AKO07} for a detailed description of this approximation and of its limitations. All these questions have been recently clarified in \cite{FV23}.\\
The consequence of this approximation is that the filament moves, after a logarithmic renormalization in time,  in the direction of the binormal with a velocity that is proportional to the curvature. That is to say, the curve evolves according to the \emph{Binormal Flow} (BF) 
\begin{equation}\label{eq:CPchi}
\begin{cases}
\chi_t  = \chi_s \wedge \chi_{ss}\\
\chi(s,0) = \chi_0(s).
\end{cases}
\end{equation}
We can rewrite the equation in \eqref{eq:CPchi} using the Frenet equations and calling $(T,n,b)$ the Frenet frame, where $n$ is the normal vector to the curve and $b$ the binormal one. Then we have for $(c,\tau)$, respectively the curvature and the torsion of $\chi$, that 
\begin{equation}\label{eq:frenet_eqs}
\begin{split}
T_s &= \qquad  cn\\
n_s &= -cT  \qquad + \tau b\\
b_s &= \qquad -\tau n \quad.
\end{split}
\end{equation}
Hence \eqref{eq:CPchi} becomes
\begin{equation}\label{eq:CPchi_frenetframe}
\begin{cases}
\chi_t = cb\\
\chi(0,s) = \chi_0(s).
\end{cases}
\end{equation}
A first consequence on \eqref{eq:CPchi} is that $\lvert T(s,t) \rvert=1$. Indeed, after differentiation in \eqref{eq:CPchi} we get an evolution equation for the tangent vector, the Landau-Lipschitz equation 
\begin{equation}\label{eq:CPtangvect}
\begin{cases}
T_t = T \wedge T_{ss}\\
T(s,0)=T_0(s).
\end{cases}
\end{equation}
Then, we can compute
\[
\frac 12  \frac{d}{dt}\,  \lvert T \rvert^2 = T \cdot T_t = T \cdot (T \wedge T_{ss}) = 0.
\]
Moreover, we notice that $T \wedge T_{ss} = J D_s T_s$ with $J$ the complex structure in the unit sphere $\mathbb{S}^2$ and $D_s$ the covariant derivative. As a consequence, \eqref{eq:CPtangvect} is also known as the \emph{Schrödinger map} (SM) equation with target $\mathbb{S}^2$ (see \cite{NSVZ} for more details). This equation can be easily generalized by changing both the domain of definition and the target. In these lectures we will only be interested in the second possibility, changing $\mathbb{S}^2$ into the Hyperbolic plane $\mathbb H^2$. In the latter case  \eqref{eq:CPchi} becomes
\begin{equation}\label{eq:chi_hypplane}
\chi_t = \chi_s \wedge_{-} \chi_{ss},
\end{equation}
and, setting again $T = \chi_s$, \eqref{eq:CPtangvect} has to be changed into
\begin{equation}\label{eq:T_thypplane}
T_t = T \wedge_{-} T_{ss},
\end{equation}
where for any $a,b \in \R^3$ the product above is defined as
\[
a \wedge_{-} b \coloneqq a \wedge \begin{pmatrix} 1 & 0 & 0 \\ 0 & 1 & 0 \\ 0 & 0 & -1 \end{pmatrix} b.
\]
Coming back to the BF and the \Cref{ex:vortex}, we can easily check that among the solutions we find, e.g.,  the straight line that remains stationary and the circle that moves in the binormal direction with velocity equals to the inverse of its radius.\\\\
Let us finish this subsection giving a fundamental transformation that links the BF \eqref{eq:CPchi} (and its modification \eqref{eq:chi_hypplane}) to the 1d focusing (defocusing in the case of $\mathbb H^2$) cubic nonlinear Schrödinger equation. In \cite{Has72} the ``filament'' function, called \emph{Hasimoto transformation}
\begin{equation}
    \label{Has_trans}
\psi(s,t) = c(s,t) e^{i \int_0^s \tau(s', t)ds'}
\end{equation}
is introduced. The author proved that if $(c,\tau)$ are, respectively, the curvature and the torsion of the filament $\chi(s,t)$ which evolves according to \eqref{eq:CPchi_frenetframe}, then $\psi$ is a solution to the \emph{1d cubic NLS}
\begin{equation}\label{eq:cubicNLSpsi}
\begin{cases}
\partial_t \psi = i \big ( \partial_{ss}\psi + \frac12 (\lvert \psi \rvert^2 + A(t)) \psi \big )\\
\psi(s,0)= \psi_0 (s).
\end{cases}
\end{equation}
Here $A(t)$ is a real function that is determined from $\chi(0,t)$. It can be computed explicitly:
\begin{align}
    A(t) = \left( 2 \frac{\partial_{ss}c - c\tau^2}{c} + c^2 \right)(0,t).
    \label{NLS_At_nomodification}
\end{align}

It is interesting to observe that in the hyperbolic case a generalized notion of curvature and torsion can be defined. As a consequence, a similar equation as \eqref{eq:cubicNLSpsi} is obtained with a defocusing nonlinear potential, i.e.\ the Hasimoto transform solves
\[
\partial_t \psi = i \big ( \partial_{ss}\psi - \frac12 (\lvert \psi \rvert^2 + A(t)) \psi \big ).
\]

As we see Hasimoto transformation involves the Frenet frame. The calculations to obtain \eqref{eq:cubicNLSpsi} are slightly more complicated in this way. It is simpler using a parallel frame. The argument, by \cite{Koi96}, goes as follows. Consider the orthonormal frame $(T, e_1, e_2)$ and the equations 
\[
\begin{split}
T_s = \qquad & \alpha e_1 + \beta e_2\\
e_{1s} = - \alpha T&\\
e_{2s} = - \beta T&.
\end{split}
\]
Then, from \eqref{eq:CPtangvect}
\[
T_t = T \wedge (\alpha e_1 + \beta e_2)_s = \alpha_s e_2 - \beta_s e_1.
\]
As a consequence, we have on one hand that
\begin{equation}\label{eq:T_st}
T_{st} = \alpha_t e_1 + \beta_t e_2 + \alpha e_{1t} + \beta e_{2t},
\end{equation}
and on the other hand
\begin{equation}\label{eq:T_ts}
T_{ts}= \alpha_{ss} e_2 - \beta_{ss} e_1 -\alpha_s \beta T + \beta_s \alpha T.
\end{equation}
The key idea is to look at $\langle e_{1t}, e_2 \rangle (= - \langle e_{2t}, e_1 \rangle)$. Following Hasimoto we compute its $s-$derivative
\begin{equation}\label{eq:e_1t,e_2}
\begin{split}
\langle e_{1t},e_2 \rangle_s & = \langle e_{1ts},e_2 \rangle + \langle e_{1t},e_{2s} \rangle\\
&= - \langle \alpha_t T,e_2\rangle - \langle \alpha T_t, e_2 \rangle - \beta \langle e_{1t}, T \rangle\\
&=- \alpha \langle T_t, e_2 \rangle + \beta \langle e_1, T_t \rangle \\
& = -\alpha \alpha_s - \beta \beta_s = - \frac12 (\alpha^2 + \beta^2)_s.
\end{split}
\end{equation}
The last step is to define the complex function
\begin{align}
\psi(s,t) = \alpha(s,t) +i\beta(s,t).
\label{eq:Hasimoto_AlphaBeta}
\end{align}
Then, from \eqref{eq:T_st}, \eqref{eq:T_ts}
\[
\begin{split}
&\alpha_t + \beta \langle e_{2t}, e_1 \rangle = - \beta_{ss}\\
&\beta_t + \alpha \langle e_{1t}, e_2 \rangle = \alpha_{ss}
\end{split}
\]
and from \eqref{eq:e_1t,e_2}
\[
\begin{split}
&\alpha_t + \frac12 \beta [(\alpha^2 + \beta^2) + A(t) ] = - \beta_{ss}\\
& \beta_t - \frac12 \alpha [(\alpha^2 + \beta^2) + A(t)] = \alpha_{ss},
\end{split}
\]
with $A(t)$ a real function. Finally we get 
\[
\psi_t = (\alpha + i\beta)_t = i \big \{ (\alpha_{ss} + i \beta_{ss} ) + \frac12 (\alpha + i \beta) (\alpha^2 + \beta^2 + A(t)) \big \}
\]
as desired. A similar computation can be done in the case of \eqref{eq:T_thypplane}, see e.g.\ \cite{dlH07}.
\\

Before going to the next subsection we observe that, with a change of phase, we can pass from solution to the focusing or defocusing NLS 
\begin{equation}\label{focdefocNLS}
\partial_t \psi = i (\partial_{ss} \psi \pm \frac12 (\lvert \psi \rvert^2 + A(t)) \psi
\end{equation}
to solution of the standard cubic NLS, that is with $A(t)=0.$ Indeed, let $\psi$ be a solution of \eqref{focdefocNLS}, we can define
\[
h(s,t) = \psi(s,t) e^{\mp \frac i2 \int_0^t A(t')dt'}.
\]
Then $h$ satisfies 
\[
i \partial_th + \partial_{ss} h \pm \frac 12 \lvert h \rvert ^2 h = e^{\mp\frac i2 \int_0^t A(t')dt'} \big( i \partial_t \psi + \partial_{ss} \psi \pm \frac12 (\lvert \psi \rvert^2 + A(t) ) \psi \big) = 0.
\]
Moreover, with a scaling we can get rid of the factor $\frac12$ in the previous equation. Indeed, we can define $\tilde h(s,t) =  h(s/\sqrt 2,t/2)$. Then, if $h$ is a solution of the previous cubic NLS, then $\tilde h$ is a solution of
\[
i \partial_t \tilde h + \partial_{ss} h \pm \lvert \tilde h \rvert^2 \tilde h = 0.
\]

 \subsection{Selfsimilar solutions}
 In this section we will find the selfsimilar solutions to the BF \eqref{eq:CPchi}. We notice that the flow is invariant by the scaling $\lambda^\alpha \chi(\lambda s, \lambda^{\alpha +3} t)$, for any $\lambda >0$ and $\alpha \in \R$. However, the only one which preserves the arc length is with $\alpha = -1$. Hence we look for solutions that satisfy
 \[
 \frac 1\lambda \chi(\lambda s, \lambda^2 t) = \chi(s,t) \quad \forall \, \lambda > 0.
 \]
 As a consequence $\chi$ can be written as
 \begin{equation}\label{eq:G}
 \chi(s,t) = \sqrt{t} G\big( s / \sqrt{t}), \quad t>0
 \end{equation}
 for some function $G(s)$. Plugging \eqref{eq:G} into \eqref{eq:CPchi} we obtain that the profile curve $G$ has to be a solution to 
 \begin{equation}\label{eq:pde_G_s}
 \frac 12 G - \frac s2 G_s = G_s \wedge G_{ss}.
 \end{equation}
 Calling again $T(s) =G_s(s)$ we obtain, after differentiation
 \[
 - \frac s2 T_s = T \wedge T_{ss}.
 \]
 From Frenet equations, as before, we get 
 \[
 \begin{split}
 - \frac s2 cn & = T \wedge (c_s n + c n_s) \\
 & = T \wedge (c_s n - c^2 T + c\tau b) = c_s b - c \tau n.
 \end{split}
 \]
 This equation determines the curvature and torsion of $G$:
 \[
 \begin{cases}
 c(s) = a\\
 \tau(s) = \frac s2
 \end{cases}
 \]
 for some $a \ge 0$. Therefore if $(T_a, n_a, b_a)$ is the unique solution of Frenet equations \eqref{eq:frenet_eqs} with $(c, \tau) = (a,s/2)$ and initial data $(T_a(0), n_a(0), b_a(0)) = \mathbbm 1_{3x3}$, then $G_a$ is uniquely determined by
 \begin{equation}\label{eq:ode_G_a}
 \begin{cases}
 (G_a)_s = T_a\\
 G_a(0)=2a b_a(0).
 \end{cases}
 \end{equation}
 Indeed, from \eqref{eq:pde_G_s} we have 
 \[
 \frac 12 G_a(0) = (G_a)_s(0) \wedge (G_a)_{ss}(0) =  T_a(0) \wedge a n_a(0) = a b_a(0).
 \]
 \begin{figure}[t]
    \centering
\includegraphics[width=1cm, height=6cm, width=\textwidth]{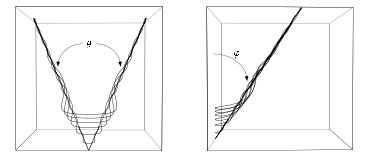}
\caption{Structure of the self-similar solution, with a corner's angle $\theta$. From our point of view, it is quite striking the similarity of the picture - $\chi_a(s,t)$ for different times - and the one obtained by the experiments of airflow through a delta wing (see \cite{HGV09})}
\label{fig:chia}
\end{figure}
We want now to describe the geometry of this profile; in particular, our main goals will be to prove that this curve resembles a ``V''- shaped curve as $t \rightarrow 0$ and as $s \rightarrow \infty$. That is, the shape of $\chi_a$ is as shown in \Cref{fig:chia}.\\
We start by observing that $G_a(s)$ approaches a straight line as $s$ tends to $+\infty$ or $-\infty$; indeed, from \eqref{eq:pde_G_s} we have 
 \[
 \bigg ( \frac {G_a}s \bigg )_s = \frac{s (G_a)_s - G_a}{s^2} = \frac{-2ab_a}{s^2}.
 \]
 Therefore $ \big ( \frac {G_a}s \big )_s$ is integrable in $\pm \infty$, recalling that $\lvert b_a(s) \rvert =1$. We can define 
 \begin{equation}\label{eq:defA_a^pm}
 A_a^\pm = \lim_{s \to \pm \infty} \frac{G_a(s)}s, \quad \lvert A^\pm_a \rvert =1,
 \end{equation}
 and
 \[
 \begin{split}
 G_a(s) & = s A_a^+ + 2as \int_s^{\infty} \frac{b_a(s')} {s'^2} d s'\quad \text{if } s>0 \\
 G_a(s) & = s A^-_a - 2as \int_{-\infty}^s \frac{b_a(s')}{s'^2} ds' \quad \text{if } s<0.
 \end{split}
 \]
 From these identities and \eqref{eq:ode_G_a} it easily follows that (see \cite{GRV03}) 
 \[
 \big \lvert \chi_a (s,t) - sA_a^+ \mathbb I_{[0,+\infty)} (s) - s A_a^- \mathbb I_{(-\infty,0]}(s) \big \rvert \le 2a\sqrt t,\quad t >0
 \]
 with
 \[
 \chi_a(s,t) = \sqrt t G_a(s/\sqrt t)
 \]
 and $\mathbb I$ denoting the characteristic function. Indeed, recall that $b_a$ is a unitary vector, so we have for $s>0$
 \[
\big \lvert  \sqrt t G_a(s/ \sqrt t) - s A_a^+ \big \rvert = \Big \lvert 2as \sqrt t \int_s^{+\infty} \frac{b_a (\sqrt t s')}{s'^2} ds' \Big \lvert \le 2 a \sqrt t.
 \]
 Hence we have found a solution to \eqref{eq:CPchi} for $t>0$ with
 \[
 \chi_a(s,0) = \begin{cases} A_a^+ s\quad s >0\\
 A_a^- s \quad s<0.
 \end{cases}
 \]
 Notice that the BF \eqref{eq:CPchi} is reversible in time: if $\chi(s,t)$ is a solution so is $\chi(-s,-t)$. Hence, if we give as initial condition at time $t=1$
 \[
 \chi_a(s,1)=G_a(s)
 \]
 and we go backwards in time we obtain that a singularity in the shape of a corner is developed at time $t=0$, if $A_a^+ \ne A_a^-$.
 Moreover, we observe that $G_a(s)$ is analytic because it is a solution to Frenet equations with $(c,\tau)= \big (a,  \frac s2 \big)$.\\

 It remains to prove that $\chi_a$ is not a straight line if $a>0$, that is the shape of $\chi_a$ is as in figure \eqref{fig:chia}.  
 
 \begin{proposition}\label{prop:A+nepmA-}\noindent
 \begin{enumerate}[i)]
 \item $A^+_a \ne -A^-_a$
 \item $A_a^+ = A^-_a \quad \text{iff} \quad  a=0$. 
 \end{enumerate}
 \end{proposition}
 
 \begin{proof}
 In order to lighten the notation in the following we will denote  $(G_a)_s = G_a'$ and the same for higher order derivatives. We recall that, from \eqref{eq:pde_G_s} and \eqref{eq:ode_G_a}, $G_a - sG_a' = 2a b_a$ and $G'''_a = T_a' = an_a$. Hence, $G_a$ solves
 \begin{equation}\label{eq:ode_G'''}
 G_a''' + \Big (a^2 + \frac{s^2}4 \Big ) G_a' - \frac{s}4 G_a =0
 \end{equation}
 with initial condition
 \[
 \begin{cases}
 G_a(0) = 2a(0,0,1)\\
 G'_a(0) = (1,0,0)\\
 G''_a(0) =a(0,1,0).
 \end{cases}
 \]
We denote $G_a(s) = (x_a(s), y_a(s), z_a(s))$.  From the invariance of \eqref{eq:ode_G'''} under $s \rightarrow -s$ and the initial condition, we conclude that 
 \[
 \begin{split}
 x_a(s) &= -x_a(-s)\\
 y_a(s) &= y_a(-s)\\
 z_a(s)&=z_a(-s).
 \end{split}
\]
Therefore $x_a,y_a,z_a$ are three linearly independent solutions of the scalar ODE
\[
w''' + \Big (a^2 + \frac{s^2}4 \Big ) w' - \frac{s}4 w =0
\]
which are tempered distributions ($G_a, G'_a, G''_a, G'''_a$ grow polynomially). Applying the Fourier transform we get

\begin{equation}\label{eq:Four_tranf_G}
\xi \hat w'' + 3\hat w' +4\xi^3 \hat w -4\xi a^2 \hat w =0.
\end{equation}
is a singular regular ODE with $r=0$ and $r=-2$ the solutions of the indicial equation 
\[
r(r-1) +3r=0.
\]
Therefore the regular solutions have to be even and the odd solutions have to behave as $\omega_a \delta'$ at the origin. Since $\hat x_a$ is odd and nontrivial we have that $\omega_a \ne 0$. To sum up, using the vectorial notation $A_a^\pm = (A^\pm_{a,1}, A^\pm_{a,2}, A^\pm_{a,3})$ we get, recalling the definition \eqref{eq:defA_a^pm}
\[
\begin{split}
A^+_{a,1} &= A^-_{a,1}\\
A^+_{a,2} &= - A^-_{a,2}\\
A^+_{a,3} &= - A^-_{a,3}.
\end{split}
\]
Moreover $\omega_a \ne 0$ amounts to say that $A_{a,1}^+ \ne 0$. Hence we proved $i)$. \\
In order to prove $ii)$ we now assume that $A_a^+ = A_a^-$. This implies $A_{a,2}^+ = A_{a,2}^- = A_{a,3}^+ = A_{a,3}^- =0$. But then the even functions $\hat y_a(\xi), \hat z_a(\xi)$ should be regular at the origin and linearly dependent. This can be true only if $a=0$ and $A^+_0 = (1,0,0)= A^-_0$.
 \end{proof}
To conclude this subsection, we observe that it is also possible to find a relation between the angle $\theta$ determined by $A_a^+$ and $A_a^-$ and the parameter $a$. In order to find it the following geometric lemma is needed. We refer to \cite{BV12} for the proof for general curves and to \cite{GRV03} for the proof in the particular case of selfsimilar solutions.
  \begin{lemma}\label{lemma:BV}
 Let $(T_j, n_j, b_j)$ , $j=1,2,3$ be the components of the Frenet frame of a regular curve $G(s)$ with $s\in\R$ the arc length parameter and $(T(0), n(0), b(0)) = \mathbbm 1_{3\times3}$. Then
 \[
 \begin{split}
 T_j(s) &= 1 - \frac{\lvert \theta_j \rvert^2}{2 E(0)},\\
c(s) \big ( n_j(s) -ib_j(s) \big) &= -2\theta_j \bar \theta_j'
 \end{split}
 \]
 for $j=1,2,3$ with $\theta_j $ the solution of
 \[
 \theta_j^{''} + \bigg ( -\frac{c'}c + i\tau\bigg) \theta_j' + \frac{c^2}4 \theta_j = 0
 \]
 with
 \[
 \begin{split}
 \big( \theta_1(0), \theta'_1(0)\big) & = \bigg (0, \frac{c(0)}{\sqrt 2} \bigg),\\
 \big (\theta_2(0), \theta_2'(0)\big) & = \bigg (1, -\frac{c(0)}2 \bigg),\\
 \big(\theta_3(0), \theta_3'(0)\big) & = \bigg(i, \frac{c(0)}2 \bigg)
 \end{split}
 \]
 and
 \[
 E(s) \coloneqq \bigg \lvert \frac{\theta_j'}c \bigg \rvert^2 + \frac{\lvert \theta_j \rvert^2}4 = E(0) = \frac12.
 \]
 \end{lemma}
 
 In our case $\theta_j$ has to be a solution of 
 \begin{equation}
 \label{eq:theta_j}
 \theta''_j + i \frac s2 \theta_j' + \frac{a^2}4 \theta_j=0.
 \end{equation}

 This equation can be easily integrated computing the Fourier transform of $\theta_j$. 
 In fact, two linearly independent solutions of \eqref{eq:theta_j} can be taken as follows
 \begin{equation}
 \label{eq:betas}
 \begin{split}
 \beta_1(s) &= \int_{-\infty}^{+\infty} e^{i(s\xi + \xi^2)} \frac{d}{d\xi} \Big [ e^{-i \frac{a^2}2 \log \lvert \xi \rvert} \mathbbm I_{[0,+\infty)} (\xi) \Big ] d\xi, \\
 \beta_2(s) &=  \int_{-\infty}^{+\infty} e^{i(s\xi + \xi^2)} \frac{d}{d\xi} \Big [ e^{-i \frac{a^2}2 \log \lvert \xi \rvert} \mathbbm I_{(-\infty,0]} (\xi) \Big ] d\xi.
 \end{split}
 \end{equation}
 Recalling that $A_a^+ = (A_{a,1}, A_{a,2}, A_{a,3})$, $A^-_a =(A_{a,1}, -A_{a,2}, -A_{a,3})$ we have
 \[
 \cos \theta =1 - 2A^2_{a,1}, \qquad \sin \frac{\theta}2 = A_{a,1}.
 \]
 Then, it suffices to compute $A_{a,1} = \lim_{s\rightarrow \infty}T_1(s)$ using the first equality in \Cref{lemma:BV}. After some computations one gets
 \[
 A_{a,1}=e^{- \frac{\pi a^2}2}.
 \]
 Finally, notice that the dispersive relation of $\beta_1$ and $\beta_2$ is
 \[
 \phi (\xi)= \xi^2 - \frac{a^2}2 \log \lvert \xi \rvert,
 \]
 which indicates that the nonlinearity only makes a logarithmic correction to the free evolution.\\
 From \eqref{eq:betas} and \Cref{lemma:BV} the asymptotic behavior of $G_a$ and $T_a$ can be deduced also. The arguments, although lengthy, are rather elemental. They can be found in \cite{GRV03} where the following is proved.
 \begin{theorem}\label{thm:asymp_onecorner}
 Let $a \ge0$. The function
 \[
 \chi_a = \sqrt t G_a \bigg ( \frac s{\sqrt t} \bigg)
 \]
 with $G_a$ defined by \eqref{eq:ode_G_a} is a real analytic solution of \eqref{eq:CPchi} for $t>0$.\\
 Moreover there exist $A_a^+, A_a^-, B_a^+, B_a^-$ such that
 \begin{enumerate}[i)]
 \item $\big \lvert \chi_a(s,t) - A^+_a s \mathbbm I_{[0,+\infty)}(s) - A^-_a s \mathbbm I_{(-\infty,0]} (s) \big \rvert \le a \sqrt t$;\\
 \item let $\lvert s\rvert > \max (2a,4)$, then the following asymptotics, as $s \rightarrow \pm \infty$, hold
 \[
 \begin{split}
 G_a(s) &= A_a^\pm \Big ( s + \frac{2a^2}s \Big ) - \frac{4a n_a}{s^2} + O\Big (\frac 1{s^3}\Big ), \\
 T_a(s) &= A_a^\pm - \frac{2ab_a}s + O\Big ( \frac 1{s^2}\Big),\\
 (n_a -ib_a)(s) &= B_a^\pm e^{i \frac{s^2}4 + ia^2 \log \lvert s \rvert} + O\Big (\frac 1s \Big);
 \end{split}
 \]
 \item $A_a^\pm = \big(A^\pm_{a1}, A^\pm_{a2}, A^\pm_{a3} \big)$ are unitary vectors and
 \[
 A^+_{a1}= A^-_{a1}= e^{- \pi \frac{a^2}2}, \quad A^+_{a2}=-A^-_{a2}, \quad A^+_{a3}=-A^-_{a3}, \quad \langle A_a^\pm, B_a^\pm \rangle=0;
 \]
 \item $\sin \frac{\theta}2= A^\pm_{a1} = e^{- \pi \frac{a^2}2}$ and $\sin \phi = \frac{A_2}{\sqrt{1-A^2_2}}$ where $\theta$ is the angle between $A_a^+$ and $-A^-_a$ and $\phi$ is the angle between the plane that contains $A_a^+$ and $A^-_a$ and the vector $b_a(0)=(0,0,1)$;
 \item there exist $a_0$ and $a_1$ with $0<a_0\le a_1 < +\infty$ such that if $a<a_0$ then $\chi_a$ has no self-intersections and if $a_1\le a$ then $\chi_a$ has finitely many self-intersections.
 \end{enumerate}
 \end{theorem} 

\begin{remark}
A similar result holds also for the hyberbolic case, see \cite{dlH07} (Thm.\ 1). In that case the vectors $A_a^\pm \in \mathbb H^2$ describing the shape of the corner satisfies, using the notations as in \Cref{thm:asymp_onecorner}
\[
A^+_{a1} = - A_{a1}^-, \quad A_{a2}^+ = -A_{a2}^-, \quad A^+_{a3}=A_{a3}^-=e^{\pi \frac{a^2}2}.
\]
Moreover, we refer the reader to \cite{HKV22} (Appendix A) where, using a completely different approach wrt \cite{dlH07} and \cite{GRV03}, the authors using the Laplace transform rederive $A_3$ in both cases the euclidean and the hyperbolic ones. 
\end{remark} 
 \subsection{An Introduction to 1d cubic NLS}\label{sec:intro1dcubicnls}\noindent
As seen in \Cref{sct:sectionBF}, the 1d cubic Schrödinger equation appears naturally in the description of the BF (via the Hasimoto trasformation \eqref{Has_trans}). Here we recall some classical properties of the above mentioned equation, for the unfamiliar reader. We refer to \cite{CW90, Tao06, LP15} and the references therein for more details.\\
The 1d cubic nonlinear Schrödinger equation (NLS) equation is given by
\begin{align}
    i \partial_t u + \partial_{xx} u \pm  |u|^2u = 0,
    \label{NLS}
\end{align}
where $u = u(x,t) \colon \R_t \times \R_x \rightarrow \mathbb C$. As it was previously pointed out, the $\pm$ sign in the right hand side refers to the focusing (euclidean) or defocusing (hyperbolic) nature of the equation. 
If we consider the Cauchy problem associated to \eqref{NLS} with initial data $u(x,0) = u_0(x)$, we can heuristically express the solution through the linear evolution and Duhamel's formula:
\begin{align}
    u(x,t) = e^{it\partial_{xx}} u_0 \pm i \int_0^t e^{i(t-\tau)\partial_{xx}} u|u|^2(\tau) d\tau,
    \label{linear&duhamel}
\end{align}
where, through Fourier inversion, we define
\begin{align}
    e^{it\partial_{xx}} u_0 = \int_{\mathbb{R}} e^{-it|\xi|^2 + ix\xi} \hat{u}_0(\xi) d\xi.
    \label{homsol}
\end{align}
We present now some relevant symmetries of the NLS \eqref{NLS}:
    \begin{itemize}
        \item \underline{Translations}: It is straightforward that the equation is invariant under space and time translations.
        \item \underline{Gauge invariance}: It is clear, both from the linear terms and the nonlinear one, that \eqref{NLS} is invariant under any Gauge transformation $u \mapsto e^{i\theta}u$, for any $\theta \in \mathbb{R}$.
        \item \underline{Scaling invariance} : given $\lambda \in \mathbb{R} \setminus \{0\}$ and $\alpha$ a parameter to be determined, consider the scaling
        \begin{align*}
            u_{\lambda}(x,t) = \frac{1}{\lambda^\alpha} u\left(\frac{x}{\lambda} , \frac{t}{\lambda^2}\right).
        \end{align*}
        Note that, for $u_\lambda$ to solve \eqref{NLS}, we need
        \begin{align*}
            &|u_\lambda|^2 u_\lambda = \lambda^{-3\alpha} \left|u\left(\frac{x}{\lambda} , \frac{t}{\lambda^2}\right)\right|^2u\left(\frac{x}{\lambda} , \frac{t}{\lambda^2}\right), \nonumber \\
            &(\partial_t - i\partial_{x}^2) u_\lambda = \lambda^{-2-\alpha} ((\partial_t - i\partial_{x}^2)u)\left(\frac{x}{\lambda} , \frac{t}{\lambda^2}\right),
        \end{align*}
        so by imposing the equality of the resulting exponents, we obtain that if $\alpha=1$, then the function
        \begin{align}
            u_{\lambda}(x,t) = \frac{1}{\lambda} u(\frac{x}{\lambda} , \frac{t}{\lambda^2})
            \label{scaling}
        \end{align}
        is still a solution of \eqref{NLS}.
        \item \underline{Galilean transformation}: the reason of this covariance comes from the fact that \eqref{NLS} comes from nonrelativistic quantum mechanics, so it must be invariant under the so-called Galilean boost
        \begin{align}
            p' = p + mv, \quad x' = x + vt,
            \label{galileanboost}
        \end{align}
        given a moving frame speed $v$. As we can see in the wave packet decomposition from the homogeneous solution \eqref{homsol}, the phase of a Schrödinger plane wave with frequency $w = w(\xi) = |\xi|^2$ is given by $x \xi - wt$. Given $m$ a mass, if we multiply such a phase by $mw/\xi^2$, we obtain
        \begin{align*}
            (\frac{mw}{\xi})x - (\frac{mw^2}{\xi^2})t, 
        \end{align*}
        which, when doing a dimensional inspection, corresponds to $px - Et$, for $p$ linear momentum, and $E$ energy, both corresponding to a plane wave. Namely, the latter corresponds, when talking of the homogeneous solution of the Schrödinger equation, to the energy of a free particle with mass $m$ and linear momentum $p$, i.e., $E = p^2/2m$. Applying Galilean boost \eqref{galileanboost} we obtain a phase corresponding to
        \begin{align*}
            &(p' - mv)(x' - vt) - \frac{(p' - mv)^2}{2m}t \nonumber \\
            &= p'x' - \frac{p'^2}{2m}t + \left( \frac{mv^2t}{2} - mvx' \right),
        \end{align*}
        what means that differences in plane waves phases comes with position and time, but not with momentum. In fact, given a superposition of plane waves, up to $x$ and $t$ the boost affects each plane wave equally, what establishes such an invariance for general solutions in $L^2(\mathbb{R})$. This implies that, given $u$ a solution of \eqref{NLS} and $v \in \mathbb{R}$, then 
        \begin{align*}
            u_v(x,t) = e^{i(mvx - \frac{mv^2t}{2})} u(x-vt,t)
        \end{align*}
        is a solution of the same equation. For convenience, given $N \in \mathbb{R}$, we take $m = 1/2$ and $v = 2N$ and write the transformation as
        \begin{align}
            u_N(x,t) = e^{iNx - iN^2t} u(x - 2Nt,t).
            \label{galileantransf}
        \end{align}
\end{itemize} 
        \begin{remark}\label{fun_gal_inv}
        This covariance drives us to the following consequence.
        Given $u_a = a \delta$ for any $a \in \mathbb{R}$ and $\delta$ the ordinary point mass distribution, we have that $e^{iNx} u_a = u_a$ for any $N > 0$. Moreover, if $u$ solves \eqref{NLS}, then $u_N$ given by \eqref{galileantransf} also solves it. Therefore, if we assume that uniqueness is guaranteed, then the solution is Galilean invariant at any instant, i.e., $u_N(t,\cdot) = u(t,\cdot)$ for any $N \in \mathbb{R}$. This implies that, for our case of cubic nonlinearity and one spatial dimension, we obtain a solution
        \begin{align}
            u(x,t) = a\frac{e^{ix^2/4t}}{\sqrt{it}} e^{\pm ia^2\log(t)}.
            \label{NLSdeltaanomaly}
        \end{align}
        In particular, this shows that the corresponding IVP is ill-posed under these conditions (see \cite{KPV01}).\\
        As an exercise, We show that the only solution left invariant by Galilean transformation is indeed give by \eqref{NLSdeltaanomaly}; we differentiate with respect to $N$ the equality
        \begin{align}
            &u(x,t) = u_N(x,t) = e^{iNx - iN^2t} u(x-2Nt,t) \nonumber\\
            &\Longrightarrow i[x - 2Nt] u(x-2Nt,t) -2t \partial_{x_1}u(x-2Nt,t) = 0 \nonumber\\
            &\Longrightarrow \partial_{x_1} u(x-2Nt,t ) = \frac{i(x-2Nt)}{2t} u(x-2Nt,t),
            \label{eq:GalileanTransfDerivative}
        \end{align}
        where we denoted $x_1$ the first component of the argument of $u$. Since \eqref{eq:GalileanTransfDerivative} holds for every $N \in \mathbb{R}$, then
        \begin{align*}
            \partial_x u(x,t) = \frac{ix}{2t} u(x,t).
        \end{align*}
        Solutions are given by
        \begin{align*}
            u(x,t) = g(t) e^{i \frac{x^2}{4t}}
        \end{align*}
        for some function $g : \mathbb{R} \mapsto \mathbb{C}$. Imposing that $u$ is a solution of the cubic NLS we find the following ODE for g:
        \begin{align}
            &i\partial_t\left(g(t) e^{\frac{ix^2}{4t}}\right) + \partial_{xx}\left(g(t) e^{\frac{ix^2}{4t}}\right) \pm \frac{1}{2} \left|g(t) e^{\frac{ix^2}{4t}}\right| \cdot g(t) e^{\frac{ix^2}{4t}} = 0  \nonumber\\
            &\Longrightarrow \left[ ig'(t) + \frac{x^2}{4t^2} g(t) + \frac{i}{2t} g(t) - \frac{x^2}{4t^2}g(t) \pm \frac{1}{2}|g(t)|^2g(t)\right] e^{\frac{ix^2}{4t}} = 0 \nonumber\\
            &\Longrightarrow g'(t) = -\frac{1}{2t}g(t) \pm \frac{i}{2} |g(t)|^2g(t) \label{eq:g(t)NLS_Delta}.
        \end{align}
        By multiplying by $t$ and $\overline{g}(t)$, and afterwards conjugating the resulting equality, we obtain the relations
        \begin{align*}
            &tg'(t) \overline{g}(t) + \frac{1}{2} |g(t)|^2 = \pm \frac{it}{2} |g(t)|^4, \nonumber\\
            &t\overline{g}'(t)g(t) + \frac{1}{2} |g(t)|^2 = \mp \frac{it}{2} |g(t)|^4.
        \end{align*} 
        By summing these two equalities,
        \begin{align*}
            2 t \text{Re}(\overline{g}(t) g'(t)) + |g(t)|^2 = 0.
        \end{align*}
        Moreover, note that
        \begin{align*}
            \partial_t(t |g(t)|^2) = 2 t \text{Re}(\overline{g}(t) g'(t)) + |g(t)|^2.
        \end{align*}
        Therefore, $t \lvert g(t) \rvert^2$ is constant in time. Then, setting $g(t)= \frac{c}{\sqrt t}e^{i\theta(t)}$ in the equation \eqref{eq:g(t)NLS_Delta}, with $c \in \mathbb C$ and $\theta$ a function in $t$ taking real values, we get \eqref{NLSdeltaanomaly}. 
        \end{remark}

Before presenting some of the useful conserved quantities associated to these equations, we describe a class of solutions, call \emph{solitons} in order to better understand the behaviour of solutions of \eqref{NLS}. Let us first consider special solutions, the standing waves. A standing wave is a solution to \eqref{NLS} of the form $u(x,t) = e^{i \omega t} Q(x)$, where $\omega \in \R$ and $0 \ne Q \in H^1(\R)$ is a localized solution to the elliptic stationary equation 
\begin{equation}\label{eq:Q}
Q''(x)  - \omega Q \pm Q^3 =0.
\end{equation}
A soliton is a standing wave that has been given a speed thanks to a Galilean transform. We showed before that \eqref{NLS} is invariant under Galilean trasformation so a soliton is still a solution of the same equation. In particular, a soliton with frequency $\omega >0$, speed $v \in \R$ has the form
\[
\phi(x,t) = Q_\omega(x - vt) e^{\frac i2 vx - \frac 14 v^2 t + i \omega t}.
\]
In the particular case of the 1d cubic NLS it has been showed that solitons can appear only in the focusing case. It is not hard to check that solutions of the focusing \eqref{eq:Q} are given by
\[
Q_\omega(x) = \sqrt {\frac 2 \omega}  \sech \Big ( \frac x{\sqrt \omega} \Big ), \quad \omega >0 .
\]

\begin{remark}
 Soliton solutions behaves very differently from free solutions. They are global non scattering solutions with a fixed profile, which is always localized, then they do not disperse (see e.g.\ \cite{LF28, Har82}). 
\end{remark}
Furthermore, we present a non exhaustive list of conservation laws for the \eqref{NLS} (see \cite{KT16} for more details):
    \begin{itemize}
        \item If we multiply \eqref{NLS} by $\overline{u}$, we obtain
        \begin{align}
            (\partial_t u) \overline{u} &= i((\partial_{xx} u)\overline{u} \pm |u|^2u\overline{u}) \nonumber\\
            & = i(\partial_x(\overline{u}\partial_x u) - \partial_x u \partial_x\overline{u} \pm |u|^4).
            \label{cl11}
        \end{align}
        If we conjugate in \eqref{cl11},
        \begin{align}
            u (\partial_t\overline{u}) = -i(\partial_x(u \partial_x \overline{u}) - \partial_x u \partial_x \overline{u} \pm |u|^4).
            \label{cl12}
        \end{align}
        Now, sum up \eqref{cl11} and \eqref{cl12} to obtain the so-called \emph{local mass conservation}:
        \begin{align}
            \partial_t |u|^2 = i \partial_x(\overline{u}\partial_x u - u \partial_x \overline{u}) = -2 \text{Im}\partial_x(u\partial_x \overline{u}).
            \label{cl13}
        \end{align}
        Then, integrating on $x$, both in $\Omega = \mathbb{R}$ or $\Omega = \mathbb{T}$, we have the \emph{mass conservation}:
        \begin{align}
            \partial_t \int_{\Omega} |u|^2 dx = 0 \Longrightarrow ||u||_{L^2(\Omega)} \quad \text{is constant}.
            \label{enconservation1}
        \end{align}
       
         \item By multiplying the equation \eqref{NLS} by $\partial_x \bar u$, and taking the imaginary part, we get the \emph{conservation of momentum}
        \[
        \partial_t \int_\Omega \text{Im} \{ u \partial_x \bar u\} dx = 0
        \]
        \item By multiplying now \eqref{NLS} by $\partial_t \overline{u}$, following a similar procedure, we obtain the \emph{energy conservation}
        \begin{align}
           \partial_t \int_{\R} \Big (|\partial_x u|^2 \mp \frac{1}{2}|u|^4 \Big )dx = 0
            \label{eq:energyconservation}
        \end{align}
        Notice that, in the defocusing case, the conservation of the energy and of the $L^2$ norm imply the conservation of the $H^1_x(\R)$ norm.\\
    \end{itemize}
    The Cauchy problem associated equation \eqref{NLS} is well-posed for initial datum in $H^s$, for $s \geq 0$ (see \cite{CW90} and \cite{GV79}). On the other hand, for $s<0$ the Cauchy problem is ill-posed: 
    as already stated, in \cite{KPV01} either existence or uniqueness were proved to be lost using the Galilean transformation, and in \cite{CCT03} norm-inflation phenomena were displayed. \\
 
    Note that, if we consider the equation \eqref{NLS} for $x \in \R^d$, $d \ge 1$ the threshold for the validity of local wellposdess in $H^s$ of the NLS 
    obtained with respect to the scaling invariance $u_{\lambda} = \lambda^{2/(k-1)} u (\lambda^2 t, \lambda x)$ is $\dot{H}^{\frac{d}{2} - \frac{2}{k-1}}(\mathbb{R})$ (for general nonlinearity $|u|^{k-1}u$). To see this, let $s \in \mathbb{R}$ and $u_0 \in \dot{H}^s = \dot{H}^s(\mathbb{R}^d)$ the corresponding initial condition. Recalling that $\hat{u}_{\lambda}(\xi) = \lambda^{-d} \hat{u}(\lambda^2 t,\frac{\xi}{\lambda})$, then
\begin{align*}
    &||\lambda^{\frac{2}{k-1}} u (\lambda^2 t, \lambda x)||_{\dot{H}^s}^2 = \lambda^{\frac{4}{k-1}} \int_{\mathbb{R}} |\xi|^{2s} |\mathcal{F}_s[u(\lambda^2 t, \lambda s)](\xi)|^2 d\xi \nonumber\\
    &=   \lambda^{\frac{4}{k-1} - 2d} \int_{\mathbb{R}} |\lambda \xi|^{2s} |\hat{u}(\lambda^2 t,\frac{\xi}{\lambda})|^2 d\xi = \lambda^{2s + \frac{4}{k-1} - d} || u(\lambda^2 t, \cdot)||_{\dot{H}^s}^2.
\end{align*}
If we take $t=0$ in the last equality, the invariance
\begin{align*}
    ||\lambda^{\frac{2}{k-1}} u (\lambda^2 t, \lambda x)||_{\dot{H}^s} = ||u_0||_{\dot{H}^s}
\end{align*}
imposes that 
\begin{align*}
    s = \frac{d}{2} - \frac{2}{k-1}.
\end{align*}
In our case we have $d=1$ and $k = 3$, so $s = -1/2$ is the critical value. For $s \leq -1/2$, the presence of norm inflating phenomena with loss of regularity was pointed out in \cite{CK17} and \cite{Ki06}, and also norm-inflating around any data was proved in \cite{Oh17}. Finally, growth control for Sobolev norms of Schwartz solutions for $s \in (-\frac{1}{2},0)$ on $\mathbb{R}$ or $\mathbb{T}$ was shown in \cite{KVZ} and \cite{KT16}.
\\

We now prove the local well-posedness in $L^2(\R)$ for the IVP 
    \begin{equation}\label{IVP_nls}
    \begin{cases}
        \partial_t u = i (\partial_{xx} u \pm \lvert u \rvert^2 u)\\
        u(x,0)=u_0(x).
    \end{cases}
    \end{equation}
    That is, we show that if $u_0\in L^2 (\R)$ there exists a positive time $T$, which will depend on the initial data, such that there exists a unique solution $u$ of \eqref{IVP_nls}, $u \in \mathcal C ([0,T]; L^2(\R))$. For that purpose, firstly we state and prove a lemma due to C. Fefferman \cite{F70}.
	\begin{lemma}
		Given $u \in L^2(\mathbb{R})$,
		\begin{align*}
			||e^{it\partial_{xx}} u||_{L^6_{t,x}(\mathbb{R}^2)} \lesssim ||u||_{L^2(\mathbb{R})}.
		\end{align*}
		\label{lemmaFefferman}
	\end{lemma}
	\begin{proof}
    Let $u \in L^2(\mathbb{R})$. We can write, through inverse Fourier transformation,
    \begin{align*}
        |e^{it\partial_{xx}}u(x,t)|^2 = \int_{\mathbb{R}}\int_{\mathbb{R}} e^{it(\xi_1^2 - \xi_2^2) + ix(\xi_1 - \xi_2)} \hat{u}(\xi_1) \overline{\hat{u}(\xi_2)} d\xi_1 d\xi_2.
    \end{align*}
    Now, do the change of variables
    \begin{equation}
        \left\{\begin{aligned}
			&\eta_1 = \xi_1^2 - \xi_2^2 \nonumber \\
			&\eta_2 = \xi_1 - \xi_2
		\end{aligned}\right.
  \end{equation}
  so that
  \begin{equation}
  \left\{\begin{aligned}
			&\xi_1 = \frac{1}{2} (\frac{\eta_1}{\eta_2} + \eta_2) \nonumber \\
			&\xi_2 = \frac{1}{2} (\frac{\eta_1}{\eta_2} - \eta_2).
		\end{aligned}\right.
    \end{equation}
    The corresponding diffeomorphism yields the following Jacobian:
    \begin{align*}
        \mathcal{J} = \left|\begin{matrix}
\frac{1}{2\eta_2} & \frac{1}{2}[1 - \frac{\eta_1}{\eta_2^2}] \\
 \frac{1}{2\eta_2} & -\frac{1}{2}[1 + \frac{\eta_1}{\eta_2^2}]
\end{matrix}\right| = -\frac{1}{2 \eta_2} = -\frac{1}{2(\xi_1 - \xi_2)}.
    \end{align*}
    Therefore, if we denote
    \begin{align*}
        F(\eta_1,\eta_2) = \frac{\hat{u}(\xi_1) \overline{\hat{u}}(\xi_2)}{2 |\xi_1 - \xi_2|},
    \end{align*}
    we have that
    \begin{align*}
        |e^{it\partial_{xx}}u(x,t)|^2 = \int_{\mathbb{R}} \int_{\mathbb{R}} e^{i t \eta_1} e^{i x \eta_2} F(\eta_1,\eta_2) d\eta_1 d\eta_2 = \mathcal{F}_{\eta_1,\eta_2}[F](x,t).
    \end{align*}
    Therefore, by reordering exponents, and applying Plancherel, we obtain
    \begin{align*}
        || e^{i t \partial_{xx}} u ||_{L^6_{t,x}(\mathbb{R}^2)} = || \hspace{1mm}|e^{i t \partial_{xx}} u|^2 \hspace{1mm} ||_{L^3_{t,x}(\{\mathbb{R}^2)}^{1/2} = || \hat{F} ||^{1/2}_{L^3_{t,x}(\mathbb{R}^2)}.
    \end{align*}
    We apply now Young-Hausdorff inequality, i.e., in our case, the fact that the Fourier transform is continuous from $L^p_{t,x}(\mathbb{R}^2)$ to  $L^{p'}_{t,x}(\mathbb{R}^2)$ for any $p \in [1,2]$. Thus, undoing the previous change of variables,
    \begin{align*}
       || e^{i t \partial_{xx}} u ||_{L^6_{t,x}(\mathbb{R}^2)} &\lesssim ||F||^{1/2}_{L^{3/2}_{\eta_1,\eta_2}(\mathbb{R}^2)}  = \left( \int_{\mathbb{R}} \int_{\mathbb{R}} |\hat{u}(\xi_1) \overline{\hat{u}(\xi_2)}|^{3/2} \frac{d \xi_1 d\xi_2}{(2|\xi_1 - \xi_2|)^{\frac{3}{2}-1}} \right)^{1/3} \nonumber\\
       & =\left( \int_{\mathbb{R}} \left[ \int_{\mathbb{R}} |\hat{u}(\xi_2)|^{3/2} \frac{d \xi_2}{\sqrt{2|\xi_1 - \xi_2|}} \right] |\hat{u}(\xi_1)|^{3/2} d\xi_1 \right)^{1/3}\\
       &= \left( \int_{\mathbb{R}} Iu(\xi_1) |\hat{u}(\xi_1)|^{3/2} d\xi_1 \right)^{1/3},
    \end{align*}
    where
    \begin{align*}
        Iu(\xi_1) = \int_{\mathbb{R}} |\hat{u}(\xi_2)|^{3/2} \frac{d \xi_2}{\sqrt{2|\xi_1 - \xi_2|}} = [|\hat{u}|^{3/2} \ast (\frac{1}{\sqrt{2\cdot}}) ] (\xi_1).
    \end{align*}
    Using a Hölder inequality,
    \begin{align*}
        || e^{i t \partial_{xx}} u ||_{L^6_{t,x}(\mathbb{R}^2)} \lesssim ||\hat{u}^{3/2}||^{1/3}_{L^{4/3}(\mathbb{R}} ||Iu||^{1/3}_{L^4(\mathbb{R})}.
    \end{align*}
    For the second factor in the RHS, we apply now Hardy-Littlewood-Sobolev inequality. This inequality, at the end, is an application of the weak Young inequality \cite{GRA14}. Note that, for any $g \in L^r(\mathbb{R}^d)$, the inequality
    \begin{align*}
        || \int_{\mathbb{R}} \frac{g(y)}{|x-y|^{\alpha}} dx||_{L^p(\mathbb{R}^d)} \lesssim || |y|^{-\alpha} ||_{L^{q,\infty}(\mathbb{R}^d)} ||g||_{L^r((\mathbb{R}^d)}
    \end{align*}
    follows from weak Young inequality, given that $1+\frac{1}{p} = \frac{1}{q} + \frac{1}{r}$, and if $\frac{1}{q} = \frac{\alpha}{d}$, in order to ensure that
    \begin{align*}
        || |y|^{-\alpha} ||_{L^{q,\infty}(\mathbb{R}^d)} = \sup_{\lambda > 0} \lambda |\{ \frac{1}{|y|^{\alpha}} > \lambda \}|^{1/q} < \infty.
    \end{align*}
    In our case, $\alpha = 1/2$, $p=4$ and $d=1$, so $r = 4/3$. All in all,
    \begin{align*}
        ||\hat{u}^{3/2}||^{1/3}_{L^{4/3}(\mathbb{R})} ||Iu||^{1/3}_{L^4(\mathbb{R})} &\lesssim \left[ \left(\int_{\mathbb{R}} |\hat{u}(\xi_1)|^{\frac{3}{2}\frac{4}{3}} d\xi_1\right)^{3/4} || \hat{u}^{3/2} ||_{L^{4/3}(\mathbb{R}))} \right]^{1/3} \nonumber\\
        &= ||\hat{u}||_{L^2(\mathbb{R})}^{1/2} ||\hat{u}||_{L^2(\mathbb{R})}^{1/2} = ||u||_{L^2(\mathbb{R})},
    \end{align*}
	where we conclude applying Plancherel theorem.
	\end{proof}

 Moreover, we recall that the Schrödinger propagator is unitary in $L^2$, that is, the following Lemma holds:
 \begin{lemma}
     Let $u \in L^2(\R)$, then for all $t \in \R$
     \begin{equation}\label{prop_unitary}
     \big \lVert e^{it \partial_{xx}} u \big \lVert_{L^2(\R)} = \lVert u \rVert_{L^2(\R)}
     \end{equation}
 \end{lemma}

We define the map 
\[
\Gamma (u) \coloneqq e^{it \partial_{xx}} u_0 \pm i \int_0^t e^{i (t-\tau) \partial_{xx}} \big (u \lvert u \rvert^2)(\tau)\big ) d\tau.
\]
Recalling the Duhamel formula \eqref{linear&duhamel} the solution $u$ to the IVP \eqref{IVP_nls} is a fixed point of $\Gamma.$ Then, we want to find a complete metric space $X$ such that $\Gamma$ is a contraction on that space, in order to get the local well-posedness by a fixed point argument.
We consider the metric space
\[
X_{T,M} \coloneqq \big \{ u \text{ measurable in } \R^2 \colon \sup_{t \le T} \lVert u \rVert_{L^2_x} + \Big ( \int_0^T \int_{\R} \lvert u (x,t) \rvert^6 dx \, dt \Big)^{\frac 16} \le M \Big \}
\]
where $T,M > 0$ will be chosen later. We endow it with the norm given by
\[
\lVert v \rVert_{X} \coloneqq \lVert v \rVert_{L^\infty ([0,T]; L^2 (\R))} + \lVert v \rVert_{L^6([0,T];L^6(\R))}
\]
which is well defined for the elements in $X_{T,M}$. We claim that there exist $T,M >0$ such that $\Gamma(u)$ maps $X_{T,M}$ in itself; we estimate, by Minkowski inequality and \eqref{prop_unitary}
\[
\begin{split}
    \lVert \Gamma(u) \rVert_{L^\infty ([0,T]; L^2 (\R))} & \le \lVert u_0 \rVert_{L^2_x} + \int_0^T \big \lVert e^{i (t-\tau)\partial_{xx}} (\lvert u \rvert^2 u)(\tau) d\tau \big \rVert_{L^\infty ([0,T]; L^2 (\R))} d\tau \\
    & \le \lVert u_0 \rVert_{L^2_x} + \int_0^T \big \lVert \lvert u \rvert^2 u(\tau)\big \rVert _{L^2_x} d\tau \\
    & \le \lVert u_0 \rVert_{L^2_x} + \int_0^T \lVert u \rVert_{L^6_x}^3 d\tau.
\end{split}
\]
Then, by Hölder inequality and estimate \eqref{lemmaFefferman} we get that
\[
\lVert \Gamma(u) \rVert_{L^\infty ([0,T]; L^2 (\R))}  \le \lVert u_0 \rVert_{L^2_x} + T^{\frac12} M^3.
\]
We observe that in the same way we can also estimate the $L^6$-norm of $\Gamma$ obtaining the same estimate. Then, if we choose $T,M$ such that 
\[
\lVert u_0 \rVert_{L^2} \le \frac M4 \quad \text{and} \quad T^{\frac12} M^2 \le \frac 14
\]
we get the claim. We show now that the map is a contraction on $X_{T,M}$. Let $u,v$ be two solutions of \eqref{IVP_nls}. We estimate, again using Minkowski inequality, \eqref{prop_unitary} and lemma \eqref{lemmaFefferman}
\begin{equation}\label{eq:contraction}
\begin{split}
   \big \lVert \Gamma(u) - \Gamma(v) \big \rVert _X \le 2 \int_0^T \big \lVert (\lvert u \rvert^2 u - \lvert v \rvert^2 v)(\tau) \big \lVert_{L^2_x} d\tau.
\end{split}
\end{equation}
We now observe that for all $u,v \in \mathbb C$ 
\[
\lvert u \rvert^2 u - \lvert v \rvert ^2 v = (u-v)(\lvert u \rvert^2 + \lvert v \rvert^2) + (\bar u - \bar v) uv.
\]
Then, using the above equality and the Hölder inequality 
\[
\begin{split}
    \eqref{eq:contraction} & \le 2 \int_0^T \big \lVert \lvert u-v\rvert (\lvert u \rvert^2 + \lvert v \rvert^2 + \lvert uv \rvert) \big \rVert_{L^2_x} d\tau\\
    &  \le 2 \int_0^T \lVert u- v\rVert_{L^6_x} \big (\lVert u \rVert_{L^6_x}^2 + \lVert v \rVert_{L^6_x}^2 + \lVert u \rVert_{L^6_x}\lVert v \rVert_{L^6_x} \big ) d\tau\\
    & \le 6 M^2 T^\frac12 \lVert u-v\rVert_X.
\end{split}
\]
Summing up, if we choose $M,T>0$ such that
\[
\begin{cases}
M \ge 4 \lVert u_0 \rVert_{L^2_x},\\
T^\frac12< \frac 1{6M^2}
\end{cases}
\]
we have that for every $u_0 \in L^2(\R)$ there exists a time $T >0$ and a solution $u$ to \eqref{IVP_nls}, $u \in C ([0,T];L^2(\R) \cap L^6([0,T]; L^6(\R))$. Moreover, we observe that the time of existence of $u$ is only depending on the $L^2$ norm of the initial data $u_0$. However, by \eqref{enconservation1} we have that the $L^2$ norm of the solution is constant in time so we can repeat the argument above in order to construct a global solution to \eqref{IVP_nls}, both for the focusing and defocusing case.

   In the following sections, we are going only to deal with the Schrödinger equation on the entire real line, mostly in the focusing case. However, it is worth to observe that the behaviour of the Schrödinger equation on compact manifolds, in particular on rational or irrational tori, can be different with respect to the one described here. For example, the estimate \eqref{lemmaFefferman} is known to fail on rational tori. We present a counterexample in the \Cref{fail_L6_tori}, following the approach from \cite{Bou93}.

    \begin{lemma}\label{fail_L6_tori}
        There exists $u \in L^2(\mathbb{T})$ such that
        \begin{align*}
            ||e^{it\partial_{xx}} u||_{{L^6}_{t,x}(\mathbb{T}^2)} \lesssim ||u||_{L^2(\mathbb{T})}
        \end{align*}
        does not hold.
    \end{lemma}
    \begin{proof}
        Consider $N \in \mathbb{N}$ and a $u_0 \in L^2(\mathbb{T})$ such that $\supp \hat{u} \subset B(0,N)$, defined by
        \begin{equation}
        \hat{u}_0(k) = 
		\left\{\begin{aligned}
			&1 \text{ if } |k| \leq N, \nonumber \\
			&0 \text{ if } |k| > N.
		\end{aligned}\right.
	\end{equation}
    We define
    \begin{align*}
         u = e^{i t \partial_{xx}} u_0 = \sum_{|k| \leq N} e^{ikx}.
    \end{align*}
    Denote $\underline{j} = (j_1,j_2,j_3)$ and $I_N = \{(j_1,j_2,j_3) : |j_i| \leq N, \hspace{1mm} i\in \{1,2,3\}\}$. We have that
    \begin{align*}
        u \overline{u}u = \sum_{\underline{j} \in I_N} e^{it(j_1^2 - j_2^2 + k_3^2) + i x (j_1 - j_2 + j_3)} = \sum_{m = -N}^{2N} e^{imx} \left[\sum_{j_1 - j_2 + j_3 = m} e^{it(j_1^2 - j_2^2 + j_3^2)}\right].
    \end{align*}
    Thus,
    \begin{align*}
        &\int_0^{2\pi} |u \overline{u} u|^2 dx = \int_0^{2\pi} \left|\sum_m e^{imx} \left[ \sum_{j_1 - j_2 + j_3 = m} e^{it(j_1^2 - j_2^2 + j_3^2)} \right] \right|^2 dx \nonumber\\
        &= \sum_m \left| \sum_{j_1 - j_2 + j_3 = m} e^{it(j_1^2 - j_2^2 + j_3^2)} \right|^2,
    \end{align*}
    where last equality follows from orthogonality. For a fixed $m$, isolate $j_3 = m - j_1 + j_2$. Then, denoting $l = j_1^2 - j_2^2 + j_3^2$, we obtain
    \begin{align*}
        &l = j_1^2 - j_2^2 + (m - j_1 + j_2)^2 = j_1^2 - j_2^2 + m^2 + j_2^2 + j_1^2 - 2j_2j_1 + 2mj_2 - 2mj_1 \nonumber\\
        &\Longrightarrow j = l - m^2 = 2j_1^2 - 2j_2j_1 + 2mj_2 - 2mj_1 = 2[j_1(j_1 - j_2) - m(j_1 - j_2)] \nonumber\\
        & = 2(j_1 - j_2)(j_1 - m).
    \end{align*}
    Then,
    \begin{align*}
        \int_0^{2\pi} |u \overline{u} u|^2 dx = \sum_m \left| \sum_{j_1 - j_2 + j_3 = m} e^{2it (j_1-m)(j_1-j_2)} \right|^2.
    \end{align*}
    Therefore, integrating also with respect to time,
    \begin{align}
        &||u||^6_{L^6_{t,x}(\mathbb{T}^2)} = \int_0^{2\pi} \int_0^{2\pi} |u \overline{u} u|^2 dx dt = \sum_m \int_0^{2\pi} \left| \sum_{j \in J_N} \sum_{2(j_1-m)(j_1-j_2)=j} e^{ijt} \right|^2 dt \nonumber\\
        &=  \sum_m \sum_{j \in J_N} |r_j|^2,
        \label{CEBouInt1}
    \end{align}
    where we used orthogonality in the third equality, and we denoted $r_j = \{ (j_1,j_2) : 2(j_1 - m)(j_1 - j_2) = j \}$ and $J_N$ the range in which the index $j$ moves. To be formally more clear, note that denoting $q = j_1 - j_2$ and $p = j_1-m$, then 
    \begin{align*}
        |r_j| = |\{(p+m,p-q+m) : 2pq = j \}| = |\{(p,q) : 2pq = j \}|.
    \end{align*}
    Now, note that counting the pairs $(p,q)$ such that $2pq = j$ is the same as counting the number of $p$ such that there exists $q \in \mathbb{Z}$ with $2pq = j$. Then, $|r_j|$ is the number of divisors of $j/2$ given that $j \in 2\mathbb{Z}$, i.e. $d(j/2)$. Then $|r_j| = d(j) - 1$, where we excluded from the counting the divisor $2$. Then, proceeding with \eqref{CEBouInt1}, using the inequality $|r_j|^2 \geq |r_j|$, and
    without clarifying how $J_N$ is, we have that
    \begin{align}
        ||u||^6_{L^6_{t,x}(\mathbb{T}^2)} \geq \sum_{m=-N}^{2 N} \sum_{j \in J_N}  |r_j| \approx N \sum_{j \in J_N} (d(j) - 1).
        \label{CEBouInt2}
    \end{align}
    We need then to elucidate $J_N$. We claim, and prove later, that
    \begin{align}
        \lim_{ N \rightarrow \infty} \frac{|J_N|}{N^2} = \frac{6}{\pi^2}.
        \label{claimCoprimes}
    \end{align}
    In particular, this implies that, for large $N$, $|J_N| \approx N^2$. Thus, from \eqref{CEBouInt2}, we have that
    \begin{align}
        ||u||^6_{L^6_{t,x}(\mathbb{T}^2)} \gtrsim N \sum_{|j| \leq N^2} d(j) - N^3.
        \label{CEBouInt3}
    \end{align}
    We invoke a result from the theorem 320 in \cite{HW08}, which states that, for any $x \geq 1$,
    \begin{align*}
        \sum_{n \leq x} d(n) = x\log(x) + (2\gamma - 1)x + \mathcal{O}(\sqrt{x}),
    \end{align*}
    where $\gamma$ is Euler's constant. Then, applying \eqref{claimCoprimes} in \eqref{CEBouInt3},
    \begin{align*}
        ||u||^6_{L^6_{t,x}(\mathbb{T}^2)}  \gtrsim N(N^2 \log(N^2)) - N^3 = N^3( 2\log(N) - 1),
    \end{align*}
    which for $N$ large can be written as follows
    \begin{align*}
        ||u||^6_{L^6_{t,x}(\mathbb{T}^2)}  \gtrsim  N^3 \log(N) = (N^{1/2})^6 \log(N). 
    \end{align*}
    Therefore, since $||u_0||_{L^2_x(\mathbb{T})} = N^{1/2}$, we have just proved that
    \begin{align}
        ||u||^6_{L^6_{t,x}(\mathbb{T}^2)} \geq ||u_0||_{L^2_x(\mathbb{T})}^6 \log(N)
        \label{eq:IntermittencyL2}
    \end{align}
    for $N$ sufficiently large. To finish the proof, we show now the claim \eqref{claimCoprimes}:
    \\
    We give a sketch of the proof. In our context, following analogous ideas from \cite{Bou93}, in order to consider $|J_N|$ it suffices to take $|\Tilde{J}_N|$, where
    \begin{align*}
        \Tilde{J}_N = \{q/p  : p, q \in \mathbb{Z}, (p,q) = 1, p \leq N\}
    \end{align*}
    Thus, it suffices to prove that $\{|\Tilde{J}_N|\}_N$ is an equivalent sequence to $\{\frac{6}{\pi^2}N^2\}_N$. Let's formulate this in probabilistic and heuristic terms. In other words and more concretely, due to the ordinary Laplace rule, we want to count the probability that, given $(p,q) \in \mathbb{N}^2$, they are coprime. Consider certain $p \in \mathbb{N}$. Now take any prime $l$ under $p$. Note that $1/l$ is the proportion of multipliers of $l$ in all $[0,\infty)$. From this idea it can be proved that there is a probability of $\frac{1}{l}$ that $p$ has $l$ in its primes factorization. Do the same with another $q \in \mathbb{N}$. Then, the probability of being $l$ in both factorizations is $1/l^2$. This yields the probability of the complement of the event of interest, which means that the probability of that $l$ is not a common factor of both integers $p$ and $q$ is $1 - \frac{1}{l^2}$. 
    \\
    All in all, if we multiply such an expression over all primer numbers, we get the desired value. Luckily, this is the Euler product identity, so
    \begin{align*}
        \lim_{N \rightarrow \infty} \frac{|J_N|}{N^2} = \prod_{l \text{ prime}} \left(1 - \frac{1}{l^2}\right) = \frac{1}{\sum_{n=1}^{\infty} \frac{1}{n^2}} = \frac{6}{\pi^2}.
    \end{align*}
    \end{proof}
The last result, and in particular equality \eqref{eq:IntermittencyL2}, is an analytical expression of the phenomenon of intermittency (see chapter 8 of \cite{F95}). However, the estimate holds if one asks the function for more regularity, as shown in the Lemma below. The proof of it, as well as for the local well-posedness in our periodic setting and for any other dimension, can be found in the seminal paper \cite{Bou93}.
    
    \begin{lemma}
        Let $\varepsilon > 0$ and $f \in H^{\varepsilon}(\mathbb{T})$. Then,
        \begin{align*}
            ||e^{it\partial_{xx}} f||_{L^6_{t,x}(\mathbb{T}^2)} \lesssim ||f||_{H^{\varepsilon}(\mathbb{T})}.
        \end{align*}
    \end{lemma}

\section{Stability of one corner}\label{sec:Stability of one corner}
    We now present some known results concerning the stability of self similar solutions of the BF. In particular, we want to understand if, under some appropriate small regular perturbations the corner, given as initial data for the unperturbed case, still remains at time zero or not. In other words, we wonder whether creation and annihilation of corners is stable or not.
	\\
	Recall that, through the Hasimoto transformation, we went from a solution of the BF \eqref{eq:ode_G_a} into solution of a 1d cubic NLS equation with a wavefunction solution given by $\psi(x,t) = c(x,t) e^{i \int_0^t \tau(s',t)ds'}$, where $(c,\tau)$ are the curvature and torsion of the filament (see \Cref{sct:sectionBF}). For our self-similar solution \eqref{eq:G}, we obtained $c = a$ and $\tau = s/2$ for some $a \geq 0$ constant. Then, our particular solution is now
 \begin{align}
     \psi_a(x,t) = \frac{a}{\sqrt{it}} e^{i s^2/4t}, \quad t > 0,
     \label{SSwavefunction}
 \end{align}
 which is a solution of 
 \begin{equation}\label{eq:Cauchypb_NLSa2t}
     \begin{cases}
			&\psi_t = i(\psi_{ss} + \frac{1}{2} (|\psi|^2 + A(t)) \psi ), \nonumber \\
			&\psi(s,0) = a\delta,
		\end{cases}
 \end{equation}
 where $A(t) = -a^2/t$.
 \begin{remark}
 Ordinary techniques to deal with NLS, such as restriction spaces or Strichartz estimates, are not in general applicable to rough initial data, such as the $\delta$-distributions. Note, for instance, that neither $\psi_a$ in \eqref{SSwavefunction} nor the initial data are in any Sobolev space $H^s(\mathbb{R})$, for any $s \geq 0$. This translates into the fact that, though 1d cubic NLS is a complete integrable system with infinitely many conservation laws, that is not useful in our framework. See \cite{BV082} for a deeper detail of this question. 
 \end{remark}
  \begin{remark}
 As studied in \cite{KPV01} and pointed out before (\Cref{fun_gal_inv}), through Galilean invariance one can prove that solutions of \eqref{eq:cubicNLS_no12} with $A(t) = 0$ and $u(0,x) = a \delta_0$ either does not exist or it is not unique, since the natural candidate \eqref{NLSdeltaanomaly} has no limit when $t \rightarrow 0$. Nevertheless, as we have just said, through Hasimoto transformation we have $A(t) = - a^2/t$ , so we have solution even for $t=0$.
 \end{remark}
 There is, however, a natural energy associated to the family of solutions $\psi_a$. To describe it, let's rewrite the cubic NLS \eqref{eq:Cauchypb_NLSa2t} in selfsimilar variables. In other words, we perform the so called \emph{pseudoconformal transformation}. In both the focusing and defocusing cases (upcoming computations work equally), we can apply a change of variables and get rid of the $\frac12$-factor in front of the nonlinearity (as shown before) to get
 \begin{align}\label{eq:cubicNLS_no12}
   i \partial_t u + u_{ss} \pm \left( |u|^2 - \frac{a^2}{t} \right) u = 0.
 \end{align}
 Now, consider the pseudo-conformal transformation
 \begin{align}
     u(s,t) = \mathcal T v(s,t) = \frac{e^{is^2/4t}}{\sqrt{t}} \overline{v} \Big (\frac{s}{t},\frac{1}{t} \Big ).
     \label{TRANSpconformal}
 \end{align}
 This transforms the NLS into
 \begin{align}
     iv_t + v_{ss} \pm \frac{1}{t} (|v|^2 - a^2) v = 0, \quad t > 0,
     \label{NLSpconformal}
 \end{align}
and
 \begin{align}
     u_a(s,t) = \frac{a}{\sqrt{2it}} e^{i s^2/4t}, \quad t > 0 
     \label{eq:ua}
 \end{align}
 (i.e., the solution $\psi_a$ after a rescaling $(s,t) \mapsto (s/\sqrt{2},t/2)$) is transformed into 
 \[
 v_a=a, \quad t >0,
 \]
 which is a particular solution of \eqref{NLSpconformal}. Therefore, we are interested in the stability of $v_a$. Notice that \eqref{TRANSpconformal} transforms $t=0$ into $t=+\infty$. Therefore to construct a solution $u$ of \eqref{eq:cubicNLS_no12} in the time interval $[0,1]$ is equivalent to find a solution $v$ of \eqref{NLSpconformal} in $[1, +\infty]$. Moreover, the limit of $u(t)$ at $t\rightarrow 0^+$ becomes the limit of $v(t)$ as $t \rightarrow +\infty$. As a consequence we want to study the long time behaviour of $v$. In particular, it will be interesting to know whether the perturbations to $a$ in the last regime, after the pseudo-conformal transformation, scatter or not. This analysis will be divided in two parts: the \emph{existence of a wave operator}, and the \emph{asymptotic completeness}.

\subsection{Existence of Wave Operator} 
\label{section:WaveOperator}

First of all, note that there exists a natural energy for \eqref{NLSpconformal}
\begin{align}
    E(t) = \frac{1}{2} \int |v_s(t)|^2ds \mp \frac{1}{4t} \int (|v(t)|^2 - a^2)^2 ds,
    \label{InfEnergy}
\end{align}
where the second term is a nonlinear potential, and which fulfills 
\begin{align}
    E'(t) &= \int \text{Re} (v_s \overline{v}_{st}) \mp (-\frac{1}{4t^2}) \int (|v(t)|^2 - a^2)^2 ds \mp \frac{1}{2t} \int \text{Re}(v\overline{v}_t) (|v|^2 - a^2) ds \nonumber\\
    &= \pm \frac{1}{4t^2} \int (|v^2| - a^2) ds,
    \label{DerivativeInfEnergy}
\end{align}
where, in the second expression, first and third terms are cancelled by integration by parts and substitution through \eqref{NLSpconformal}.
Considering the defocusing case (i.e.\ the negative sign in the cubic nonlinearity), it would be enough to construct a global solution $v(t)$ of \eqref{NLSpconformal} with $E(t) < \infty$ and $||v(t) -a||_{L^2} < \infty$. This, along with \eqref{InfEnergy} and \eqref{DerivativeInfEnergy}, provides us with orbital stability of the solution $v_a = a$ (see details in \cite{BV081}). This defocusing case is related to a Schrödinger map \eqref{eq:T_thypplane} with $\mathbb{H}^2$ as target. However, for the evolution of vortex filaments we have to deal with $\mathbb{S}^2$ as target, or equivalently with the focusing sign. In this case, energy $E(t)$ from \eqref{InfEnergy} is not useful, since, due to such a negative sign, both terms in $E(t)$ can grow indefinitely keeping the difference of them finite. Therefore, we need other techniques which, as we will see, will work in both cases defocusing and focusing.
\\
We perturb the solution $v_a$, and write the resulting sum as
\begin{align*}
    v= w + a.
\end{align*}
Then, if we want $v$ to be a solution of \eqref{NLSpconformal}, $w$ has to solve
\begin{align}
    iw_t + w_{ss} \pm \frac{1}{t}(|w+a|^2 - a^2)(w+a) = 0.
    \label{PerturbedNLS}
\end{align}
Note that we can rewrite the nonlinear potential as
\begin{align*}
    \frac{1}{t}(|w+a|^2 - a^2)(w+a) &= \frac{1}{t}(|w|^2 + 2\text{Re}(w)a)(w+a) \nonumber\\
    &= \frac{|w|^2w}{t}  + a \frac{2 \text{Re}(w)w}{t} + a\frac{|w|^2}{t} + a^2\frac{2\text{Re}(w)}{t}.
\end{align*}
This means that the linear part in \eqref{PerturbedNLS} is $\pm \frac{a^2}t (w + \bar w)$. We observe that we can cancel the linear factor $\pm \frac{a^2}t w$. Indeed, if we take $z(s,t) = e^{\mp i a^2 \log t} w(s,t)$, it solves 
\[
\begin{split}
  iz_t + z_{ss} \pm \frac1t \big ( z \lvert z\rvert^2 +2a \text{Re}( e^{\pm i a^2 \log t} z)z + a \lvert z \rvert^2 e^{\mp i a^2 \log t} + a^2 e^{\mp 2i a^2 \log t} \bar z \big ) = 0.  
\end{split}
\]
We rewrite it as
\begin{align}
    i z_t + z_{ss} \pm e^{\mp 2 i a^2 \log t} \frac{a^2}{t} \overline{z }\pm \frac{1}{t}F(z,\overline{z}) = 0,
    \label{NLSIntFactor}
\end{align}
where the nonlinear potential, given by quadratic and cubic terms in $z$ and $\overline{z}$, takes the form
\begin{align}
    F(z,\overline{z}) =  |z|^2z + a(2\text{Re}(e^{\pm i a^2 \log t}z)z + |z|^2 e^{\mp i a^2 \log t}).
    \label{NLPotentialEstability}
\end{align}
Given an asymptotic state $u_+$, we look for a solution of \eqref{NLSIntFactor} for $t>1$ such that
\begin{align}
    ||z(t) - e^{it\partial_{ss}} u_+ ||_{L^2_x} \rightarrow 0 \text{ as } t \rightarrow  +\infty.
    \label{AsymptoticConv}
\end{align}
\begin{remark}
Without introducing the integrating factor, we would have to estimate a Duhamel term like 
\begin{align*}
    \int_t^{\infty} e^{i(t-\tau)\partial_{ss}} e^{i \tau \partial_{ss}} u_+ \frac{d\tau}{\tau},
\end{align*}
which is not controllable. 
\end{remark}
This can be done by finding a fixed point for the operator
\begin{align*}
    A(z) = e^{it\partial_{ss}} u_+ \pm i \int_t^{\infty} e^{i(t-\tau)\partial_{ss}} \left(\frac{a^2}{\tau} e^{\mp 2 i a^2 \log(\tau)} \overline{z}(\tau) + \frac{1}{\tau}F(z,\overline{z})\right) d\tau.
\end{align*}
The delicate part is the one given by the linear term $\overline{z}(\tau)$. We sum and substract $\overline{e^{i \tau \partial_{ss}}u_+}$ inside that part of the integral
\begin{align*}
    &||\int_t^{\infty} e^{i(t-\tau)\partial_{ss}} \frac{a^2}{\tau} e^{\mp 2 i a^2 \log(\tau)} \overline{z}(\tau) d\tau||_{L^2_x} \nonumber\\
    &\leq ||\int_t^{\infty} e^{i(t-\tau)\partial_{ss}}\frac{a^2}{\tau} e^{\mp 2 i a^2 \log(\tau)} \left(  \overline{z}(\tau) - \overline{e^{i \tau \partial_{ss}}u_+} \right) d\tau||_{L^2_x} \nonumber\\
    &+ ||\int_t^{\infty} \frac{a^2}{\tau^{1 \pm 2ia^2}}\overline{e^{i \tau \partial_{ss}}u_+}  d\tau||_{L^2_x} \\
    & =(i) + (ii).
\end{align*}
The term $(i)$ can be estimated by
\begin{align*}
    (i) \leq \left| \int_t^{\infty} e^{i(t-\tau)\partial_{ss}} \frac{a^2}{\tau^{1\pm 2ia^2}} ||\overline{z}(\tau) - e^{-i\tau\partial_{ss}} \overline{u}_+ ||_{L^2_x}  d\tau \right|.
\end{align*}
Thus we expect that if we consider $||\overline{z}(\tau) - e^{-i\tau\partial_{ss}}\overline{u}_+ ||_{L^2_x} = \mathcal{O}(t^{-\alpha})$ for some $\alpha \in \mathbb{R}$ to be computed, then a fixed point can be obtained by the classical Picard iteration. This partially characterizes the space on which we apply the fixed point argument. Therefore, the key point is to find a good $L^2$ estimate for the oscillatory integral
\begin{align}
    I(t) = a^2\int_t^{\infty} e^{i(t-\tau)\partial_{ss}}  e^{-i\tau\partial_{ss}}\overline{u}_+  \frac{d\tau}{\tau^{1 \pm 2ia^2}}.
    \label{StatIntEstability}
\end{align}
This can be done using the Strichartz estimates for the Schrödinger linear flow that we recall for completeness (see e.\ g.\ \cite{KT98}). Given $d \in \mathbb{N}$, they are given by
\begin{align}
    || e^{it\Delta}f||_{L^p(\mathbb{R},L^q(\mathbb{R}^d))} || \lesssim ||f||_{L^2_x(\mathbb{R}^d)}.
    \label{EstStrichartzSchrod}
\end{align}
where $(p,q,d) \neq (2,\infty,2)$, $2 \leq p,q \leq \infty$ and must satisfy an $d$-admissibility condition:
\begin{align}
    \frac{1}{p} = \frac{d}{2}\left( \frac{1}2  - \frac{1}{q}\right).
    \label{dAdmissiblitySchrod}
\end{align}
The dual form of \eqref{EstStrichartzSchrod} is
\begin{align}
    || \int_{\mathbb{R}} e^{-is\Delta} F(\cdot,s) ds||_{L_x^2(\mathbb{R}^d)} \lesssim ||F||_{L^{p'}(\mathbb{R},L^{q'}(\mathbb{R}^d))}.
    \label{EstStrichartzSchrod_dual}
\end{align}
For \eqref{StatIntEstability} we must restrict the time domain in \eqref{EstStrichartzSchrod_dual} to $(t,\infty)$, from which
\begin{align*}
    & a^2 \Big \lVert \int_t^{\infty} e^{i(t-\tau)\partial_{ss}}  e^{-i\tau\partial_{ss}}\overline{u}_+  \frac{d\tau}{\tau^{1 \pm 2ia^2}}\Big \rVert_{L^2_s(\mathbb{R})} \\
    & = a^2 \Big \lVert \int_t^{\infty} e^{-2i\tau\partial_{ss}} \left(  \frac{\overline{u}_+}{\tau^{1 \pm 2ia^2}} \right) d\tau \Big \rVert_{L^2_s(\mathbb{R})} \nonumber\\
    &\lesssim a^2 \Big \lVert \frac{\overline{u}_+}{\tau} \Big \rVert_{L^{p'}((t,\infty),L^{q'}(\mathbb{R}))} = a^2 ||u_+||_{L^{q'}(\mathbb{R})} ||1/\tau||_{L^{p'}(t,\infty)}.
\end{align*}
Given $d=1$ in \eqref{dAdmissiblitySchrod}, we take $q=\infty$ (i.e., $q'=1$), and so $p=4$. Therefore,
\begin{align*}
    a^2 ||\int_t^{\infty} e^{i(t-\tau)\partial_{ss}}  e^{-i\tau\partial_{ss}}\overline{u}_+  \frac{d\tau}{\tau^{1 \pm 2ia^2}}||_{L^2_s(\mathbb{R})} \lesssim a^2 ||u_+||_{L^1_s(\mathbb{R})} t^{-1/4}.
\end{align*}
Denote
\begin{align*}
    v_1(s,t) = a + e^{\mp 2 i a^2 \log(t)} e^{it\partial_{ss}} u_+(s) = a + w_1(s,t).
\end{align*}
With a little bit of work the following can be proved (more details in \cite{BV082}).
\begin{theorem}
    Let $t_0 > 0$. There exists a constant $a_0 > 0$ such that fror all $a < a_0$, and for all $u_+$ small in $L^1 \cap L^2$ with respect to $a_0$ and $t_0$, equation \eqref{NLSpconformal} has a unique solution $v$ with
    \begin{align*}
        v - v_1 \in C([t_0,\infty),L^2(\mathbb{R})) \cap L^4([t_0,\infty),L^{\infty}(\mathbb{R})),
    \end{align*}
    satisfying, as $t \rightarrow \infty$,
    \begin{align*}
        ||v(t) - v_1(t)||_{L^2_s(\mathbb{R})} + ||v - v_1||_{L^4((t,\infty),L^{\infty}_s(\mathbb{R})} = \mathcal{O}(t^{-1/4}).
    \end{align*}
     \label{BadUBEstability}
\end{theorem}

\begin{remark}
Notice that we are assuming two smallness conditions. One is the parameter $a$ and the other one is in $u_+$. However, we will see that with appropriate assumptions on $u_+$ the smallness of $a$ can be avoided. 
\end{remark}

Once $v$ has been constructed we can undo the pseudoconformal transformation and recover $u$ that will be a perturbation of $u_a$ \eqref{eq:ua}. Our next step is to construct $\chi(s,t)$ a solution of the BF using the functions $u(s,t)$ that we have just constructed. However, it turns out that the $t^{-1/4}$ is not sufficient in order to prove the stability of the corner created at $t=0$ for the family of curves $\chi_a(s,t)$. We need to improve it. We exploit deeper the oscillations concerning the linear remaining term \eqref{StatIntEstability}. For that, we use the Fourier transform and strengthen the conditions on the asymptotic state $u_+$
\begin{align*}
    &\int_t^{\infty} e^{i(t-2\tau)\partial_{ss}} \left(  \frac{\overline{u}_+}{\tau^{1 \pm 2ia^2}} \right) d\tau = \int_t^{\infty} \int_{\mathbb{R}} \frac{e^{-i(t-2\tau) \xi^2}}{\tau^{1 \pm 2ia^2}} e^{is\xi} \hat{u}_+ (-\xi) d\xi d\tau \nonumber\\
    &= \int_{\mathbb{R}} \frac{e^{it\xi^2}}{t^{1 \pm 2ia^2}} e^{is\xi} \frac{\hat{u}_+(-\xi)}{2i\xi^2} d\xi + (1 \pm 2ia^2) \int_t^{\infty} \int_{\mathbb{R}} \frac{e^{-i(t-2\tau)\xi^2}}{\tau^{2 \pm 2ia^2}} e^{is\xi} \frac{\hat{u}_+(-\xi)}{2i\xi^2} d\xi d\tau \nonumber\\
    &= \int_{\mathbb{R}} \left[ \frac{e^{it\xi^2}}{t^{1 \pm 2ia^2}} \frac{\hat{u}_+(-\xi)}{2i\xi^2} + (1 \pm 2ia^2) \int_t^{\infty}  \frac{e^{-i(t-2\tau)\xi^2}}{\tau^{2 \pm 2ia^2}}  \frac{\hat{u}_+(-\xi)}{2i\xi^2}d\tau \right]  e^{is\xi} d\xi,
\end{align*}
where we applied integration by parts with respect to variable $\tau$ in the second equality. Thus, using Plancherel
\begin{align*}
    &||a^2 \int_t^{\infty} e^{i(t-2\tau)\partial_{ss}} \left(  \frac{\overline{u}_+}{\tau^{1 \pm 2ia^2}} \right) d\tau||_{L^2_s(\mathbb{R})} \nonumber\\
    &\leq a^2||  \frac{e^{it\xi^2}}{t^{1 \pm 2ia^2}} \frac{\hat{u}_+(-\xi)}{2i\xi^2} + (1 \pm 2ia^2) \int_t^{\infty}  \frac{e^{-i(t-2\tau)\xi^2}}{\tau^{2 \pm 2ia^2}}  \frac{\hat{u}_+(-\xi)}{2i\xi^2}d\tau ||_{L^2_s(\mathbb{R})} \nonumber\\
    &\lesssim \frac{a^2}{t} \left| \left|  \frac{\hat{u}_+(-\xi)}{\xi^2}\right|  \right|_{L^2_s(\mathbb{R})} = \frac{a^2}{t} ||u_+||_{\dot{H}^{-2}_s(\mathbb{R})}.
\end{align*}
This can be extended to the fact that the derivatives $\nabla_s^k$ of such a linear Duhamel term have a $1/t$ decay, with constants depending on the $\dot{H}^{-2+k}(\mathbb{R})$ of $u_+$. In \cite{BV082}, a better decay is obtained, as we are about to formulate.
\begin{theorem}\label{thm:GoodUBEstability}
    Let $t_0 > 0$ and $m \in \mathbb{N}^*$. There exists a constant $a_0 > 0$ such that for all $a < a_0$ and for all $u_+$ small in $\dot{H}^{-2}(\mathbb{R}) \cap H^m(\mathbb{R}) \cap W^{m,1}(\mathbb{R})$ with respect to $a_0$ and $t_0$, equation \eqref{NLSpconformal} has a unique solution $v$ with
    \begin{align*}
        v - v_1 \in C([t_0,\infty),H^m(\mathbb{R})),
    \end{align*}
    satisfying, as $t \rightarrow \infty$, for all entire $k \in(0,m]$,
    \begin{align*}
        ||v(t) - v_1(t)||_{L^2_s(\mathbb{R})} = \mathcal{O}(t^{-1/2}), \quad ||\partial^k_s(v(t) - v_1(t))||_{L^2_s(\mathbb{R})} = \mathcal{O}(t^{-1}).
    \end{align*}
    
\end{theorem}
Notice that if, given $a$, we take $u_+$ and its derivative small enough we can be sure that $\lvert v \rvert > \frac a2$.
\begin{remark}
The existence of two different rates of decay from \Cref{thm:GoodUBEstability} is due to the presence in the nonlinear potential \eqref{NLPotentialEstability} of the quadratic terms $\frac{2a}{t}|w|^2$ and $\frac{a}{t} w^2$. The less oscillating one, given the fact that we are omitting the phase of $w$, comes from the first one, which gives a rate of $t^{-1/2}$. We sketch how we treat the Duhamel term coming from such a worse quadratic term, by using Fourier oscillations:
\begin{align*}
    &\partial^k_s \int_t^{\infty} e^{i(t-\tau)\partial_{ss}} \left( \frac{\left| e^{i\tau \partial_{ss}} u_+(\tau) \right|^2}{\tau} \right) d\tau \nonumber\\
    &= \int_t^{\infty} \int_{\mathbb{R}} |\xi|^k \frac{e^{-i(t-\tau)\xi^2}}{\tau} e^{is\xi} \widehat{e^{i\tau \partial_{ss}} u_+} \ast \widehat{e^{-i\tau \partial_{ss}} \overline{u}_+} d\xi d\tau \nonumber\\
    &= \int_t^{\infty} \int_{\mathbb{R}}\int_{\mathbb{R}} |\xi|^k \frac{e^{is \xi}}{\tau} e^{-it\xi^2} e^{i\tau[\xi^2 - \eta^2 + (\xi-\eta)^2]} \hat{u}_+(\eta) \widehat{\overline{u}}_+(\xi-\eta) d\eta d\xi d\tau \nonumber\\
    & = \int_t^{\infty} \int_{\mathbb{R}}\int_{\mathbb{R}} |\xi|^k \frac{e^{is \xi}}{\tau} e^{-it\xi^2 - 2i\tau \xi(\eta-\xi)} \hat{u}_+(\eta) \widehat{\overline{u}}_+(\xi-\eta) d\eta d\xi d\tau,
\end{align*}
where we extended the convolution, and reordered the exponents through $\xi^2 - \eta^2 + (\xi - \eta)^2 = -2(\eta-\xi)\xi$. By integrating by parts in $\tau$,
\begin{align*}
    &\int_t^{\infty} \int_{\mathbb{R}}\int_{\mathbb{R}} |\xi|^k \frac{e^{is \xi}}{\tau} e^{-it\xi^2 - 2i\tau \xi(\eta-\xi)} \hat{u}_+(\eta) \widehat{\overline{u}}_+(\xi-\eta) d\eta d\xi d\tau \nonumber\\
    &= \int_{\mathbb{R}}\int_{\mathbb{R}} |\xi|^k \frac{e^{is \xi}}{t} e^{-i(t\xi^2 + 2\xi t(\eta-\xi))} \frac{\hat{u}_+(\eta) \widehat{\overline{u}}_+(\xi-\eta)}{2i\xi(\eta - \xi)} d\eta d\xi \nonumber\\
    &- \int_t^{\infty} \int_{\mathbb{R}}\int_{\mathbb{R}} |\xi|^k \frac{e^{is \xi}}{\tau^2} e^{-it\xi^2 - 2i\tau \xi(\eta-\xi)} \frac{\hat{u}_+(\eta) \widehat{\overline{u}}_+(\xi-\eta)}{2i\xi(\eta - \xi)} d\eta d\xi d\tau.
\end{align*}
Using again Plancherel like before, and taking $k=1$, we can then estimate the quadratic Duhamel term through a $\dot{H}^1$ norm
\begin{align*}
    &||\partial^k_s \int_t^{\infty} e^{i(t-\tau)\partial_{ss}} \left( \frac{\left| e^{i\tau \partial_{ss}} u_+(\tau) \right|^2}{\tau} \right) d\tau||_{L^2_s(\mathbb{R})} \lesssim \frac{a}{t} \left|\left| \widehat{u}_+ \ast \frac{\widehat{\overline{u}}_+}{\cdot} \right| \right|_{L^2_s(\mathbb{R})} \nonumber\\
    &= \frac{a}{t} \left|\left| u_+ \mathcal{F}^{-1}\left( \frac{\widehat{\overline{u}}_+}{\cdot} \right) \right|\right|_{L^2_s(\mathbb{R})} \leq \frac{a}{t}||u_+||_{\dot{H}^1(\mathbb{R})} ||u_+||_{\dot{H}^{-1}(\mathbb{R})},
\end{align*}
where $\mathcal{F}^{-1}$ denotes the inverse Fourier transform operator. This argument can be applied for any $k > 1$ in $H^k(\mathbb{R})$, but not for $k=0$, and that is why we get a better decay when differentiating.
\end{remark}
Undoing the pseudo-conformal transformation \eqref{TRANSpconformal} and recalling \eqref{eq:ua}, we get that $u(t)$ is close to
\begin{align}
    u_1(s,t) = u_a(s,t) + e^{\pm i a^2 \log(t)} \hat{\overline{u}}_+ \left( -\frac{s}{2} \right).
    \label{eq:u1}
\end{align}
The condition $u_+ \in \dot{H}^{-2}$ is translated on $\widehat{\overline{u}}_+$ by a cancelling condition at the origin, which could be related to the blow-up result of J. Bourgain and W. Wang \cite{BW}. It's important to notice that the abscense of a limit in $L^2$ for $u(t) - u_a(t)$ for time zero does not exclude the possibility to construct from $u$ a BF with a trace at zero, since curvature and torsion defined from $u$ correspond to second derivatives of the curves. 
\\
Now we come back to the BF equation. We want to export the stability result for solutions to the 1d cubic NLS into a stability result for self-similar solutions of BF. For that, we pose further questions: 
\begin{enumerate}
    \item From solutions of \eqref{thm:GoodUBEstability}, can we construct a family of curves $\chi$ which are solutions of the BF, and get information about the limit at $t=0$?
    \item Are those curves $\chi$ close to the self-similar solution $\chi_a$?
    \item Is there persistence in the formation of a singularity at time zero?
\end{enumerate}
The answers to all these questions will be found in \Cref{them_NLStoBF}. However, we would like to reach the thesis of the above mention Theorem step-by-step, referring to \cite{BV082} for all the details.\\
Let $v$ be a solution as in \Cref{thm:GoodUBEstability} whose regularity will be specified later and whose pseudo-conformal transformation is $u$. 
We can find regular functions $c(s,t) > \frac{a}{2 \sqrt t}$ and $\phi(s,t)$ such that 
\[
u(s,t) = c(s,t) e^{i\phi(s,t)}.
\]
Hence if we define the torsion as
\[
\tau(s,t) = \partial_s \phi (s,t),
\]
the filament function 
\[
\tilde u (s,t)= c(s,t) e^{i \int_0^s \tau(s',t)ds' + \phi(0,t)}
\]
will be a solution of \eqref{eq:Cauchypb_NLSa2t} with $A(t)= \frac {a^2}t + \partial_t \phi(0,t).$ Moreover, $(c,\tau)$ satisfy the intrinsic equations (see \cite{DaR06} and \cite{Be})
\begin{align}
    c_t = -(c\tau)_s -  c_s\tau_s, \quad
    \tau_t = \left(\frac{c_{ss} - c\tau^2}{c}  \right)_s + c_s c;
    \label{IntrinsicEqs}
\end{align}
As will be seen later, there exists a binormal flow of curves such that the curvature and the torsion are $c,\tau$ defined above. As $t$ goes to 0, we shall compute the leading terms of $(c, \tau)$. We have 
\begin{align}\label{c_NLStoBF}
    c(s,t) = |u(s,t)| = \frac{1}{\sqrt{t}} \left| v\left( \frac{s}{t} , \frac{1}{t} \right) \right|,
\end{align}
and 
\begin{align}
    \tau(s,t) = \text{Im } \left( \frac{u_s(s,t)}{u(s,t)} \right) = \text{Im } \left(  \frac{\frac{is}{2t} \overline{v}(\frac{s}{t},\frac{1}{t}) + \overline{v}_s(\frac{s}{t},\frac{1}{t})}{\overline{v}(\frac{s}{t},\frac{1}{t})} \right).
    \label{tau_NLStoBF}
\end{align}

Note from \eqref{IntrinsicEqs} that we need $c$ to be, at least, in $C^2$. We then impose $v \in H^{3+\varepsilon}(\mathbb{R})$ for an arbitrarily small $\varepsilon>0$, so that by Sobolev lemma we get $v \in C^{5/2}(\mathbb{R})$, given the continuous injection $H^{3+\varepsilon}(\mathbb{R}) \hookrightarrow C^{5/2}(\mathbb{R})$.
\\
By construction,
\begin{align*}
    v(s,t) = a + e^{\pm i a^2 \log(t)} e^{i t \partial_{ss}} u_+(s) + (v - v_1)(s,t).
\end{align*}
Then, it is possible to show, exploiting the decay from \Cref{thm:GoodUBEstability} and the decay properties of the linear flow that $(c,\tau)$ is close $(c_a,\tau_a) = \Big ( \frac {a}{\sqrt t}, \frac{s}{2t} \Big )$ with strong rate of decay as $t \rightarrow 0$, with constants depending on the norms of the asymptotic state $u_+$. Let us now construct the curves $\chi$, given $(c,\tau)$ from \eqref{tau_NLStoBF}, \eqref{c_NLStoBF}:
\begin{enumerate}
    \item We first construct its tangent, normal and binormal vectors $(T,n,b)(s,t)$ in the following way. For a given $(T,n,b)(\tilde t_0,0)$ we define $(T,n,b)(t,0)$ by imposing 
    \[
     \begin{pmatrix}
        T \\
        n \\
        b \\
        \end{pmatrix}_t(0,t) = \begin{pmatrix}
            0 & -c\tau & c_s \\
            c\tau & 0 & \frac{c_{ss} - c\tau^2}c \\
            - c_s & - \frac{c_{ss} - c\tau^2}c & 0 \\
        \end{pmatrix}
        \begin{pmatrix}
        T \\
        n \\
        b \\
        \end{pmatrix}(0,t).
    \]
    The choice of this system is due to the fact that
    \[
    \frac{a^2}t + \partial_t \phi(0,t) = 2 \Big ( \frac{c_{ss} - c \tau^2}c \Big) (0,t) + c(0,t)^2
    \]
    and the change of phase 
    \[
    \tilde n + i \tilde b= e^{i \frac {\phi}2} (n +ib).
    \]
    Then, we construct $(T,n,b)(s,t)$ from $(T,n,b)(0,t)$ by integrating the Frenet system for fixed $t$
    \begin{align}
        \begin{pmatrix}
        T \\
        n \\
        b \\
        \end{pmatrix}(s,t) 
        &= \begin{pmatrix}
        T \\
        n \\
        b \\
        \end{pmatrix}(0,\Tilde{t}_0) - \int_t^{\Tilde{t}_0} \begin{pmatrix}
0 & -c\tau & \partial_s c \\
c\tau & 0 & (\frac{\partial_{ss}c - c\tau^2}{c}) \\
-\partial_s c & -(\frac{\partial_{ss}c - c\tau^2}{c}) & 0 \\
\end{pmatrix}\begin{pmatrix}
        T \\
        n \\
        b \\
        \end{pmatrix}(0,t')dt' \nonumber\\
        &+ \int_0^s \begin{pmatrix}
0 & c & 0 \\
-c & 0 & \tau \\
0 & -\tau & 0 \\
\end{pmatrix} \begin{pmatrix}
        T \\
        n \\
        b \\
        \end{pmatrix}(s',t)ds'.
       \label{GralFormFrenetFrame_NLStoBF} 
    \end{align}
    \item In this way $T$ satisfies the Schrödinger map equation
    \begin{align*}
        T_t = T \wedge T_{ss},
    \end{align*}
    as a tangent vector of a BF.
    \item Once $T$ is constructed, for a given $\chi(0,\Tilde{t}_0)$ 
    we define
    \begin{align}
        \chi(s,t) = \chi(0,\Tilde{t}_0) - \int_t^{\Tilde{t}_0} (cb)(0,t')dt' + \int_0^s T(s',t) ds'.
        \label{1stDefChi_NLStoBF}
    \end{align}
    By using the Frenet system, we have that
    \begin{align}
        T_t &= T\wedge T_{ss} = T \wedge (cn)_s \\
        & = T \wedge ( c_sn + cn_s) = T \wedge ( c_sn - c^2 T + \tau b c) \nonumber\\
        &= c_sb - \tau c n,
    \end{align}
    so from \eqref{1stDefChi_NLStoBF}
    \begin{align*}
        \chi_t &= (cb)(0,t) + \int_0^s T_t(s',t)ds' = (cb)(0,t) + \int_0^s ( c_s b - \tau c n) (s',t) ds' \nonumber\\
        &=(cb)(0,t) + (cb(s',t))\rvert_0^s - \int_0^s (cb_s + \tau c n) (s',t) ds' \\
        &= (cb)(s,t) \nonumber + \int_0^s(-\tau c n + \tau c n)(s',t) ds' = (cb)(s,t),
    \end{align*}
    where we used integration by parts for the integral concerning $c_sb$. Therefore, $\chi$ satisfies the BF.
    \item Finally, $\chi(s,t)$ is obtained in this way for all $t \in (0,\Tilde{t}_0)$, as long as $(c,\tau)$ are regular enough. Moreover, since $c$ is close to $c_a = a/\sqrt{t}$, we have
    \begin{align*}
        &|\chi(s,t_1) - \chi(s,t_2)| = |\int_{t_1}^{t_2} c(s,t) b(s,t) dt| \nonumber\\
        &\leq |\int_{t_1}^{t_2}(c(s,t) - \frac{a}{\sqrt{t}}) b(s,t) dt| + |\int_{t_1}^{t_2} \frac{a b(s,t)}{\sqrt{t}} dt| \nonumber\\
        &\leq \int_{t_1}^{t_2} \frac{Ca}{\sqrt{t}} dt \lesssim Ca (\sqrt{t_2} - \sqrt{t_1}) \rightarrow 0, \text{ for } t_1,t_2 \rightarrow 0, 
    \end{align*}
    where in the second inequality we omitted the integral with the difference $c - \frac{a}{\sqrt{t}}$, since we know we can make it small enough; and where $C$ is some positive constant. This implies the existence of certain limit $\chi_0(s)$ as $t \rightarrow 0$. We similarly obtain that for all $s \in (-\infty,\infty)$
    \begin{align*}
        |\chi(s,t) - \chi_0(s)| \leq Ca\sqrt{t}.
    \end{align*}
    \item Now, we want to show that $\chi(s,0)$ is close to $\chi_a(s,0)$. First, we show that $T(0,t)$ is close to $T_a(0,t)$, and from there, that $T(s,t)$ is close to $T_a(s,t)$. 
    We provide a sketch of the proof for this last point, which is a key part.
    \\
    From \eqref{GralFormFrenetFrame_NLStoBF},
    \begin{align*}
         \begin{pmatrix}
        T \\
        n \\
        b \\
        \end{pmatrix}_t(0,t) = \begin{pmatrix}
0 & -c\tau & c_s \\
c\tau & 0 & (\frac{c_{ss} - c\tau^2}{c}) \\
-c_s & -(\frac{c_{ss} - c\tau^2}{c}) & 0 \\
\end{pmatrix} \begin{pmatrix}
        T \\
        n \\
        b \\
        \end{pmatrix}(0,t).
    \end{align*}
    With some computations, one has that at $s=0$
    \[
    \frac{a^2}{t} + \phi_t(0,t) = 2 \left( \frac{c_{ss} - c\tau^2}{c}\right)(0,t) + c^2(0,t),
    \]
    which is equivalent to the equality, at $s=0$,
      \begin{align*}
     \frac{c_{ss} - c\tau^2}{c} = \frac{c_a^2 - c^2}{2} + \frac{\phi_t}{2}.
    \end{align*}
    Then
    \begin{align*}
        \begin{pmatrix}
        T \\
        n \\
        b \\
        \end{pmatrix}_t(0,t) = \begin{pmatrix}
0 & -c\tau & c_s \\
c\tau & 0 & \frac{c_a^2 - c^2}{2} + \frac{\phi_t}{2} \\
-c_s & -(\frac{c_a^2 - c^2}{2} + \frac{\phi_t}{2}) & 0 \\
\end{pmatrix} \begin{pmatrix}
        T \\
        n \\
        b \\
        \end{pmatrix}(0,t).
    \end{align*}
    We want to get rid of the factor $\phi_t/2$, so we do the change of variable $\Tilde{n} + i \Tilde{b} = e^{i\frac{\phi}{2}} (n + ib)$ (note that such a factor only affects linearly to $n$ and $b$), and we impose the initial condition $(T,\Tilde{n},\Tilde{b})(0,\Tilde{t}_0) = (T_a,n_a,b_a)(0,\Tilde{t}_0)$. 
    \\
    Therefore, for having $(T,\Tilde{n},\Tilde{b})(0,t)$ close to $(T_a,n_a,b_a)(0,t)$, it suffices to show that
    \begin{align*}
        \begin{pmatrix}
0 & |c\tau| & |c_s| \\
|c\tau| & 0 & \left|\frac{c_a^2 - c^2}{2}\right| \\
|c_s| & \left|\frac{c_a^2 - c^2}{2}\right| & 0 \\
\end{pmatrix} (0,t)
    \end{align*}
    is integrable at $t=0$. 
    This is true, since we have
    \begin{enumerate}
        \item Strong decay of $(c,\tau) - (c_a,\tau_a)$.
        \item Good decay of $|e^{i\frac{\phi}{2}} - 1|$, so we have closedness of $(T,\Tilde{n},\Tilde{b})(0,t)$ to $(T,n,b)(0,t)$.
    \end{enumerate}
    Thus, $(T,n,b)(0,t)$ is close to $(T_a,n_a,b_a)(0,t)$. 
   \item We obtain finally that for any $\varepsilon > 0$, and $u_+$ small enough,
   \begin{align}
       |(T-T_a)(s,t)| \leq \varepsilon.
       \label{ConvergenceTTa}
   \end{align}
   In particular, from $\chi_s = T$, and summing and substracting the corresponding terms,
   \[
   \begin{split}
       \chi(s,t) - \chi(0,t) &= \int_0^s T(s',t)ds' \\
       & = A^+_a s + \int_0^s (T - T_a)(s',t)ds' + \int_0^s T_a(s',t)ds' - A^+_a s,
   \end{split}
   \]
   so given \eqref{ConvergenceTTa} and the fact that the third term tends to $0$ as $t \rightarrow 0$ uniformly in $s$ (see \cite{GRV03}); then, for $t \rightarrow 0$,
   \begin{align*}
       |\chi_0(s) - \chi_0(0) - sA^+_a| \leq \varepsilon s.
   \end{align*}
   More precisely, we have the following result. 
   \begin{theorem} \label{them_NLStoBF}
       Let $\varepsilon > 0$, $t_0 > 0$, $0 < a < a_0$, where $a_0$ is the constant in theorem \eqref{thm:GoodUBEstability}. Let $u_+$ small in $\dot{H}^{-2} \cap H^3 \cap W^{3,1}$ with $s^2 u_+$ small in $H^1$ with respect to $\varepsilon$, $a$ and $t_0$, and let $v$ be the corresponding solution obtained in theorem \eqref{thm:GoodUBEstability}. By using Hasimoto transformation, we construct from $v$ a family of curves $\chi(s,t)$ which solves the BF for $\Tilde{t}_0 > t > 0$, and such that there exists a unique $\chi_0$ satisfying
       \begin{align*}
           |\chi(s,t) - \chi_0(s)| < Ca\sqrt{t}
       \end{align*}
       uniformly in $s \in (-\infty,\infty)$.
       \\
       Moreover,
       \begin{align*}
           |\chi_0(s) - \chi_0(0) - sA^+_a \mathbb{I}_{[0,\infty)}(s) - sA^-_a \mathbb{I}_{(-\infty,0]}(s)| < \varepsilon|s|,
       \end{align*}
       where $A^{\pm}_a$ are two distinct non opposite unit vectors, such that the angle $\theta$ between $A^+_a$ and $A^-_a$ is determined by the relation
       \begin{align*}
           \sin \frac{\theta}{2} = e^{-\frac{\pi a^2}{2}}.
       \end{align*}
   \end{theorem}
   Since
   \begin{align*}
       \chi_a(s,0) = s A^+_a \mathbb{I}_{[0,\infty)}(s) + sA^-_a \mathbb{I}_{(-\infty,0]}(s),
   \end{align*}
   we have that $\chi_0(s)$ lies in a $\varepsilon$-cone around $\chi_a(s,0)$. This tells us that we have obtained a family of curves, solutions of the BF, which are close to the self-similar solution $\chi_a$. In particular, a corner is still formed for this curves $\chi(s,t)$ at $t=0$ and at $s=0$, meaning an affirmative answer to one of the main initial questions. The angle of this corner can be made as close as wanted to the one in between $A^+_a$ and $-A^-_a$, by taking $\varepsilon$ small enough. This shows that the formation of a singularity in finite time of $\chi_a$ is not an isolated phenomenon. In the next section, we will use these results to study the asymptotics to a generic initial data close to the selfsimilar solutions $\chi_a$. 
\end{enumerate}

\subsection{Asymptotic completeness}

In the previous section we proved the existence of solutions for $t \geq 1$ of the equation \eqref{NLSpconformal}, that at infinity behave as 
 \begin{align*}
     a - e^{\pm i a^2 \log(t)} e^{i t \partial_{ss}} u_+,
 \end{align*}
where $u_+$ is any given asymptotic state small in $L^1 \cap L^2$. Our purpose now is to find a solution for $t > 1$ and for any initial condition at $t = 1$, that belongs to an appropriate function space and that is sufficiently small with respect to $a$. Moreover, we aim to prove the asymptotics of this solution by finding the corresponding asymptotic state $u_+$. Finally, we will translate the result obtained for these solutions to solutions of the BF close to $\chi_a$. In \cite{BV12_1} the reader can find the detailed proofs for the results presented in this section.\\

We consider the linearization of equation \eqref{NLSpconformal} that we obtained in the previous section (equation \eqref{NLSIntFactor})
\begin{align}
    iz_t + z_{ss} \pm \frac{a^2}{t^{1 \pm 2ia^2}} \overline{z} = 0,
    \label{linearized_NLS}
\end{align}
with initial condition at $t_0 \geq 1$ given by $z(s,t_0)$.
\\
We start by presenting some a priori estimates on the Fourier modes of $z(t)$, that will allow us (see \Cref{proposition_LinearPseudconfEq_FreeSchrod}) to get a satisfactory global existence result.
\begin{lemma}
    If $z$ solves \eqref{linearized_NLS}, then
    \begin{align*}
        |\hat{z}(\xi,t)| \leq \frac{t^{a^2}}{t_0^{a^2}} (|\hat{z}(\xi,t_0)| + |\hat{z}(-\xi,t_0)| ).
    \end{align*}
    In particular, 
    \begin{align*}
        ||z(t)||_{\dot{H}^k} \lesssim \left(  \frac{t}{t_0}\right)^{a^2} ||u(t_0)||_{\dot{H}^k}
    \end{align*}
    for any $k \in \mathbb{Z}$.
    \label{lemma1_GrowthFourier}
\end{lemma}
\begin{proof}
    Through the Fourier transform, we have that
    \begin{align*}
        &\partial_t |\hat{z}(\xi,t)|^2 = (\partial_t \hat{z}(\xi,t)) \overline{\hat{z}(\xi,t)} + \hat{z}(\xi,t) (\partial_t\overline{\hat{z}(\xi,t)}) = -i\xi^2(|\hat{z}(\xi,t)|^2 - |\hat{z}(\xi,t)|^2) \nonumber\\
        &\pm ia^2 \left(\frac{\hat{\overline{z}} \cdot \overline{\hat{z}}}{t^{1 \pm 2ia^2}} - \frac{\check{z} \cdot \hat{z}}{t^{1 \mp 2ia^2}} \right) = \pm ia^2 \left(\frac{\overline{\hat{z} \cdot \check{z}}}{t^{1 \pm 2ia^2}} -\frac{\check{z} \cdot \hat{z}}{ t^{1 \mp 2ia^2}} \right) \nonumber\\
        &= \mp 2a^2 \text{Im} \left( \frac{1}{t^{1 \pm 2ia^2}}\overline{\hat{z}(\xi,t) \hat{z}(-\xi,t)} \right),
    \end{align*}
    where we substituted from the linearized equation \eqref{linearized_NLS} in the second equality. Then,
    \begin{align*}
        &\partial_t |\hat{z}(\xi,t)| = \partial_t \sqrt{|\hat{z}(\xi,t)|^2} = \frac{1}{2 |\hat{z}(\xi,t)|} \partial_t |\hat{z}(\xi,t)|^2 \leq \frac{a^2}{|\hat{z}(\xi,t)|} \text{Im} \left( \frac{1}{t^{1 \pm 2ia^2}}\overline{\hat{z}(\xi,t) \hat{z}(-\xi,t)} \right) \nonumber\\
        &\leq \frac{a^2}{t} |\hat{z}(-\xi,t)| \cdot \frac{|\hat{z}(\xi,t)|}{|\hat{z}(\xi,t)|} = \frac{a^2}{t} |\hat{z}(-\xi,t)|,
    \end{align*}
    where we have used the previous bound for $\partial_t|\hat{z}(\xi,t)|^2$, as well as $\text{Im }z \leq |z|$. If we do the same for $\partial_t|\hat{z}(-\xi,t)|^2$, we obtain
    \begin{align*}
        \partial_t |\hat{z}(-\xi,t)| \leq \frac{a^2}{t} |\hat{z}(\xi,t)|.
    \end{align*}
    Putting everything together, we obtain the desired result applying Gronwall inequality.
\end{proof}
As a complement to the previous lemma, we obtain the following one.
\begin{lemma}
    Let $\delta > 0$. If $z$ solves \eqref{linearized_NLS} then for all $\xi \neq 0$ and $0 < t_0 \leq t$,
    \begin{align*}
        |\hat{z}(\xi,t)| \leq \left( c(a) + \frac{c(a,\delta)}{(\xi^2 t_0)^{\delta}} \right)(|\hat{z}(\xi,t_0)| + |\hat{z}(-\xi,t_0)|).
    \end{align*}
    \label{lemma2_GrowthFourier}
\end{lemma}
\begin{proof}
    We give a sketch of the proof. It is better to go back to the unknown perturbation $w$ from \eqref{PerturbedNLS}, so that we write
    \begin{align*}
        z(s,t) = e^{\mp i a^2 \log(t)} w(s,t),
    \end{align*}
    for $w$ a solution of
    \begin{align}
        iw_t + w_{ss} \pm \frac{a^2}{t} (w + \overline{w}) = 0.
        \label{perturbedNLS_1}
    \end{align}
    We define, for $\xi \neq 0$,
    \begin{align}
        Y_{\xi}(t) = \widehat{\text{Re}w} (\xi,\frac{t}{\xi^2}), \quad Z_{\xi}(t) = \widehat{\text{Im}w} (\xi,\frac{t}{\xi^2}).
        \label{YxiZxi}
    \end{align}
    Note that $\widehat{w}(\xi,t/\xi^2) = Y_{\xi}(t) + i Z_{\xi}(t)$. On the other hand, applying Fourier transform on the equation \eqref{perturbedNLS_1},
    \begin{align*}
        &i\partial_t \widehat{w} - \xi^2 \widehat{w} \pm \frac{a^2}{t} (\widehat{w} + \widehat{\overline{w}}) = 0 \nonumber\\
        &\Longrightarrow i \partial_t \widehat{w} = \xi^2 (\widehat{\text{Re}w} + i\widehat{\operatorname{Im}w}) \mp \frac{2a^2}{t} \widehat{\text{Re}w}  \nonumber\\
        &\Longrightarrow -\partial_t\widehat{\text{Im}w} + i\partial_t\widehat{\text{Re}w} = (\xi^2 \mp \frac{2a^2}{t}) \widehat{\text{Re}w} + i\xi^2 \widehat{\text{Im}w}.
    \end{align*}
    Evaluate this equality with $(\xi,t)$ substituted by $(\xi,t/\xi^2)$, so that
    \begin{align*}
        \partial_t\widehat{\text{Im}w} = -\xi^2(1 \mp \frac{2a^2}{t \xi^2}) \widehat{\text{Re}w}; \quad \partial_t \widehat{\text{Re}w} = \xi^2 \widehat{\text{Im}w}.
    \end{align*}
    Therefore,
        \begin{align}
        &Y'_{\xi}(t) = \frac{\xi^2}{\xi^2}\widehat{\text{Im}w} (\xi,\frac{t}{\xi^2}) = Z_{\xi}(t), \nonumber\\
        &Z'_{\xi}(t) = -\frac{\xi^2}{\xi^2} \left(1 \mp \frac{2a^2}{t}\right) \widehat{\text{Re}w}(\xi,\frac{t}{\xi^2}) = \left(-1 \pm \frac{2a^2}{t}\right) Y_{\xi}(t).
        \label{DerYxiDerZxi}
    \end{align}
    Let $\varepsilon > 0$. We define
    \begin{align*}
        \sigma_{\varepsilon}(t) = \frac{1}{\varepsilon}|Y_{\xi}(t)|^2 + \varepsilon |Z_{\xi}(t)|^2.
    \end{align*}
    At this point, we only consider the positive sign case (i.e.\ the focusing case), since it is more delicate. We have that
    \[
        \sigma'_{\varepsilon}(t) \leq \left( \frac{1}{\varepsilon} - \varepsilon + \varepsilon\frac{2a^2}{t}  \right) \sigma_{\varepsilon}(t).
    \]
     Therefore    
     \[
      \log(\sigma_{\varepsilon}(t) / \sigma_{\varepsilon}(\Tilde{t}_0)) \leq \left( \frac{1}{\varepsilon} - \varepsilon \right) (t - \Tilde{t}_0) + 2a^2\varepsilon \log(t/\Tilde{t}_0),
     \]   
        which implies
     \[
     |\sigma_{\varepsilon}(t)| \leq |\sigma_{\varepsilon}(\Tilde{t}_0)| e^{(\frac{1}{\varepsilon} - \varepsilon) (t - \Tilde{t}_0)}  \left( \frac{t}{\Tilde{t}_0} \right)^{2a^2\varepsilon}.
     \]   
    After integration, an appropriate choice of $\varepsilon$ gives the estimate
    \begin{align*}
        |Y_{\xi}(t)|^2 + |Z_{\xi}(t)|^2 \leq \frac{t^{2a^2}}{\Tilde{t}_0^{2a^2}} (|Y_{\xi}(\Tilde{t}_0)|^2 +  |Z_{\xi}(\Tilde{t}_0)|^2) \leq \frac{c(a,\delta)}{\Tilde{t}_0^{2\delta}} (|Y_{\xi}(\Tilde{t}_0)|^2 +  |Z_{\xi}(\Tilde{t}_0)|^2),
    \end{align*}
    for $0 < \tilde t_0\le t \le \max \{4a^2, \frac12 \}$.
    For the case $4a^2 \leq \Tilde{t}_0 \leq  t$, the result follows by diagonalizing the system. In this region, the potential $1/t$ has to be taken into consideration and a logarithmic perturbation of the phase at infinity is necessary. After some manipulations (details can be found in \cite{BV12_1}, Lemma 2.2), it can be proved the existence of $(\dot{Y}^+_{\xi},\dot{Z}^+_{\xi})$ such that
    \begin{align*}
        &\left| \left( \frac{1}{2} Y_\xi - \frac{i}{2\alpha} Z_{\xi} \right) - e^{i \phi(t)} \dot{Y}^+_{\xi} \right| + \left| \left( \frac{1}{2} Y_\xi + \frac{i}{2\alpha} Z_{\xi} \right) - e^{i \phi(t)} \dot{Z}^+_{\xi} \right| \nonumber\\
        &\leq \frac{1}{t} \left( c(a) + \frac{c(a,\delta)}{\Tilde{t}_0^{\delta}} \right)(|Y_{\xi}(\Tilde{t}_0)| + |Z_{\xi}(\Tilde{t}_0)|)
    \end{align*}
    with
    \begin{align*}
        \alpha(t) = \sqrt{1-\frac{2a^2}{t}},
    \end{align*}
    and
    \begin{align*}
        \phi(t) = t - a^2 \log(t) - \int_t^{\infty} \left(\alpha(s) - 1 + \frac{a^2}{s}\right)ds.
    \end{align*}
    Notice that
    \begin{align*}
        \phi(t) = t - a^2 \log(t) + \mathcal{O}(1/t).
    \end{align*}
    Then, if we define $u_+$ as
    \begin{align*}
        \dot{Z}_{\xi}^+ = e^{i a^2 \log(\xi^2)} \hat{u}_+(\xi),
    \end{align*}
    we have that
    \begin{align*}
        |\hat{z}(\xi,t) - e^{it\xi^2 + ia^2\log(\xi^2)} \hat{u}_+(\xi) | \leq \frac{1}{t \xi^2} \left( c(a) + \frac{c(a,\delta)}{(\xi^2t_0)^{\delta}} \right)   (|\hat{z}(\xi,t_0)| + |\hat{z}(-\xi,t_0)|).
    \end{align*}
\end{proof}\noindent
From \Cref{lemma1_GrowthFourier} we expect that for an initial data in $H^s$ the corresponding solution has a growth of $\lVert z(t) \rVert_{H^s}$. To avoid this issue, we shall start with an initial data in a more restrictive space. Given $t_0$ fixed, we define the $X_{t_0}$ as the space of functions $f(s)$ such that
\begin{align*}
    ||f||_{X_{t_0}} := t_0^{-1/4} ||f||_{L^2} + t_0^{-1/2} ||\hat{f}||_{L^{\infty}(\xi^2 \leq t_0^{-1})} < \infty
\end{align*}
and $Y_{t_0}$ as the space of functions $g(s,t)$ such that
\begin{align*}
    ||g||_{Y_{t_0}} := \sup_{t \geq t_0} \left( t_0^{-1/4} ||g||_{L^2} + \left( \frac{t_0}{t} \right)^{a^2} t_0^{-1/2} ||\hat{g}(t)||_{L^{\infty}(\xi^2 \leq \frac{1}{t})} \right) < \infty.
\end{align*}

\begin{proposition}\label{proposition_LinearPseudconfEq_FreeSchrod}
    Let $t_0 \geq 1$ and $z_0 \in X_{t_0}$.
    Then, there exists a unique global solution $z \in Y_{t_0}$ of equation \eqref{linearized_NLS} with $z_0$ initial data at time $t_0$, and
    \[
    \lVert z \rVert_{Y_{t_0}} \le C(a) \lVert z_0 \rVert_{X_{t_0}}.
    \]
    More precisely,
    \begin{align*}
        \sup_{t \geq t_0} \frac{1}{t_0^{1/4}} ||z(t)||_{L^2(\mathbb{R})} &\leq C(a) \lVert z_0 \rVert_{X_{t_0}}, \nonumber\\
        \sup_{t \geq t_0} \left( \frac{t_0}{t} \right)^{a^2} \frac{1}{\sqrt{t_0}} || \hat{z}(\xi,t) ||_{L^{\infty}(\xi^2\leq 1)} &\leq C \frac{1}{\sqrt{t_0}} || \hat{z}_0(\xi) ||_{L^{\infty}(\xi^2\leq 1)},
    \end{align*}
    where $C$ and $C(a)$ are positive constants.
\end{proposition}
All in all, we get that the solution $z(t)$ of the linearized equation behaves as a free Schrödinger wave, as the following Proposition shows.
\begin{proposition}
    Let $t_0 \geq 1$, $\delta \in [0,\frac{1}{4})$ and $z_0$ a function such that
    \begin{align*}
        \frac{1}{t_0^{1/4}} ||z_0||_{L^2(\mathbb{R})} + \frac{t_0^{\delta}}{\sqrt{t_0}} || \hspace{1mm}|\xi|^{2 \delta} \hat{z}_0(\xi) ||_{L^{\infty}(\xi^2\leq 1)} < \infty.
    \end{align*}
    Then, the unique global solution $z$ that satisfies
    \begin{align*}
         \sup_{t \geq t_0} \left( \frac{1}{t_0^{1/4}} ||z(t)||_{L^2} + \left( \frac{t_0}{t} \right)^{a^2} \frac{t_0^{\delta}}{\sqrt{t_0}} |||\xi|^{2 \delta} \hspace{1mm} \hat{z}(t)||_{L^{\infty}(\xi^2 \leq 1)} \right) < \infty,
    \end{align*}
    with $z_0$ as initial data at time $t_0$, scatters in $L^2$. More precisely, there exists a scattering state $z_+ \in L^2(\mathbb{R})$ such that
    \begin{align*}
        &||z(t) - e^{i(t-t_0)\partial_{ss}} z_+ ||_{L^2(\mathbb{R})} \nonumber\\
        &\leq C(a,\delta) \frac{(1 + \log t_0)t_0^{\frac{1}{2}-(\delta + \delta_0)}}{t^{\frac{1}{4} - (\delta + \delta_0)}} \left( \frac{1}{t_0^{1/4}} ||z_0||_{L^2(\mathbb{R})} + \frac{t_0^{\delta}}{\sqrt{t_0}} || \hspace{1mm}|\xi|^{2 \delta} \hat{z}_0(\xi) ||_{L^{\infty}(\xi^2\leq 1)}  \right) \nonumber\\ &\rightarrow 0, \text{ for } t \rightarrow \infty,
    \end{align*}
    for any $\delta_0 \in (0,\frac{1}{4} - \delta)$, and where $C(a,\delta)$ is a nonnegative constant depending on $a$ and $\delta$.
    \label{proposition_LinearPseudconfEq_ScatteringState}
\end{proposition}
Our next goal is to construct global solutions to the nonlinear equation and caracterize its asymptotics. 

\begin{remark}
  In order to construct the nonlinear solution we will also need to prove the Strichartz estimate for the linearized equation. We use Duhamel's integral formula and the asymptotic behaviour of the linearized equation. Namely, we can write the solution of \eqref{linearized_NLS}, for $t_0=1$, as 
    \begin{align*}
        &z(t) = e^{i (t-1) \partial_{ss}} z(1) \pm ia^2 \int_1^t e^{i(t-\tau)\partial_{ss}} \frac{\overline{z(\tau) - e^{i(\tau - 1)\partial_{ss}}u_+}}{\tau^{1\pm 2ia^2}} d\tau \nonumber\\
        & \pm ia^2 \int_1^t e^{i(t-2\tau)\partial_{ss}} \frac{e^{i \partial_{ss}} \overline{u}_+}{\tau^{1 \pm 2ia^2}} d\tau = (i) + (ii) + (iii).
    \end{align*}
    Recalling the Strichartz estimates \eqref{EstStrichartzSchrod} from the linear Schrödinger flow  we have that $(p,q)= (4,\infty)$ is a 1d admissible pair. This, along with the estimates from propositions \eqref{proposition_LinearPseudconfEq_FreeSchrod} and \eqref{proposition_LinearPseudconfEq_ScatteringState}, gives us the upper bound of the $L^4((1,\infty),L^{\infty}(\mathbb{R}))$ norm of (i) and (ii). The third term is more involved (see \cite{BV12}). 
\end{remark}
Bounds from \Cref{lemma2_GrowthFourier} and \Cref{proposition_LinearPseudconfEq_FreeSchrod} suggest a possible growth of the zero Fourier mode. In fact, this can be formally seen as follows, for both cases focusing and defocusing. 
Take $w$ solution of \eqref{perturbedNLS_1}, and denote
\begin{align*}
    y_1(t) = \text{Re} \int_{-\infty}^{\infty} w(s,t) ds, \quad y_2(t) = \text{Im} \int_{-\infty}^{\infty} w(s,t) ds,
\end{align*}
which are the real and imaginary parts of the zero Fourier mode of $w(s,t)$. Thus, integrating with respect to $s$ and separating real and imaginary parts in equation \eqref{perturbedNLS_1},
\begin{align*}
    iy'_1(t) - y'_2(t) = \mp \frac{2a^2}{t} y_1(t).
\end{align*}
Hence $y_1(t)$ is constant, so we can state that $y_1(t) = y_1(1)$ and, therefore, $y_2(t) = \pm 2a^2y_1(1) \log(t)$, what yields the enunciated growth. In order to give a rigorous prove, it is enough to assume decay of $w(s,1)$ (see Appendix B of \cite{BV12}).
\\
In the focusing case, it is also possible to give conditions on $w(s,1)$ such that the function $Z_{\xi}(t)$ from \eqref{YxiZxi} has a logarithmic growth in $t$. 
 In fact, if we take $Y_{\xi}(\Tilde{t}_0) = Z_{\xi}(\Tilde{t}_0) = 1$ and $0 < \Tilde{t}_0 \leq t \leq \min\{a^2,1/2\}$ then, from \eqref{DerYxiDerZxi}, $Z_{\xi}'$ will be positive as long as $-1 + \frac{2a^2}{t} > 0$. In that case, $Y_{\xi}(t) > 1$ and $Z_{\xi}(t) > -1 + \frac{2a^2}{t}$. Using that $Y'_{\xi}(t) = Z_{\xi}(t)$, then $Z_{\xi}(t)$ has a logarithmic growth.
\\
We will see later that under generic conditions on $v(1)$, the zero-Fourier mode of $v(t)$ also grows logarithmically in $t$, in both situations focusing and defocusing.
\\\\
Once the solution of the linearized equation is obtained, we can solve the nonlinear equation \eqref{NLSpconformal}. 
From now on, and for simplicity, we will take $t_0=1$.
\begin{theorem}
    Let $a > 0$, and $v(1) \in X_1$ small with respect to $a$. Then, there exists a unique global solution $v$ of \eqref{NLSpconformal} with
    \begin{align*}
        v - a \in Z = Y_1 \cap L^4((1,\infty),L^{\infty}(\mathbb{R})),
    \end{align*}
    $v(1)$ the initial datum at $t=1$, and
    \begin{align*}
        v = a + e^{\pm i a^2 \log(t)} u, \quad ||u||_Z \leq c(a) ||u(1)||_{X_1}.
    \end{align*}
    Moreover, $u$ scatters in $L^2$, i.e. there exists $f_+ \in L^2(\mathbb{R})$ such that
    \begin{align*}
        || u(t) - e^{i(1-t)\partial_{ss}} f_+ ||_{L^2(\mathbb{R})} \leq \frac{c(a,\delta)}{t^{\frac{1}{4}- \delta}} ||u(1)||_{X_1},
    \end{align*}
    for any $0 < \delta < \frac{1}{4}$. Finally, the asymptotic state $f_+$ satisfies, for all $\xi^2 \leq 1$ the estimate
    \begin{align*}
        |\hat{f}_+(\xi)| |\xi|^{2\delta} \leq c(a,\delta) ||u(1)||_{X_1}.
    \end{align*}
    \label{thm_Xt0Yt0}
\end{theorem}\noindent
As we said, the zero Fourier mode of $v(t)$ can explode at infinity, and for that we can argue as follows. First of all, as it can be seen in Lemma B.1 in\cite{BV12}, if $sv(1) \in L^2(\mathbb{R})$, then
\begin{align*}
    ||sv(t)||_{L^2(\mathbb{R})} \leq c(a,u(1)) t^{\Tilde{c}(a,u(1))}
\end{align*}
for some constants $c$ and $\Tilde{c}$.\\
As in the proof of \Cref{lemma2_GrowthFourier}, it is better in this case to work with the perturbation with form $w$, as the solution of
\begin{align}
    iw_t + w_{ss} = \mp \frac{1}{t}(|a+w|^2 - a^2)(a+w),
    \label{NLSperturbed_2}
\end{align}
recalling that $v = a + w = a + e^{\pm i a^2 \log(t)}u$. Consider the conservation law
\begin{align*}
    \partial_t \int_{\mathbb{R}}(|a+w|^2 - a^2)ds = 0,
\end{align*}
obtained in the usual way. Integrating with respect to $s$ and $t$ in \eqref{NLSperturbed_2}, we get
\begin{align*}
    &\int_{t_0}^t \int_{\mathbb{R}} (iw_t(\tau) + w_{ss}(\tau)) ds
    d\tau = \mp \int_{t_0}^t \int_{\mathbb{R}}  (|a+w(\tau)|^2 - a^2)(a+w(\tau)) ds \frac{d\tau}{\tau} \nonumber\\
    &\Longrightarrow \int_{\mathbb{R}} w(t) ds = \int_{\mathbb{R}} w(t_0) ds \pm i \int_{t_0}^t \int_{\mathbb{R}} (|w(\tau) + a|^2 -a ^2)(w(\tau)+a)ds \frac{d\tau}{\tau} \nonumber\\
    &= \int_{\mathbb{R}} w(t_0)ds \pm ia \log(t/t_0) \int_{\mathbb{R}} (|w(\tau) + a|^2 -a ^2) ds  \nonumber\\
    &\pm i \int_{t_0}^t \int_{\mathbb{R}} (|w(\tau)|^2 + 2a\text{Re}w(\tau))w(\tau) ds \frac{d\tau}{\tau},
\end{align*}
where we have integrated by parts in $\tau$, and applied the previous conservation law. From Strichartz estimates \eqref{EstStrichartzSchrod} 
we know that the cubic term appearing in the last term is integrable in time, so we have to take care only of the quadratic one. For that, we get
\begin{align*}
    \left| \int_{t_0}^t \int_{\mathbb{R}} 2a\text{Re}w(\tau)w(\tau) ds  \frac{d\tau}{\tau}\right| \leq c(a) ||u(t_0)||_{X_{t_0}}^2 \log(t/t_0).
\end{align*}
On the other hand, 
\begin{align*}
    \left|\int_{\mathbb{R}} (|w(t_0) - a|^2 - a^2) ds \right| \geq 2a \left|\int_{\mathbb{R}} \text{Re}w(t_0) ds\right| - \int_{\mathbb{R}} |w(t_0)|^2 ds.
\end{align*}
Therefore, it is enough to take $w(t_0)$ such that
\begin{align*}
    c(a,t_0) \left( ||w(t_0)||_{L^2(\mathbb{R})}^2 + ||\widehat{w}(t_0)||^2_{L^{\infty}(\xi^2 \leq \frac{1}{t_0})} \right) < \left|\int_{\mathbb{R}} \text{Re} w(t_0) ds\right|.
\end{align*}
Lastly, we come back to the filament function $\psi$ and the family of curves $\chi$ which are solution of the BF that are constructed from the result obtained in \Cref{thm_Xt0Yt0}.
Note that $v$ and $\psi$ are related through the following conformal transformation:
\begin{align*}
    \psi(s,t) = \frac{e^{is^2/4t}}{\sqrt{t/2}} \overline{v}\left( \frac{\sqrt{2}s}{t} , \frac{2}{t}  \right).
\end{align*}
Therefore, 
\begin{align*}
    \psi(s,1) = a\sqrt{2} e^{is^2} + \psi_1(s)
\end{align*}
where $\psi_1(s) = a \sqrt{2} + e^{\pm i a^2 \log(2)} u_1(\sqrt{2}s)$. For simplicity, we impose $\psi_1 \in L^1(\mathbb{R}) \cap L^2(\mathbb{R})$, in order to fulfill the hypothesis $\hat{z}(\xi) \in L^{\infty}(\xi^2 \leq 1) \cap L^2(\mathbb{R})$, needed in \Cref{thm_Xt0Yt0}. From the asymptotic behaviour of the solutions of the free Schrödinger equation we obtain, from \Cref{thm_Xt0Yt0}, the following result.
\begin{theorem}
    Let $a>0$ and $\psi_1 \in L^1(\mathbb{R}) \cap L^2(\mathbb{R})$ small with respect to $a$. Then there exists a unique solution $\psi$ of
    \begin{align*}
        &\psi_t = i \left( \partial_{ss}\psi + \frac{1}{2} \left( |\psi|^2 - \frac{a^2}{t} \right) \psi \right), \quad 0 < t < 1, \nonumber\\
        &\psi(1) = \psi_1 + ae^{is^2/4}
    \end{align*}
    such that
    \begin{align*}
        \psi(s,t) - \frac{a}{\sqrt{t}} e^{is^2/4t} \in L^{\infty}((0,1),L^2(\mathbb{R})) \cap L^4((0,1),L^{\infty}(\mathbb{R})).
    \end{align*}
    Moreover, there exists $\psi_+ \in L^2(\mathbb{R})$ such that 
    \begin{align*}
        \left| \left| \psi(t) - \frac{a}{\sqrt{t}} e^{is^2/4t} - e^{\pm i a^2 \log(t)}  e^{it \partial_{ss}} \psi_+ \right|\right|_{{L^2(\mathbb{R})}} \leq c(a,\delta) t^{\frac{1}{4} - \delta} ||\psi_1||_{L^1(\mathbb{R}) \cap L^2(\mathbb{R})},
    \end{align*}
    for any $0 < \delta < 1/4$, and for $|s| \leq 2$
    \begin{align*}
        |s|^{2 \delta} |\psi_+(s)| \leq c(a,\delta) ||\psi_1||_{L^1(\mathbb{R}) \cap L^2(\mathbb{R})}.
    \end{align*}
\end{theorem}\noindent
It can also be proved that if $u_1$ is regular (recall \eqref{eq:u1}), so is $u(t)$ as given in theorem \eqref{thm_Xt0Yt0}. So in particular $u(t)$ is uniformly bounded by the size of $u_1$ and its derivatives. Therefore, if $u_1$ is small enough with respect to $a$, then
\begin{align}
    \frac{a}{2\sqrt{t}} \psi(s,t) \leq \frac{3a}{2t}.
    \label{PsiSingular}
\end{align}
Consequently, $\psi(s,t)$ becomes singular as $t$ goes to zero. Also, we can use the Frenet system to construct a regular solution of the BF equation \eqref{eq:CPchi} for $0<t<1$, and the corresponding Frenet frame that will become more singular as $t$ approaches to zero (see \cite{NSVZ} and appendix of \cite{BV082}). 
\\
In the focusing case, we also have that the binormal vector is unitary. Hence, from the inequalities \eqref{PsiSingular} we conclude the existence of $\chi_0(s)$ given by 
\begin{align}
    \chi_0(s) = \chi(s,1) - \int_0^1 c(s,\tau) b(s,\tau) ds
    \label{chi0singular}
\end{align}
with $c = |\psi|$ the curvature. As a conclusion, we get the following result.
\begin{theorem}
    Let $a>0$ and $\chi_1(s)$ a regular curve with $(c_1,\tau_1)$ the corresponding curvature and torsion. We define
    \begin{align*}
        \psi_1(s) = c_1(s) e^{i \int_0^s \tau_1(s')ds'}, \quad u_1(s) = e^{-is^2/4} \psi_1(s) - a,
    \end{align*}
    and assume that $u_1 \in L^1(\mathbb{R}) \cap H^3(\mathbb{R})$, and that it is small with respect to $a$.
    \\
    Then, there exists a unique regular solution $\chi(s,t)$ of the BF equation \eqref{eq:CPchi} for $0<t<1$ with $\chi(s,1)=\chi_1(s)$. Moreover, its curvature and torsion $(c,\tau)$ satisfy
    \begin{align*}
        \left| c(s,t) - \frac{a}{\sqrt{t}} \right| \leq \frac{c(u_1)}{t^{1/4+}}, \quad \left| \tau(s,t) - \frac{s}{2t} \right| \leq \frac{c(u_1)}{t^{3/4+}},
    \end{align*}
    and taking $\chi_0(s)$ as in \eqref{chi0singular}, then
    \begin{align*}
        |\chi(s,t) - \chi_0(s)| \leq c(u_1) \sqrt{t}.
    \end{align*}
\end{theorem}

\section{Several corners}\noindent\label{sec:Several Corners}
In the previous sections we described the evolution of a solution with a corner and its stability. At this point, it is natural to ask the following questions: what if we have two corners instead of one? What if we have, in general, more than one corner? Or even more, what if we have infinitely many corners? 
\\
The answer in all these cases is the same: the difficulty to answer the question does not change between one, finitely many or infinitely many corners, as long as: 
\begin{enumerate}
    \item The distance between corners remain the same.
    \item The corresponding polygonal line becomes flat at infinity.
\end{enumerate}
One of the key procedures developed in these lecture notes is based on the link between the BF and the 1d cubic NLS. In particular, we used this approach with the self-similar solutions of the BF, constructed from the solutions
\begin{equation}\label{FtalSolution_AtNLS}
    u_{a}(s,t) = a \frac{e^{is^2/4t}}{\sqrt{it}} = a e^{it \Delta} \delta_0(s)
\end{equation}
of the 1d cubic NLS renormalized with the inclusion of the real potential $A(t) = -a^2/t$, i.e. through the Cauchy problem
\[
\begin{cases}
    &i \partial_t u + \Delta u \pm \frac{1}{2} (|u|^2 + A(t)) u = 0, \nonumber\\
    &u(0) = a \delta_0,
    \label{eq:AtNLS}
\end{cases}
\]
with a Dirac mass at initial time. As shown in \cite{KPV01} and briefly explained above, through the Galilean invariance one can prove that if exists a unique solution of the 1d cubic NLS \eqref{NLS}, then it should be for positive times of the form
\begin{align*}
    e^{\mp i|a|^2\log(\sqrt{t})} e^{is^2/4t} \frac{a}{ \sqrt{it}},
\end{align*}
 so that initial datum is not recovered. This justifies the need of $A(t)$. Recall that adding $A(t)$ is a natural geometric normalization since the BF solution from the NLS solution $u(s,t)$ is the same as the one constructed from $e^{i \phi(t)} u(s,t)$, for some real-valued function $\phi$. 
 The BF self-similar solutions are of the type $\chi(s,t) = \sqrt{t}G(s/\sqrt{t})$, yielding a $1$-parameter family $\{\chi_a : a\geq 0\}$ of curves, each of them with curvature $c_a(s,t) = a/\sqrt{t}$ and torsion $\tau_a(s,t) = s/2t$. 
 The rigorous proof of existence of a trace at $t=0$ was given in \cite{GRV03} and briefly exposed here in \Cref{section:WaveOperator}. In particular it was shown that $\chi_a(0)$ is a broken line with one corner having an angle $\theta$ such that (recall \Cref{them_NLStoBF})
\begin{align}
    \sin\left( \frac{\theta}{2} \right) = e^{-\frac{\pi a^2}{2}}.
    \label{a_theta_relation}
\end{align}

A planar regular polygon with $M$ sides is expected to evolve through the BF to skew polygons with $Mq$ sides at rational multiples of $1/2\pi$, i.e. times $t_{p,q}=\frac{p}{2\pi q}$  for $p,q \in \mathbb{Z}$ and odd $q$ (see numerical simulations in \cite{JS}, \cite{GDV}, \cite{GG99}, and \cite{GG95}, where the integration of the Frenet equations at rational times $t_{p,q}$ is also done). We move now in the general framework of initial data being polygonal lines, and we aim to summarize some results which are important steps forward to fill the gap between the case of one corner and the problem of closed polygons. In this section, the first three subsections will follow mainly \cite{BV18}, while the last subsection is based on \cite{DLHV14}.
\subsection{Polygonal lines evolution}
In this section we introduce the main tools involved in the several corners problem. 
We start with a necessary definition of certain weighted summation spaces, and a theorem that condensates some relevant information about the evolution of polygonal lines through the BF (see \cite{BV18}). As we have done so far, we consider curves parametrized by arc length.
\begin{definition}
    We define $l^{p,s}$ as the normed space of sequences $\{a_k\}$ such that
    \begin{align*}
        || \{a_k\} ||_{l^{p,s}} < \infty,
    \end{align*}
    where
    \begin{align*}
        || \{a_k\} ||^p_{l^{p,s}} := \sum_{k \in \mathbb{Z}} (1 + |k|)^{ps} |a_k|^p.
    \end{align*}
\end{definition}
\begin{theorem}
    Let $\chi_0$ be a polygonal line with corners at $s \in \mathbb{Z}$, with the sequence of angles $\theta_n \in (0,\pi)$ such that the sequence defined by
    \begin{align*}
        a_n = \sqrt{-\frac{2}{\pi} \left( \sin\left( \frac{\theta_n}{2} \right) \right)} \quad (\Longleftrightarrow \sin\left( \frac{\theta_n}{2} \right) = e^{-\pi a_n^2/2})
    \end{align*}
    belongs to $l^{2,3}$. Then, there exists $\chi(t)$, smooth solution of the BF \eqref{eq:CPchi} on $t \neq 0$ and solution of \eqref{eq:CPchi} in the weak sense on $\mathbb{R}$, with 
    
    \begin{align*}
        |\chi(s,t) - \chi_0(s)| \leq C \sqrt{t}, \quad \forall s \in \mathbb{R}, |t|  \leq 1.
    \end{align*}
    \label{thm:SolsNLSSeveralCorners}
\end{theorem}
\begin{proof}
    We provide a sketch of the proof. As for the one corner case, we use the Hasimoto trasformation to go from the BF equation to the 1d cubic NLS given by
    \[
    i \partial_t u + \partial_{ss} u \pm \frac12 \big ( \lvert u \rvert^2 - \frac Mt \big),
    \]
    with initial data
    \begin{align}
        \sum_{k \in \mathbb{Z}} a_k \delta_k,
        \label{GralICNLSSeveralCorners}
    \end{align}
    with $a_k$ are complex numbers defined in a suitable way from the expression of $\chi_0$. As we will see in the next section, (see \Cref{thm:SeveralDiracMassesNLS}), we can get a smooth solution $u(t)$ for $t>0$. From this smooth solution with times values on $(0,\infty)$, we construct a smooth solution $\chi(t)$ of the BF on $(0,\infty)$, and from here one can prove that it has a limit $\chi(0)$ at $t=0$. All in all, our goal is to prove that, modulo a translation and a rotation, $\chi(0)$ is $\chi_0$. To prove it, one should follow several steps: 
    \begin{enumerate}
        \item Prove that the tangent vector has a limit at $t=0$.
        \item Prove that this limit is piecewise constant, so $\chi(0)$ is a segment for $s \in (n,n+1)$, for any $n \in \mathbb{Z}$.
        \item Prove, by studying the frame of the curve through paths of self-similar variables, that $\chi(0)$ presents corners at the same locations as $\chi_0$, of same angles as $\chi_0$.
        \item Recover the characteristic geometric elements of $\chi_0$ through the analysis for modulated normal vectors
        \begin{align*}
            \Tilde{N}(s,t) = e^{i \sum_{j \neq s} |a_j|^2 \log \frac{s-j}{\sqrt{t}}}N(s,t).
        \end{align*}
        \item Recover $\chi_0$ modulo a translation and a rotation. These translation and rotation applied to $\chi(t)$ yields the desired solution of the BF for $t>0$, with limit $\chi_0$ at $t=0$. Uniqueness holds for the class of curves having filament functions of type \eqref{FilamentFunctionsUnique_SeveralCorners} below.
        \item Using the recipe just described to construct the evolution of a polygonal line for $t > 0$, one can extend $\chi(t)$ to negative times by using the time reversibility of the  BF equation.
    \end{enumerate}
\end{proof}
Although we still do not have the tools to deepen in it properly, we provide two remarks which extend these notions. 
\begin{remark}
    Under suitable conditions on the initial data $\chi_0$, we can have intermittent behaviour of the evolution. More concretely, at rational multiple times $t_{p,q} = p/2\pi q$ the curvature of $\chi(t)$ displays concentrations near the locations $s$ such that $s \in \frac{1}{q}\mathbb{Z}$, and $\chi(t)$ is almost a straight segment in between.
\end{remark}
\begin{remark}
    There is an interesting difference between the several corners and the one corner cases. For the latter, the trajectory in time of the corner at $(s,t) = (0,0)$ of a self-similar solution, $\chi_a(t,0)$, is given by a straight line for $t>0$, since the Frenet frame of $\chi_a(t)$ is constant at $s=0$. However, in \cite{BV18} (section 4.1) it is shown that for the evolution of a polygonal line with several corners, the trajectory of each corner, as $t$ goes to $0$, is a logarithmic spiral. That is, the presence of another corner on a nonclosed curve immediately creates a modification of the trajectory.
\end{remark}
\subsection{Cubic NLS on $\mathbb{R}$ with initial data given by several Dirac masses}
As we saw before, the problem of the description of the evolution of a polygonal line is translated in a study of a particular 1d cubic NLS with sum of Dirac masses as initial data. In this section we describe, in a not exhaustive way, some of the results corcening the above mentioned IVP. We start by considering the standard cubic NLS equation
\begin{align*}
    i \partial_t u + \Delta u \pm \frac{1}{2} |u|^2u = 0,
\end{align*}
with a Dirac mass as an initial data. We recall that the Dirac delta distribution is almost in $\dot{H}^{-1/2}$, in the sense that $\delta \in \dot H^{- \frac12 - \epsilon}(\R)$, $\forall \epsilon >0$. 
As already stated, $s_c= - \frac 12$ is the critical value for local wellposedness in $\dot H^s$ (\Cref{sec:intro1dcubicnls}). It implies that, if $u(0) = a \delta_0$, then the equation is ill-posed. However, by a simple change of phase we can avoid that issue, by considering the modified cubic NLS \eqref{eq:Cauchypb_NLSa2t} and the solution $u_a$ given by \eqref{FtalSolution_AtNLS}.
However, even with this geometric renormalization the problem is still ill-posed, in the sense that small regular perturbations of $u_a(t)$ at $t=1$ were proved in \cite{BV12_1} to behave near $t=0$ as $u_a(t) + e^{i \log (t)} f(s)$, for some $f \in H^1$, i.e. there is a loss of phase as $t$ goes to zero.
\\
As studied in \cite{BV13} and in previous sections (see \eqref{section:WaveOperator}), there is a geometric way of continuing the BF solution after the singularity, in the shape of a corner, is created. 
Our goal is now to extend this construction to the case of a polygonal line, which can be viewed as a rough perturbation of a broken line with one corner. There is no need for the line to be planar and infinitely many corners are permitted. In this case, new problems concerning the phase loss appear at the NLS and frame level, and again the characterization of the self-similar solutions plays a crucial role. 
\\
For these reasons, in the following we will consider initial data like \eqref{GralICNLSSeveralCorners}, with coefficients in $l^{p,s}$ spaces. Combinations of Dirac masses as initial data for the NLS with nonlinearity $|u|^{p-1}u$, with subcritical regime $p<3$, was considered in \cite{Ki06}, whose ansatz is a key element of the following analysis. One of its underlying ideas is that the nonlinearity creates new harmonics, considering a determining correspondence between corners' angles $\theta_k$ and the coefficients $a_k$\footnote{This partially turns the analytical problem of several Dirac masses into a geometric one, distinguishing between the two cases:
\begin{align*}
    \sum_{k \in \mathbb{Z}} |a_k| < \infty, \quad \sum_{k \in \mathbb{Z}} \langle k \rangle^{2s} |a_k|^2 < \infty,
\end{align*}
with $s > 1/2$, the second easier to deal with than the first one due to the impossibility in the $l^1$ case of using conservation laws of the NLS. A more concrete form of the correspondence between $\theta_k$ and $a_k$ was given in theorem \eqref{thm:SolsNLSSeveralCorners}.}. It was shown in the aforementioned paper that such a Cauchy problem admits a unique solution of the form
\begin{align}
    u(s,t) = \sum_{k \in \mathbb{Z}} A_k(t) e^{it\Delta} \delta_k(s),
    \label{eq:kita_ansatz}
\end{align}
where $\{A_k\} \in C([0,T],l^{2,1}) \cap C^1((0,T],l^{2,1})$. Indeed, as the nonlinear power approaches the critical cubic power, things look more singular. Even though, we deal with the same ansatz, valid for a naturally renormalized cubic equation.
\\
A useful way of characterizing these functions comes from the Fourier transform of the initial data \eqref{GralICNLSSeveralCorners}, which yields
\begin{align}
    \widehat{u(0)}(\xi) = \sum_{k \in \mathbb{Z}} a_k e^{-ik\xi},
    \label{eq:decomp_fourier_id}
\end{align}
so in particular $\widehat{u(0)}$ is $2\pi$-periodic, and $\{a_k\} \in l^{2,s}$ turns into $\widehat{u(0)} \in H^s(0,2\pi)$. Conversely, since every $2\pi$-periodic function can be decomposed as \eqref{eq:decomp_fourier_id}, then it also represents the Fourier transform in $\mathbb{R}$ of a combination of Dirac masses as \eqref{eq:kita_ansatz}. We denote
\begin{align*}
    H^s_{pF} &:= \{ u \in \mathcal{S}'(\mathbb{R}): \hat{u}(\xi + 2\pi) = \hat{u}(\xi), \hspace{2mm} \hat{u} \in H^s(0,2\pi) \} \nonumber\\
    &\subset \{ u \in \mathcal{S}'(\mathbb{R}): \{||\hat{u}||_{H^s(2\pi j,2\pi(j+1))}\}_j \in l^{\infty} \},
\end{align*}
and
\begin{align*}
    ||u||_{H^s_{pF}} = ||\hat{u}||_{H^s(0,2\pi)}.
\end{align*}
This establishes, for the study of these combinations of Dirac masses, a hybrid model between the real line and the periodic case. Now we state a result concerning the existence of solutions for initial data in $H^s_{pF}$.
\begin{theorem}
    Let $s > \frac{1}{2}$, $0 < \gamma < 1$ and $\{a_k\} \in l^{2,s}$. We consider the 1d cubic NLS equation of the form
    \begin{align}
        i \partial_t u + \Delta u \pm \frac{1}{2} (|u|^2 - \frac{M}{t})u = 0
        \label{eq:NLSwithMSeveralCorners}
    \end{align}
    where $M = \sum_{k \in \mathbb{Z}} |a_k|^2$. There exists $T>0$ and a unique solution on $(0,T)$ of the form
    \begin{align}
        u(s,t) = \sum_{k \in \mathbb{Z}} e^{\mp i|a_k|^2 \log \sqrt{t}} (a_k + R_k(t)) e^{it\Delta} \delta_k(s)
    \label{FilamentFunctionsUnique_SeveralCorners}
    \end{align}
    with
    \begin{align*}
        \sup_{0<t<T} t^{-\gamma} || \{R_k(t)\} ||_{l^{2,s}} + t||\{\partial_tR_k(t)\}||_{l^{2,s}} < C
    \end{align*}
    for some $C>0$. Moreover, considering as initial data a finite sum of $N$ Dirac masses like in \eqref{GralICNLSSeveralCorners} with coefficients of equal modulus $|a_k| = a$, and equation \eqref{eq:NLSwithMSeveralCorners} renormalized with $M = (N - \frac{1}{2})a^2$, we have a unique solution
    \begin{align*}
        u(t) = e^{it\Delta} u(0) \pm i e^{it\Delta} \int_0^t e^{-i \tau \Delta} \left( \left( |u(\tau)|^2 - \frac{M}{\tau} \right) u(\tau) \right) \frac{d\tau}{2},
    \end{align*}
    such that $\widehat{e^{-it \Delta} u(t)} \in C^1((-T,T),H^2(0,2\pi))$ with
    \begin{align*}
        || e^{-it\Delta}u(t) - u(0) ||_{H^s_{pF}} \leq Ct^{\gamma}, \quad \forall t \in (-T,T).
    \end{align*}
    Moreover, if $s \geq 1$, then the solution is global in time.
    \label{thm:SeveralDiracMassesNLS}
\end{theorem}
\begin{proof}
     We will picture the fixed point framework for the corresponding dynamical system. Missing details, arguments for the required estimates needed to apply a fixed point argument, as well as its global extension in time, can be found in \cite{BV18}.
    \\
    We denote $\mathcal{N}(u) = u|u|^2/2$, and plug the Kita ansatz \eqref{eq:kita_ansatz} into the equation \eqref{eq:NLSwithMSeveralCorners}:
    \begin{align}
        \sum_{k \in \mathbb{Z}} i\partial_tA_k(t) e^{it\Delta}\delta_k = \mathcal{N}(u) - \frac{M}{t} u = \mathcal{N}(\sum_{j \in \mathbb{Z}} A_j(t) e^{it\Delta}\delta_j) - \frac{M}{t} (\sum_{k \in \mathbb{Z}} A_k(t) e^{it\Delta}\delta_k),
        \label{eq:SumDynamicalSystem}
    \end{align}
    where the Laplacian term dissapears. We choose for simplicity the $-$ sign case (for the $+$ sign case, the procedure is the same). The family of elements
    \begin{align*}
        e^{it\Delta} \delta_k(x) = \frac{e^{i(x-k)^2/4t}}{\sqrt{it}}
    \end{align*}
    is an orthonormal family in $L^2(0,t)$, so by taking its usual inner product with $e^{it\Delta} \delta_k(x)$ of both sides of \eqref{eq:SumDynamicalSystem}, we get
    \begin{align*}
        i \partial_t A_k(t) = t \int_0^t \mathcal{N}(A_j(t)\sum_{j\in \mathbb{Z}} \frac{e^{i(x-j)^2/4t}}{\sqrt{it}}) \frac{e^{-i(x-k)^2/4t}}{\overline{\sqrt{it}}} dx - \frac{M}{t} A_k(t).
    \end{align*}
    Thanks to a straightforward application of Cauchy-Schwarz inequality, we have that
    \begin{align*}
        \sum_{j \in \mathbb{Z}} |A_j(t)| \leq \left( \sum_{j \in \mathbb{Z}} |A_j(t)| \langle j \rangle^{2s} \right)^{1/2} \left( \sum_{j \in \mathbb{Z}} \langle j \rangle^{-2s} \right)^{1/2},
    \end{align*}
    so that when $s > 1/2$, we have $\{A_j\} \in l^{2,s} \subset l^1$. Therefore, we can develop the cubic power in the exponent of the previous integral, considering that such exponents fulfill
    \begin{align*}
        &(x-j_2)^2 - (x-j_1)^2 - (x-j_3)^2 + (x-k)^2 = j_2^2 - j_1^2 - j_3^2 + k^2 \nonumber\\
        &+ 2xj_1 + 2xj_3 - 2xj_2 - -2xk = j_2^2 - j_1^2 - j_3^2 + k^2 + 2x(j_1 - j_2 + j_3 - k),
    \end{align*} 
    so that, when integrating, we restrict the summation to $k - j_1 + j_2 - j_3 = 0$. Thus, 
    \begin{align}
        i \partial_t A_k(t) = \frac{1}{2t} \sum_{k - j_1 + j_2 - j_3 = 0} e^{-i\frac{k^2 - j_1^2 + j_2^2 - j_3^2}{4t}} A_{j_1}(t) \overline{A_{j_2}(t)} A_{j_3}(t) - \frac{M}{t} A_k(t).
        \label{eq:DS_NLSAk}
    \end{align}
    Consider an undefined sequence of real numbers $a(k)$, and observe that
    \begin{align*}
        &\partial_t \sum_{k \in \mathbb{Z}} a(k) |A_k(t)|^2 = 2 \sum_{k \in \mathbb{Z}} a(k) \text{Re}(\overline{A_k(t)} \partial_t A_k(t)) \nonumber\\
        &= 2  \sum_{k \in \mathbb{Z}} a(k) \text{Re}\{-\frac{i\overline{A_k(t)}}{2t}  \sum_{k - j_1 + j_2 - j_3 = 0} e^{-i\frac{k^2 - j_1^2 + j_2^2 - j_3^2}{4t}} A_{j_1}(t) \overline{A_{j_2}(t)} A_{j_3}(t) + \frac{iM}{t} |A_k(t)|^2 \} \nonumber\\
        &=2  \sum_{k \in \mathbb{Z}} a(k) \text{Im}\{\frac{\overline{A_k(t)}}{2t}  \sum_{k - j_1 + j_2 - j_3 = 0} e^{-i\frac{k^2 - j_1^2 + j_2^2 - j_3^2}{4t}} A_{j_1}(t) \overline{A_{j_2}(t)} A_{j_3}(t)\}  \nonumber\\
        &= \frac{1}{2ti}  \sum_{k \in \mathbb{Z}} a(k) \{ \sum_{k - j_1 + j_2 - j_3 = 0} e^{-i\frac{k^2 - j_1^2 + j_2^2 - j_3^2}{4t}} A_{j_1}(t) \overline{A_{j_2}(t)} A_{j_3}(t)\overline{A_k(t)} \nonumber\\
        &- \overline{\sum_{k - j_1 + j_2 - j_3 = 0} e^{-i\frac{k^2 - j_1^2 + j_2^2 - j_3^2}{4t}} A_{j_1}(t) \overline{A_{j_2}(t)} A_{j_3}(t)\overline{A_k(t)}}\}  \nonumber\\
        &= \frac{1}{2ti}\sum_{k - j_1 + j_2 - j_3 = 0} (a(k) - a(j_3))e^{-i\frac{k^2 - j_1^2 + j_2^2 - j_3^2}{4t}} A_{j_1}(t) \overline{A_{j_2}(t)} A_{j_3}(t)\overline{A_k(t)} \nonumber\\
        &= \frac{1}{4ti}\sum_{k - j_1 + j_2 - j_3 = 0} (a(k) - a(j_1) + a(j_2) - a(j_3))e^{-i\frac{k^2 - j_1^2 + j_2^2 - j_3^2}{4t}} A_{j_1}(t) \overline{A_{j_2}(t)} A_{j_3}(t)\overline{A_k(t)}.
    \end{align*}
    Therefore, if we take $a(k)=1$ for all $k \in \mathbb{Z}$, we obtain the conservation of the mass
    \begin{align}
        \sum_{k\in\mathbb{Z}} |A_k(t)|^2 = \sum_{k \in \mathbb{Z}} |A_k(0)|^2, \quad \forall t \in \mathbb{R};
        \label{eq:DS_ConservationMass}
    \end{align}
    and if we take $a(k) = k$, we obtain the conservation of momentum
    \begin{align*}
        \sum_{k\in\mathbb{Z}} k|A_k(t)|^2 = \sum_{k \in \mathbb{Z}} k|A_k(0)|^2, \quad \forall t \in \mathbb{R}.
    \end{align*}
    We split now the summation indices from \eqref{eq:DS_NLSAk} into the following sets, corresponding to non-resonant and resonant terms:
    \begin{align*}
        &NR_k = \{(j_1,j_2,j_3) \in \mathbb{Z}^3 : k - j_1 + j_2 - j_3 = 0, k^2 - j_1^2 + j_2^2 - j_3^2 \neq 0\}, \nonumber\\
        &R_k = \{(j_1,j_2,j_3) \in \mathbb{Z}^3 : k - j_1 + j_2 - j_3 = 0, k^2 - j_1^2 + j_2^2 - j_3^2 = 0\}.
    \end{align*}
    Regarding the set that corresponds to the resonant terms, note that
    \begin{align*}
        k - j_1 + j_2 - j_3 = 0 \Longrightarrow k^2 - j_1^2 + j_2^2 - j_3^2 = 2(j-j_1)(j_1-j_2).
    \end{align*}
    Since the last expression is zero if and only if $j=j_1$ or $j_1=j_2$, then
    \begin{align*}
        R_k = \{(k,j,j),(j,j,k) : j \in \mathbb{Z}\}.
    \end{align*}
    Putting all together, given certain $k \in \mathbb{Z}$, there are three kinds of terms in the sum \eqref{eq:DS_NLSAk}: regarding $R_k$, $(k,j,j)$ and $(j,j,k)$, and regarding $NR_k$, $(j_1,j_2,k-j_1+j_2)$ with $j_2 \neq j_1$ and $j_1 \neq k$. Thus,
    \begin{align*}
        &\sum_{k - j_1 + j_2 - j_3 = 0} e^{-i\frac{k^2 - j_1^2 + j_2^2 - j_3^2}{4t}} A_{j_1}(t) \overline{A_{j_2}(t)} A_{j_3}(t) \nonumber\\
        &=  \sum_{j_1\neq k,j_1 \neq j_2} e^{-i\frac{(k-j_1)(j_1-j_2)}{4t}}  A_{j_1}(t) \overline{A_{j_2}(t)} A_{k-j_1+j_2}(t) \nonumber\\
        &+ \sum_{j_1 \neq k}A_{j_1}(t) \overline{A_{j_1}(t)}A_k(t) + \sum_{j_2 \in \mathbb{Z}}A_{k}(t) \overline{A_{j_2}(t)}A_{j_2}(t),
    \end{align*}
    where we considered the case $(k,k,k)$ only in the last sum, which is the reason why, in the previous sum, we took $j_1 \neq k$, so that we do not repeat terms. Therefore, system \eqref{eq:DS_NLSAk} turns into
    \begin{align*}
        &i \partial_t A_k(t) = \frac{1}{2t} \sum_{(j_1,j_2,j_3) \in NR_k} e^{-i\frac{k^2 - j_1^2 + j_2^2 - j_3^2}{4t}}  A_{j_1}(t) \overline{A_{j_2}(t)} A_{j_3}(t) \nonumber\\
        &+ \frac{1}{2t} A_k(t) (2 \sum_{j \in \mathbb{Z}} |A_j(t)|^2 - |A_k(t)|^2 - 2M).
    \end{align*}
    From the conservation of mass \eqref{eq:DS_ConservationMass}, since $M = \sum_{j \in \mathbb{Z}}|\alpha_j|^2$, then finding a solution for $t > 0$ satisfying
    \begin{align}
        \lim_{t \rightarrow 0} |A_j(t)| = |\alpha_j|
        \label{eq:DS_ConvergenceAlphaj}
    \end{align}
    is equivalent to finding a solution for $t > 0$ satisfying \eqref{eq:DS_ConvergenceAlphaj}, for the following also `mass'-conserving system:
    \begin{align}
        i \partial_t A_k(t) = \frac{1}{2t}\sum_{(j_1,j_2,j_3) \in NR_k} e^{-i\frac{k^2 - j_1^2 + j_2^2 - j_3^2}{4t}}  A_{j_1}(t) \overline{A_{j_2}(t)} A_{j_3}(t) - \frac{1}{2t} |A_k(t)|^2  A_k(t).
       \label{eq:DS_NLSAkZeroMass} 
    \end{align}
    By doing a change of phase $A_k(t) = e^{i|\alpha_k|^2\log \sqrt{t}} \Tilde{A}_k(t)$, we get the equation 
    \begin{align}
        i \partial_t \Tilde{A}_k(t) = f_k(t) - \frac{1}{2t} (|\Tilde{A}_k(t)|^2 - |\alpha_k|^2) \Tilde{A}_k(t),
        \label{eq:DS_NLStildeAk}
    \end{align}
    where
    \begin{align*}
        &f_k(t) = \frac{1}{2t} \sum_{(j_1,j_2,j_3) \in NR_k} e^{-i\frac{k^2 - j_1^2 + j_2^2 - j_3^2}{4t}} e^{-i(|\alpha_k|^2 - |\alpha_{j_1}|^2 + |\alpha_{j_2}|^2  -|\alpha_{j_3}|^2)\log \sqrt{t}} \nonumber\\
        &\times \Tilde{A}_{j_1}(t) \overline{\Tilde{A}_{j_2}(t)} \Tilde{A}_{j_3}(t).
    \end{align*}
    Note that the solution of \eqref{eq:DS_NLStildeAk} fulfills
    \begin{align}
        \partial_t |A_k(t)|^2 = 2 \text{Im}(f_k(t) \overline{\Tilde{A}_k(t)})
        \label{eq:DS_NLStildeAkModSquare}
    \end{align}
    so obtaining a solution of \eqref{eq:DS_NLStildeAk} for $t>0$ with
    \begin{align}
        \lim_{t \rightarrow 0} |\Tilde{A}_k(t)| = |\alpha_k|,
        \label{eq:DS_tildeConvergenceAlphaj}
    \end{align}
    is equivalent to obtaining a solution for $t > 0$ also satisfying \eqref{eq:DS_tildeConvergenceAlphaj}, for the following equation, which also fulfills \eqref{eq:DS_NLStildeAkModSquare}:
    \begin{align*}
        i \partial_t \Tilde{A}_k(t) = f_k(t) - \frac{1}{2t} \int_0^t 2 \text{Im}(f_k(\tau) \overline{\Tilde{A}_k(\tau)}) d\tau \Tilde{A}_k(t).
    \end{align*}
    Recall that we expect solutions to be of the form $A_k(t) = e^{i|\alpha_k|^2 \log\sqrt{t}}(\alpha_k + R_k(t))$, with $\{R_k\}$ in the space 
    \begin{align*}
        &X^{\gamma} := \nonumber\\
        &:= \{\{f_k\} \in C^1((0,T),l^{2,s}), || \{t^{-\gamma} f_k(t)\}||_{L^{\infty}((0,T),l^{2,s})} + || \{t \partial_t f_k(t)\}||_{L^{\infty}((0,T),l^{2,s})}  < \infty \},
    \end{align*}
    with $T$ to be especified. To prove the theorem, we shall show that we have a contraction on a suitable ball of size $\delta$ of the space $X^{\gamma}$ for the operator $\Phi$, which sends $\{R_k\}$ into
    \begin{align*}
        \Phi(\{R_k\}) = \{\Phi_k(\{R_j\})\}
    \end{align*}
    with
    \begin{align*}
        \Phi_k(\{R_j\})(t) = i\int_0^t g_k(\tau) d\tau - i\int_0^t \int_0^{\tau} \text{Im}(g_k(s) \overline{(\alpha_k + R_k(s))}) ds (\alpha_k + R_k(\tau)) \frac{d\tau}{\tau},
    \end{align*}
    where
    \begin{align*}
        &g_k(t) = \frac{1}{2t} \sum_{(j_1,j_2,j_3) \in NR_k} e^{-i\frac{k^2 - j_1^2 + j_2^2 - j_3^2}{4t}} e^{-i w_{k,j_1,j_2,j_3} \log \sqrt{t}} (\alpha_{j_1} + R_{j_1}(t)))\overline{(\alpha_{j_2} + R_{j_2}(t)))} \nonumber\\
        & \times (\alpha_{j_3} + R_{j_3}(t))),
    \end{align*}
    where
    \begin{align*}
        w_{k,j_1,j_2,j_3} = |\alpha_k|^2 - |\alpha_{j_1}|^2 + |\alpha_{j_2}|^2  - |\alpha_{j_3}|^2.
    \end{align*}
    To finish, note that in the case of $N$ Dirac masses with coefficients $|\alpha_k| = a$ and equation \eqref{eq:NLSwithMSeveralCorners} with $M=(N-\frac{1}{2})a^2$, then, instead of \eqref{eq:DS_NLSAkZeroMass}, we get the system
    \begin{align*}
        &i \partial_t A_k(t) = \frac{1}{2t}\sum_{(j_1,j_2,j_3) \in NR_k} e^{-i\frac{k^2 - j_1^2 + j_2^2 - j_3^2}{4t}}  A_{j_1}(t) \overline{A_{j_2}(t)} A_{j_3}(t)  \nonumber\\
        &- \frac{1}{2t} (|A_k(t)|^2 - |\alpha_k|^2) A_k(t).
    \end{align*}
    Hence, we can write $A_k(t) = \alpha_k + R_k(t)$, and the same fixed point argument works for $\{R_k\}$.
\end{proof}
	\subsection{Talbot effect}\noindent
	We focus now on the Talbot effect for both the linear and nonlinear Schrödinger flows on a singular setting, hybrid between $\mathbb{T}$ and $\mathbb{R}$. The case of the torus and the interplay with optics have been largely studied, for instance in \cite{Be96}, \cite{KP99}, \cite{ET13}, and \cite{BK96}, respectively. In this section devoted to the Talbot effect, we follow \cite{BV18}, including the proofs of both Propositions \eqref{prop:TalbotLinear} and \eqref{prop:TalbotNonLinear}.
    \\
    From \eqref{eq:NLSwithMSeveralCorners}, we know that the corresponding solution $u(t)$ with initial data given by \eqref{GralICNLSSeveralCorners} with $|\alpha_k|=a$, behaves like $e^{i t \Delta} u(0)$ for small times. This linear evolution displays a Talbot effect.
	\begin{proposition}
		Let $p \in \mathbb{N}$ and $u_0$ and $\hat{u}_0$ $2\pi$-periodic. For all $t_{p,q} = \frac{p}{2 \pi q}$ with $q$ odd and for all $ x \in \mathbb{R}$ we have
		\begin{align*}
			e^{i t_{p,q} \Delta}u_0(x) = \frac{1}{\sqrt{q}} \int_0^{2 \pi} \hat{u}_0(\xi) e^{-i t_{p,q} \xi^2 + ix\xi} \sum_{l \in \mathbb{Z}} \sum_{m = 0}^{q-1} e^{i\theta_{m,p,q}} \delta(x - 2t_{p,q}\xi - l - \frac{m}{q}) d\xi
		\end{align*}
	for some $\theta_{m,p,q} \in \mathbb{R}$. We suppose moreover $\hat{u}_0$ is located modulo $2 \pi$ only in a neighborhood of zero of radius less than $\eta \frac{\pi}{p}$ with $\eta \in (0,1)$. For a given $x \in \mathbb{R}$, we define
	\begin{align*}
		\xi_x := \frac{\pi q}{p} \text{dist}(x,\frac{1}{q} \mathbb{Z}) \in [0,\frac{\pi}{p}).
	\end{align*}
	Then there exists $\theta_{x,p,q} \in          \mathbb{R}$ such that
	\begin{align*}
		e^{i t_{p,q} \Delta} u_0(x) = \frac{1}{\sqrt{q}} \hat{u}_0(\xi_x) e^{- i t_{p,q} \xi_x^2 + ix\xi_x + i\theta_{x,p,q}}.
	\end{align*}
	In particular,
	\begin{align*}
		|e^{i t_{p,q} \Delta} u_0(l + \frac{m}{q})| = |e^{i t_{p,q} \Delta} u_0(0)|,
	\end{align*}
	and if $x$ is at distance larger than $\eta/q$ from  $\frac{1}{q} \mathbb{Z}$ then $e^{i t_{p,q} \Delta} u_0(x)$ vanishes. Moreover, the solution can concentrate near $\frac{1}{q} \mathbb{Z}$ in the sense that there is a family of initial data $u_0^{\lambda} = \sum_{k \in \mathbb{Z}} \alpha_k^{\lambda} \delta_k$ and $C>0$ such that
	\begin{align*}
		\left|  \frac{e^{i t_{p,q} \Delta} u_0(0)}{e^{i t_{p,q} \Delta} \alpha_0^{\lambda} \delta_0(x)} \right| = \frac{C}{q} \lambda \rightarrow \infty, \quad \lambda \rightarrow \infty.
	\end{align*}
	\label{prop:TalbotLinear}
	\end{proposition}
	Note that the initial datum $u_0 = \sum_{k \in \mathbb{Z}} \delta_k$ fulfills the hypothesis in Proposition \eqref{prop:TalbotLinear}, in terms of $2\pi$-periodicity in Fourier and localization in Fourier, since $u_0 = \hat{u}_0$, where we considered Poisson summation formula. Thus, $e^{it\Delta}u_0 = 0$ for $x \notin \frac{1}{q} \mathbb{Z}$, and it is a Dirac mass otherwise, which is a classical Talbot effect. However, the non-summation of its coefficients makes these data not to enter in the nonlinear theorem framework. On the other hand, the result \eqref{prop:TalbotNonLinear} shows a Talbot effect at the nonlinear level too. 
	\begin{proposition}
		Let $p \in \mathbb{N}$ and $u_0$ be such that $\hat{u}_0$ is a $2\pi$-periodic, located modulo $2\pi$ in a neighborhood of zero of radius less than $\eta \frac{\pi}{p}$ for $\eta \in (0,1)$ and having Fourier coefficients such that $||\{ \alpha_k \}||_{l^{2,s}} = \varepsilon << 1$ for some $s > \frac{1}{2}$. Let $u(x,t)$ be a solution from equation \eqref{eq:NLSwithMSeveralCorners} with initial data $u(0)$ like \eqref{GralICNLSSeveralCorners}, obtained in theorem \eqref{thm:SeveralDiracMassesNLS}. Then, for all $t_{p,q} = \frac{p}{2\pi q}$ with $1 \leq q \leq q_{\varepsilon}$ odd and for all $x$ at distance larger than $\eta/q$ from $\frac{1}{q} \mathbb{Z}$, the function $u(x,t)$ almost vanishes in the sense:
		\begin{align*}
			|u(t_{p,q},x)| \leq \frac{C}{q^{\gamma - \frac{1}{2}}} || \{ \alpha_k \} ||_{l^{2,s}} + C\sqrt{q} \log (q)||\{\alpha_k\}||_{l^{2,s}}^3,
		\end{align*}
		where $q_{\varepsilon}$ is the largest $q$ such that $\varepsilon^2 \sqrt{q} \log (q) << 1$; so in particular $q_{\varepsilon} \rightarrow \infty$ for $\varepsilon \rightarrow 0$.
		\\
		Nevertheless, the solution can concentrate near $\frac{1}{q} \mathbb{Z}$ in the sense that there is a family of initial data $u_0^{\lambda} = \sum_{k \in \mathbb{Z}} \alpha_k^{\lambda} \delta_k$ and $C>0$ such that
		\begin{align*}
			\left|  \frac{u^{\lambda}(t_{p,q},0)}{e^{i t_{p,q} \Delta} \alpha_0^{\lambda} \delta_0(0)} \right| \leq \frac{C}{q} \lambda + \frac{C}{q^{\gamma}} \lambda^{s+\frac{1}{2}} + C \log (q) \lambda^{(\frac{3}{2}+s) - \varepsilon'} \rightarrow \infty, \quad \lambda \rightarrow \infty,
		\end{align*}
		for $\varepsilon'>0$ arbitrarily small.
		\label{prop:TalbotNonLinear}
	\end{proposition}   
    \begin{figure}[h]
\centering
\includegraphics[width=6.5cm, height=10cm]{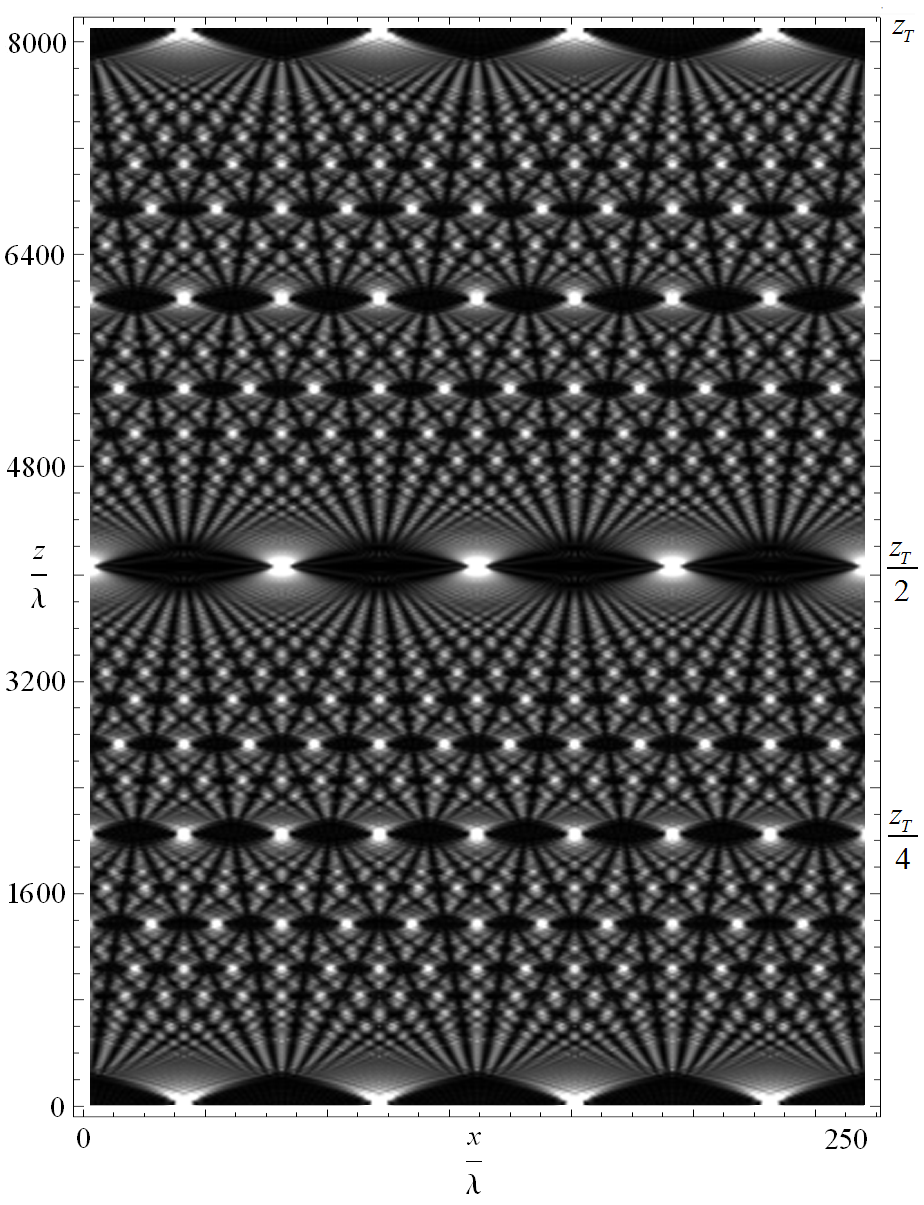}
\caption{Graphical representation of Talbot effect, where $z$ denotes the distance from the light source (and $z_T$ denotes such a distance, for a given period $T$), $x$ denotes the horizontal position, and $\lambda$ denotes the wavelength. (Source: \url{https://en.wikipedia.org/wiki/Talbot_effect})
}
\end{figure}
	\subsection{Evolution of $M$-regular polygonal lines through the BF}\label{sec:Evolution of $M$-regular polygonal lines through the BF}
	We assume now that  the solutions of \eqref{eq:BSintegral} at $t=0$ are given by a regular polygon with $M$ sides (for more details in this issue, see for instance \cite{DLHV14} and \cite{V15}). We use again Hasimoto formalism, so that a regular polygon can be written as
	\begin{align}
		\psi(s,0) = \frac{2\pi}{M} \sum_{k \in \mathbb{Z}} \delta \big (s - \frac{2\pi k}{M} \big).
		\label{IDMRegularPolygon}
	\end{align}
	Observe that $\psi(s,0)$ from \eqref{IDMRegularPolygon} satisfies
	\begin{align*}
		e^{2 \pi i j Ms} \psi(s,0) = \psi(s,0), \quad \forall j \in \mathbb{Z}.
	\end{align*}
	From analogous computations done above regarding the invariance under Galilean transformation \eqref{galileantransf}, if there is uniqueness for the initial value problem \eqref{eq:NLSwithMSeveralCorners}, then we should have that 
	\begin{align*}
		\tilde{\psi}_k(s,t) = e^{iks - ik^2t} \psi(s - 2kt,t) = \psi(s,t), \quad \forall k,t,s \in \mathbb{R}.
	\end{align*}
    Consider now the Fourier coefficient
    \begin{align*}
        \hat{\psi}(j,t) = \frac{M}{2\pi} \int_0^{2\pi/M} e^{-iMjs} \psi(s,t)ds, 
    \end{align*}
    so that
    \begin{align*}
        &\hat{\psi}(j,t) =  \frac{M}{2\pi} \int_0^{2\pi/M} e^{-iMjs} [e^{iMks - i(Mk)^2t }\psi(s - 2Mkt,t)]ds \nonumber\\
        &= \frac{M e^{-i(Mk)^2t}}{2\pi} \int_0^{2\pi/M} e^{-iM(j-k)s}\psi(s-2Mkt,t)ds \nonumber\\
        &=  \frac{M e^{-i(Mk)^2t}}{2\pi} \int_0^{2\pi/M} e^{-iM(j-k)(s+2Mkt)} \psi(s,t)ds \nonumber\\
        &= \frac{M e^{-i(Mk)^2t - iM(j-k)2Mkt}}{2\pi}  \int_0^{2\pi/M} e^{-iM(j-k)s} \psi(s,t) ds \nonumber\\
        &= e^{-i(Mk)^2t - iM(j-k)2Mkt} \hat{\psi}(j-k,t).
    \end{align*}
    Thus, if we take $j=k$,
    \begin{align*}
        \hat{\psi}(j,t) = e^{-i(Mj)^2t} \hat{\psi}(0,t), \quad \forall j \in \mathbb{Z}.
    \end{align*}
Consider now `rational times'
\begin{align*}
    t_{p,q} = \frac{2\pi p}{(M^2 q)}
\end{align*}
so that
\begin{align*}
    &\psi(s,t_{p,q}) = \hat{\psi}(0,t_{p,q}) \sum_{k \in \mathbb{Z}} e^{-i(Mk)^2 2 \pi p/(M^2q) + iMks} \nonumber\\
    &= \hat{\psi}(0,t_{p,q}) \sum_{k \in \mathbb{Z}} e^{-2 \pi i(p/q)k^2 + iMks} \nonumber\\
    &=\hat{\psi}(0,t_{p,q}) \sum_{l=0}^{q-1} \sum_{k \in \mathbb{Z}} e^{-2 \pi i(p/q)(qk+l)^2 + iM(qk+l)s} \nonumber\\
    &= \hat{\psi}(0,t_{p,q}) \sum_{l=0}^{q-1} e^{-2 \pi i(p/q)l^2 + iMls}  \sum_{k \in \mathbb{Z}} e^{iMqks},
\end{align*}
    which, in particular, 
    displays the Talbot effect (see \cite{KP99}, \cite{ET13}). This concrete Talbot effect can be geometrically interpreted by observing that at time $t_{p,q}$ a new polygon of $Mq$ sides is generated. This is determined by the generalized quadratic Gauss sum, defined by
    \begin{align*}
        G(a,b,c) = \sum_{l=0}^{|c|-1} e^{2 \pi i(al^2 + bl)/c},
    \end{align*}
    given $a,b,c\in \mathbb{Z}$, $c \neq 0$. In our case, $a=-p$, $b = m = \frac{Mqs}{2\pi}$, $c = q$, so that 
    \begin{align}
        G(-p,m,q) = \begin{cases}
        & \begin{array}{cc}
        \sqrt{q} e^{i \theta_m} & \text{if } q \text{ odd} \\
        \sqrt{2q} e^{i \theta_m} & \text{if } q \text{ even and } q/2 = m \text{ mod }2 \\
          0 & \text{if } q \text{ even and } q/2 \neq m \text{ mod }2
    \end{array}
\end{cases}
\label{eq:GaussSum_pqm}
\end{align}
    with angle $\theta_m$ depending on $m$, $p$ and $q$. Gathering all the information, we get that 
    \begin{align*}
        \psi(s,t_{p,q}) = \begin{cases}
        & \begin{array}{cc}
        \sum_{m=0}^{q-1}(\alpha_m + i \beta_m) \delta\left( s - \frac{2\pi m}{Mq}\right) & \text{if } q \text{ odd} \\
        \sum_{m=0}^{q/2-1}(\alpha_{2m+1} + i \beta_{2m+1}) \delta\left( s - \frac{4 \pi m + 2\pi}{Mq}\right)  & \text{if } q/2 \text{ odd} \\
          \sum_{m=0}^{q/2-1}(\alpha_{2m} + i \beta_{2m}) \delta\left( s - \frac{4 \pi m}{Mq}\right)  & \text{if } q/2 \text{ even}
    \end{array}
\end{cases}
\end{align*}
    where
    \begin{align*}
        |\alpha_m + i \beta_m| = \begin{cases}
        & \begin{array}{cc}
        \frac{2\pi}{M\sqrt{q}} \hat{\psi}(0,t_{p,q}) & \text{if } q \text{ odd} \\
        \frac{2\pi}{M\sqrt{q/2}} \hat{\psi}(0,t_{p,q})   & \text{if } q \text{ even and } q/2 = m \text{ mod }2 \\
         0  & \text{if } q \text{ even and } q/2 \neq m \text{ mod }2,
    \end{array}
\end{cases}
\end{align*}
from which we deduce that the angle $\rho$ between two adjacent sides is constant. Moreover, if we write
\begin{align*}
    \alpha_m + i \beta_m = \rho e^{i \theta_m},
\end{align*}
we see that the resulting polygon is determined by the angles $\theta_m$.  The question about how to determine $\rho$ is rather delicate.
It was conjectured in \cite{DLHV14}  that
\begin{align*}
    \cos(\rho) =  \begin{cases}
        & \begin{array}{cc}
        2 \cos^{2/q}(\pi/M) - 1 & \text{if } q \text{ odd} \\
        2 \cos^{4/q}(\pi/M) - 1 & \text{if } q \text{ even}.
    \end{array}
\end{cases}
\end{align*}
This was formally justified in \cite{BV20} (section 4), see also \cite{BV23} for more details. All in all, we have enough information to compute tangent vector $T$ at $t = t_{p,q}$, except the trajectory of one point. The latter is obtained using the symmetries of the $M$-polygon as in \cite{DLHV14}. See figure \eqref{fig:tangent} for a comparison between numerical and theoretical solutions. The construction of the former has been just sketched.

\begin{figure}[h]
\centering
\includegraphics[scale=0.1, width=\textwidth]{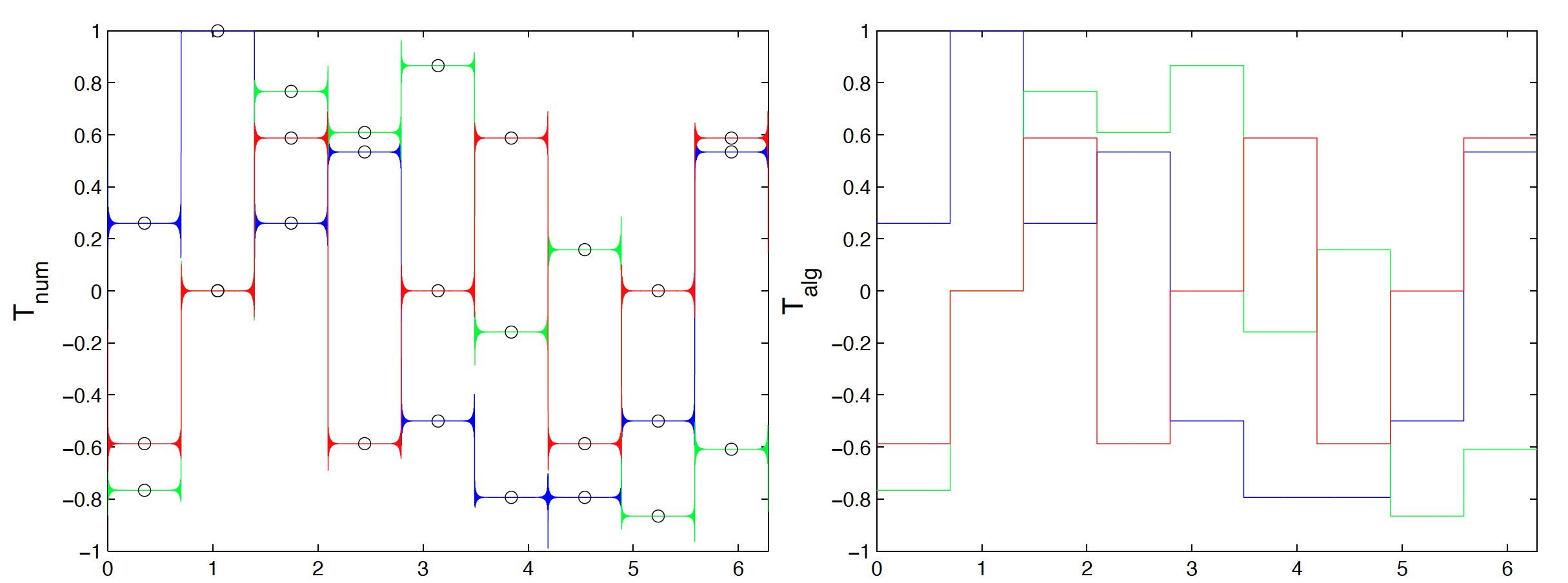}
\caption{$T_{num}$ versus $T_{alg}$, for M = 3, at $t_{1,3} = 2\pi/27$. $T_1$ appears in blue, $T_2$ in
green, $T_3$ in red. In $T_{num}$, the Gibbs phenomenon is clearly visible. The black circles
denote the points chosen for the comparisons. Figure extracted from \cite{DLHV14}.}
\label{fig:tangent}
\end{figure}

It is interesting to study, on the one hand, the evolution of $\chi(0,t)$ itself, and on the other hand the trajectory of one of the corners of $\chi(0,t)$. Again through the symmetries of the $M$-polygon, it is easy to see that the trajectory lies in a plane. See \Cref{fig:polygonM15} for sumerical simulation describing the evolution of a corner of a polygon with $M=15$ sides, and \Cref{fig:InkPointM4,fig:RiemannInkPoint} for simulations with different values of  $M$.

\begin{figure}[h]
\centering
\includegraphics[scale=0.1, width=\textwidth]{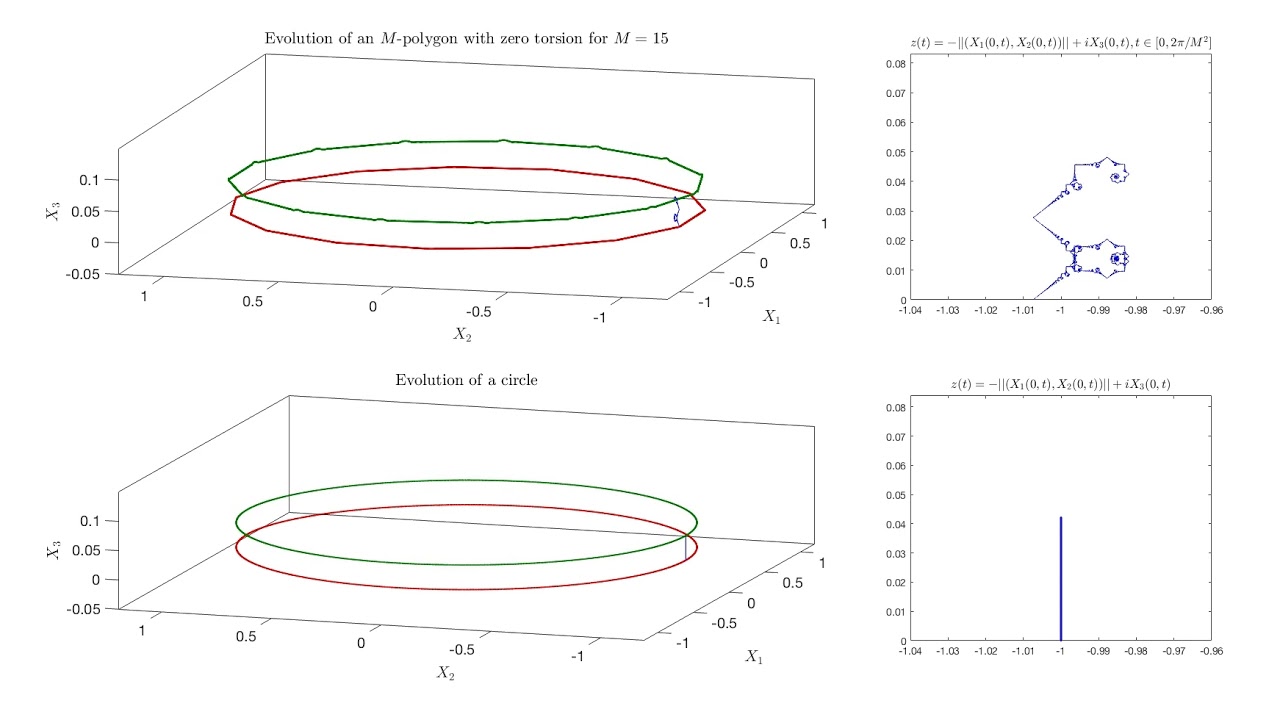}
\caption{Simulation of $\chi(0,t)$ from a regular polygon of $M=15$ sides. Screenshots obtained from Sandeep Kumar's website \url{https://sites.google.com/view/skumar1712/simulation-videos}}
\label{fig:polygonM15}
\end{figure}

\begin{figure}[h]
\centering
\includegraphics[width=8cm, height=6cm]{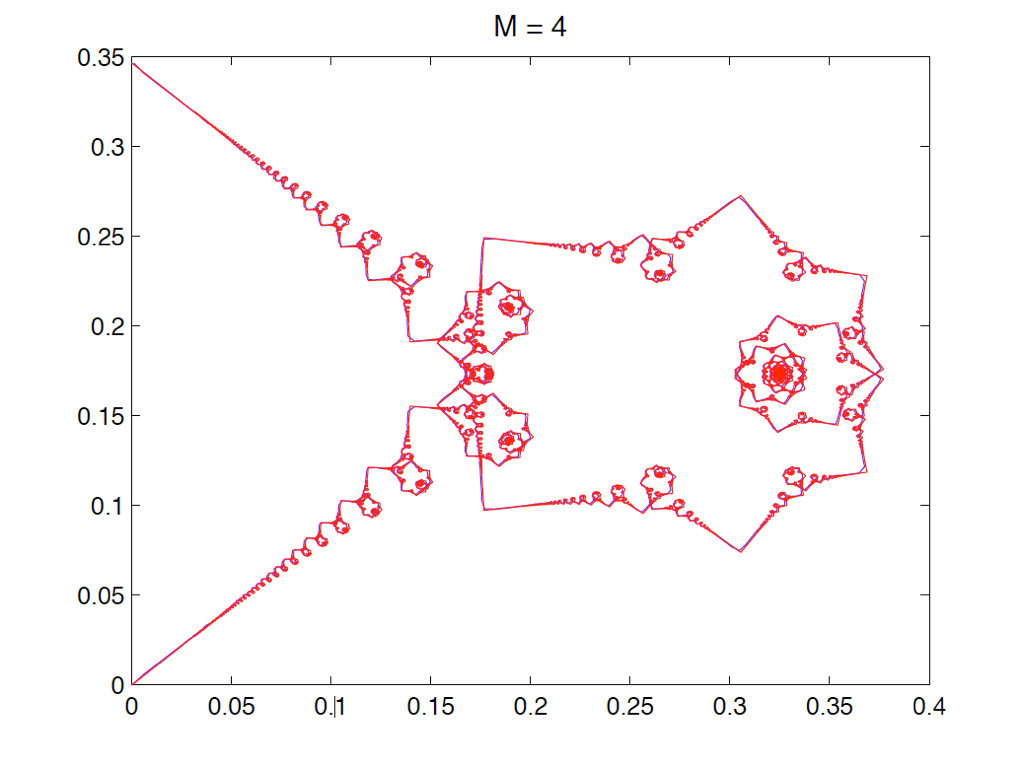}
\caption{Simulation of one of the corners of the curve $\chi(0,t)$, for a regular polygon of $M=4$ sides.}
\label{fig:InkPointM4}
\end{figure}

\noindent
These pictures are quite reminiscent of the graph of the so-called Riemann's non-differentiable function
\begin{align}
    \phi(t) = \sum_{k \neq 0} \frac{e^{-2 \pi i t k^2}}{k^2}.
    \label{RiemannNDfunction}
\end{align}
Observe that $\phi(0) = \phi(1)$. If we renormalize properly $\chi(0,t)$, we see in figures \eqref{fig:RiemannInkPoint} that at the qualitative level there is a big similitude between them. Furthermore, it becomes stronger the bigger the number of sides $M$. This is numerically proved in \cite{DLHV14}. In \cite{J96} it is proved that the graph of function \eqref{RiemannNDfunction} is a fractal satisfying the so-called Frisch-Parisi conjecture (see \cite{FP85}). 

\noindent
We provide now some extra remarks about the behavior of real fluids. Results from \cite{DLHV14}, briefly exposed here, suggest that at half of a period the starting regular polygon with $M$ sides becomes the same polygon but with the axis switched with an angle $\pi/M$. For such a switch, take the Gauss sum from \eqref{eq:GaussSum_pqm} with $q=2$, and observe that it is zero for half of the cases, so that the resulting polygon, at that time, has $M$ sides instead of $2M$. Therefore, the axis switching phenomenon follows from the analysis of which are the values of $m$ that makes \eqref{eq:GaussSum_pqm} trivial. This phenomena has been observed in real fluids, in particular for non-circular jets, and it is well documented (see \cite{DLHV14} for appropriate references).
To conclude, we mention the domestic experiment, made by the lecturer and collaborators. A smoke cannon is made with a cardboard box that has a hole with the shape of an equilateral triangle. A camera is located in front of the box. This experiment shows the flip-flop phenomenon for the case $M=3$, switching the axis by an angle of $\frac \pi6$.  We refer to \url{https://www.ehu.eus/luisvega/data/uploads/sapporo-agosto-2014.pdf} for the figures. Of course, the corners are smoothed out due to the effects of the viscosity. The analysis presented here also implies that skew polygons with six sides should appear as a consequence of the Talbot effect mentioned above. However, this is not conclusive and further evidence in this direction is needed.

\end{document}